\documentclass[a4paper,USenglish,numberwithinsect]{lipics-v2019}

\usepackage{xargs}

\usepackage[inline]{enumitem}
\usepackage[dvipsnames, table]{xcolor}
\newcommand{\Nat}{I\!\!N}

\newcommand{\Rs}{\overrightarrow{r}}
\newcommand{\Rp}{\overleftarrow{r}}
\newcommand{\B}{B}
\newcommand{\Q}{Q}
\newcommand{\A}{A}
\newcommand{\C}{C}
\newcommand{\CP}{\widetilde{P}}

\def\til#1{\widetilde{#1}}

\usepackage{tikz}
\usepackage{tikz-qtree}
\definecolor{lightgray}{gray}{0.9}

\usepackage[linesnumbered,algoruled,boxed,lined]{algorithm2e}
\usepackage{amsmath}

\def\idtt#1{\ensuremath{\mathtt{#1}}}

\def\op#1{\idtt{#1}}

\usetikzlibrary{arrows,automata,arrows.meta,decorations,
	decorations.pathreplacing,calc,bending,positioning,chains,patterns,
	decorations.pathmorphing, snakes}
\usetikzlibrary{decorations.markings}
\usetikzlibrary{backgrounds, scopes}
\usetikzlibrary{decorations.shapes}
\usetikzlibrary{shapes.geometric}
\tikzset{paint/.style={ draw=#1!50!black, fill=#1!50 },
	decorate with/.style=
	{decorate,decoration={shape backgrounds,shape=#1,shape size=2mm}}}

\usepackage[colorinlistoftodos,prependcaption,textsize=tiny]{todonotes}
\newcommandx{\unsure}[2][1=]{\todo[linecolor=red,backgroundcolor=red!25,bordercolor=red,#1]{#2}}
\newcommandx{\change}[2][1=]{\todo[linecolor=blue,backgroundcolor=blue!25,bordercolor=blue,#1]{#2}}
\newcommandx{\info}[2][1=]{\todo[linecolor=OliveGreen,backgroundcolor=OliveGreen!25,bordercolor=OliveGreen,#1]{#2}}
\newcommandx{\improvement}[2][1=]{\todo[linecolor=Plum,backgroundcolor=Plum!25,bordercolor=Plum,#1]{#2}}
\newcommandx{\thiswillnotshow}[2][1=]{\todo[disable,#1]{#2}}

\nolinenumbers
\hideLIPIcs

\title{Space-Efficient Vertex Separators for Treewidth}
\titlerunning{Space-Efficient Vertex Separators for Treewidth}

\author{Frank Kammer}{THM, University of Applied Sciences Mittelhessen,
Giessen, 
Germany}{frank.kammer@mni.thm.de}{https://orcid.org/0000-0002-2662-3471}{}

\author{Johannes Meintrup}{THM, University of Applied Sciences Mittelhessen, 
	Giessen, 
	Germany}{johannes.meintrup@mni.thm.de}{https://orcid.org/0000-0003-4001-1153}{}

\author{Andrej Sajenko}
{THM, University of Applied Sciences Mittelhessen, Giessen, Germany}
{andrej.sajenko@mni.thm.de}{https://orcid.org/0000-0001-5946-8087}
{Funded by the Deutsche Forschungsgemeinschaft (DFG, German Research Foundation) -- 379157101.}

\authorrunning{F. Kammer, J. Meintrup and A. Sajenko}

\Copyright{Frank Kammer, Johannes Meintrup and Andrej Sajenko}

\ccsdesc[500]{Mathematics of computing~Graph algorithms}

\keywords{FPT, tree decomposition, network flow, subgraph stack}%mandatory

\acknowledgements{
We thank the anonymous reviewers for their careful reading of our manuscript and their 
helpful comments.}%optional

\EventEditors{}
\EventNoEds{0}
\EventLongTitle{}
\EventShortTitle{}
\EventAcronym{}
\EventYear{}
\EventDate{}
\EventLocation{}
\EventLogo{}
\SeriesVolume{}
\ArticleNo{ }
\begin{document}

\maketitle

\begin{abstract}
For $n$-vertex graphs with treewidth $k = O(n^{1/2-\epsilon})$ and 
an arbitrary $\epsilon>0$, we present a word-RAM algorithm to compute 
vertex separators using only $O(n)$ bits of working memory.
As an application of our algorithm, 
we give an $O(1)$-approximation algorithm for tree decomposition.
Our algorithm computes a tree decomposition
in $c^k n (\log \log n) \log^* n$  time using
$O(n)$ bits for some constant $c > 0$.

We finally use the tree decomposition 
obtained by our algorithm
to solve \textsc{Vertex Cover}, 
\textsc{Independent Set},
\textsc{Dominating Set},
\textsc{MaxCut} and $q$-\textsc{Coloring}
by using $O(n)$ bits as long as the treewidth of the graph is smaller 
than $c' \log n$ for some problem dependent constant $0 < c' < 1$.
\end{abstract}

\section{Introduction}
For solving problems in the context of the ever-growing field of big data
we require algorithms and data structures that do not only focus on runtime 
efficiency,
but consider space as an expensive and  valuable resource. Some
reasons for saving memory are that slower memory in the memory 
hierarchy 
has to be used, less 
cache faults
arise and the available memory allows us to run more parallel tasks on a
given problem.

As a solution researchers began to provide {\em space-efficient} 
algorithms and data structures to solve basic problems like graph-connectivity 
problems
~\cite{AsaIKKOOSTU14, BanerjeeC0S18, ChoGS18, ElmHK15, Hag18}%
, other graph problems~\cite{HagKL19,KamS20}, memory 
initialization~\cite{HagK17, 
KatG17}, dictionaries with constant-time operations~\cite{BroM99, Clark98,
Hag18cd} or graph interfaces~\cite{BarAHM12,HagKL19} 
space efficiently, i.e., they 
designed practical algorithms and data-structures that run (almost) as fast 
as 
standard solutions for the problem under consideration while using 
asymptotically 
less space.
Several space-efficient algorithms and data structures mentioned above are implemented in an open source GitHub project~\cite{KamS18g}.

Our model of computation is the word RAM where we assume to have the
standard operations to read, write as well as arithmetic operations
(addition, subtraction, multiplication and bit shift) take constant time on
a word of size $\Theta(\log n)$ bits where $n\in \Nat$ is the size of the
input.
To measure the total amount of memory that an algorithm requires we
distinguish between the {\em input memory}, i.e., the read-only memory that
stores the input, and the {\em working memory}, i.e., the read-write memory
an algorithm additionally occupies during the computation.
In contrast, the \emph{streaming model} \cite{MunP80} allows us to
access the input only in a purely
sequential fashion, and the main goal is
to minimize the number of passes over the input.
There are several algorithms
for NP-hard problems on a streaming model, e.g., 
a kernelization algorithm due to Fafianie and Kratsch~\cite{FafianieK14}.
Elberfeld, 
Jakoby, and Tantau~\cite{ElberfeldJT10} 
proved that the problem of 
obtaining a so-called tree decomposition is contained in LSPACE.

We continue the research 
of Banerjee, Chakraborty, Raman, Roy and Saurabh~\cite{BanCRRS15} on space-efficient algorithm 
on the word RAM for
 NP-hard problems, but they assumed that a tree decomposition is already given
with the input.

One approach to solve NP-hard graph problems
is to  
decompose the given graph into a tree decomposition consisting of a tree where 
each node of the tree has a bag containing
a small fraction of the original vertices of the graph. %
The quality of a tree decomposition is measured by its {\em width}, i.e.,
the number of vertices of the largest bag minus~1.
The {\em treewidth} of a graph is the smallest width over all tree decompositions
for the graph.
Having a tree decomposition $(T, B)$ of a graph, the problem  
is solved 
by first determining the solution size of the problem in a bottom-up
traversal on $T$ and second in a top-down
process computing a solution for the whole graph.
Treewidth is also a topic in current interdisciplinary research,
such as smart contracts using cryptocurrency~\cite{chatterjee2019treewidth}
or computational quantum physics~\cite{DBLP:journals/corr/abs-1807-04599},
which are fields that often work with big data sets. So it is
important to have~space efficient algorithms.
 
Several algorithms are known for computing a tree 
decomposition.
For the following, we assume that the given graphs have $n$ vertices and 
treewidth~$k$.
Based on ideas in \cite{RobS86}, Reed~\cite{Reed92} showed an algorithm for computing a tree decomposition of width $O(k)$ in time $O(c^kn \log n)$ for some constant~$c$.
His algorithm repeatedly uses a so-called
	balanced separator that splits the input graph
into roughly equally sized subgraphs, each used as an input
for a recursive call of the algorithm.
%\journal{
Further tree-decomposition algorithms can be found in~\cite{Amir10, Bod96, 
Bodl95, Feige05, Lagr96}. 
%}%
Recently, Izumi and Otachi~\cite{IzuO20} presented an algorithm that runs with $o(n)$
bits, but $\Omega(n^2)$ time on graphs with treewidth $k=O(n^{1/2-\epsilon})$
for an arbitrary $\epsilon>0$.
The basic
strategy of repeated separator searches is the foundation of all
treewidth approximation algorithms, as mentioned 
by Bodlaender et al.~\cite{BodDDFLP16}.
%For a detailed description of these algorithms see~\cite{Amir10, Bod96, Bodl95, Feige05, Lagr96}.
Using the same strategy, Bodlaender et al.~also %
presented an algorithm that runs in $2^{O(k)}n$ time and finds a tree decomposition having width~$5k + 4$.

To obtain a space-efficient approximation algorithm for treewidth we modify Reed's algorithm.
We finally use a hybrid approach, which combines our new algorithm and 
Bodlaender
et al.'s 
algorithm~\cite{BodDDFLP16} to find a tree decomposition 
in $b^k n(\log^*n)\log \log n$ time for some
constant~$b$.
{The general idea for the hybrid approach 
is to
use our space-efficient algorithm for treewidth only for constructing the nodes
of height at most $b^k \log \log n$}. For the subgraph induces by the bags of the vertices
 below a node of height $b^k \log \log n$, we use 
Bodlaender
et al.'s algorithm.
The most computationally difficult task of this paper is the computation of the 
separators with $O(n)$ bits.

Finding separators requires finding vertex-disjoint paths for which  
%running 
a DFS as a subroutine is usually used. %needed. 
All recent 
space-efficient DFS require $\Omega(n)$
bits~\cite{AsaIKKOOSTU14, ChoGS18, ElmHK15, Hag18}
or $\Omega(n^2)$ time due to $\Omega(n)$ separator searches~\cite{IzuO20}. 
Thus, 
our challenge is to compute a separator and subsequently a 
tree decomposition with $O(n)$ bits with a running time almost linear
in~$n$.

To compute a separator of size at most $k$ with $O(n \log k)$ bits, 
the 
idea
is to store up to $k$ vertex disjoint paths by assigning a color $c\in
\{0,\ldots, k\}$ to each
vertex $v$ such that we know to which path $v$ belongs. We also
number the vertices along a path with $1,2,3,1,2,3,$ etc.\ so that we 
know
the direction of the path. 
Since we want to find separators with only $O(n)$ bits, we further show
that it suffices to store the color information only at every 
$\Theta(\log 
k)$th vertex.
We so manage to 
 find 
separators of size at most $k$ with $O(n + k^2(\log k) \log n)$ bits. 
If $k=O(n^{1/2-\epsilon})$
for an arbitrary $\epsilon>0$, we thus use $O(n)$ bits.

Our solution to find a separator is in particular interesting because 
previous space-efficient graph-traversal algorithms either reduce the space 
from 
$O(n \log n)$ 
bits to $O(n)$, e.g., 
 depth first search (DFS) or 
breath first search (BFS)~\cite{ChoGS18, ElmHK15, Hag18}, or reduce the 
space from $O(m \log n)$ to $O(m)$, e.g., Euler partition~\cite{HagKL19} and
cut-vertices~\cite{KamKL19}. In contrast, we reduce the space for the 
separator search
from $O((n+m)\log n)$ bits to $O(n)$ bits for small treewidth~$k$.

Besides the separator search, many algorithms for treewidth store large subgraphs of the $n$-vertex, $m$-edge input graph during recursive calls, i.e., they require $\Omega((n + m)\log n)$ bits.
We modify and extend the algorithm presented by Reed
with space-efficiency techniques (e.g., store recursive graph instances with 
the so-called subgraph stack~\cite{HagKL19}) to present an iterator
that allows us to
output
the bags of a tree decomposition of width $O(k)$ in an Euler-traversal order 
using $O(kn)$ bits in $c^{k}n \log n \log^*n$ time for some constant 
$c$.
To lower the space bound further, we
use the subgraph stack only to 
store the vertices of the recursive graph instances.
For the edges we present a new problem specific solution.

%This allows us to
%lower the space for storing recursive graph instances to $O(n + k \log^2 n)$ 
%bits.

In Section~\ref{sec:prelim}, we summarize known 
data structures and algorithms that we use afterwards.
Our main result, the computation of $k$-vertex disjoint paths is shown in Section~\ref{sec:nBitVDP}.
We sketch Reed's algorithm in Section~\ref{sec:reedsalg}, where we also show
a space-efficient computation of a balanced vertex separator using $O(n)$ bits. %
In Section~\ref{sec:stream} we present an iterator that
outputs the bags of a tree decomposition using $O(kn)$ bits.
In the following section we lower the space bound to 
$O(n)$ bits for small treewidth, and show our hybrid approach. %
We conclude the paper by showing 
in Section~\ref{sect:appl} %
that our tree decomposition iterator
can be used to solve NP-hard problems like \textsc{Vertex Cover},
\textsc{Independent Set},
\textsc{Dominating Set}, 
\textsc{MaxCut} and $q$-\textsc{Coloring}
with $O(n)$ bits on graphs with small treewidth.
The table below summarizes the space bound of the algorithms 
described in this paper. %

\begin{table}[h!]
	\centering
	\rowcolors{1}{lightgray}{}
	\begin{tabular}{r|c|c|c}
		& standard  & intermediate goal  & final goal \\
		\hline
		$k$ vertex-disjoint paths & $\Omega(kn \log n)$& 
		\multicolumn{2}{c}{$\Theta(n)$ 
		(Sect.~\ref{sec:nBitVDP})} \\
		balanced vertex separator & $\Omega(kn \log n)$ & 
		\multicolumn{2}{c}{$\Theta(n)$ (Sect.~\ref{sec:reedsalg})} \\
		subgraph stack & $\Omega(kn \log n)$ & $\Theta(kn)$ 
		(Sect.~\ref{sec:stream}) & 
		$\Theta(n)$ (Sect.~\ref{subsect:recStack})  \\
		\hline
		\hline
		iterator for a t.d. & $\Omega(kn \log n)$ &
		$\Theta(kn)$ (Sect.~\ref{sec:stream}) & $ \Theta(n)$ 
		(Sect.~\ref{sec:improvementsinspace}) \\
		\hline
		\hline
		NP-hard problems above  & $\Omega(kn \log n)$& 
		\multicolumn{2}{c}{$\Theta(n)$ 
		(Sect.~\ref{sect:appl})} 
	\end{tabular}
	\caption{The space requirements in bits to run the different parts of the algorithm to compute a 
	tree decomposition (t.d.) with applications for an $n$-vertex graph with sufficiently small treewidth $k$.}
	\label{tab:myfirsttable}
\end{table}

\section{Preliminaries}
\label{sec:prelim}
 Let $G=(V,E)$ be an undirected $n$-vertex $m$-edge graph.
If it is helpful, we consider an edge $\{u,v\}$ as two directed edges (called {\em 
arcs}) $(u,v)$ and $(v,u)$.
 As 
 usual for graph algorithms we define $V=\{1, \ldots, n\}$. 
 For every vertex $v \in V$, we access the in- and out-degree of~$v$ through
 functions 
 $\op{deg}^{\rm{in}}_{G},\op{deg}^{\rm{out}}_{G}: V \rightarrow \Nat$ that returns the number of 
incoming and outgoing
 edges of vertex~$v$, respectively. 
 Moreover, $\op{deg}_G=\op{deg}^{\rm{in}}_{G}+\op{deg}^{\rm{out}}_{G}$.
 We denote by $N(v) \subset V$ the neighborhood of~$v$.
 For $A = \{(v, k) \in V \times \Nat\ |\ 1 \le k \le \op{deg}_G(v)\}$,
 let $\op{head}_G: A \rightarrow V$ be a function such that $\op{head}_G(v, k)$ 
 returns the $k$th neighbor of~$v$. Intuitively speaking, each vertex has 
 an adjacency array.
 If space is not a concern, it is customary to store each graph that 
 results from a transformation separately. To save space, we always use the 
 given graph $G$ and store only some 
auxiliary information that helps us to implement the following graph interface 
for 
a transformed graph.
 
 \begin{definition}{(graph interface)}\label{def:graphinterface}
 	A data structure for a graph  $G = (V, E)$ implements the {\em graph 
 	interface (with adjacency arrays)} 
 	exactly if it provides the functions $\op{deg}_G: V \rightarrow \Nat$, 
 	$\op{head}_G: A \rightarrow V$, where $A = \{(v, k) \in (V \times \Nat) \ |\ 
 	1 \le 
 	k \le \op{deg}_G(v)\}$, and gives access to the number $n$ of vertices and 
 	the 
   	number $m$ of edges. 
 \end{definition}
        Let $G = (V, E)$ be a graph and let $V' \subseteq V$.
%Unless stated otherwise, we assume that our input graphs always provide a graph interface with
%a mate function.
During our computation, some 
of our graph interfaces can support 
$\op{head}_G(v, k)$ and $\op{deg}_G(v)$ only for vertices  $v \notin V'$.
For vertices in $V'$, we can access their neighbors via adjacency
lists, i.e., we can use the functions
$\op{adjfirst} : V' \rightarrow P \cup \{\op{null}\}$, $\op{adjhead} : P
\rightarrow V$ and $\op{adjnext} : P  \rightarrow P \cup \{\op{null}\}$ for a 
set of pointers
$P$ to output the neighbors of a vertex $v$ as follows:
$p := \op{adjfirst}(v)$; $\op{while}$ $(p \neq \op{null})$ \{
$\op{print}\ \op{adjhead}(p)$;
$p := \op{adjnext}(p)$;
\}.
We then say that we have a {\em graph interface with $|V'|$ access-restricted 
vertices.} 

In an undirected graph it is common to store an edge at both endpoints,
 hence, every undirected edge $\{u, v\}$ is stored as an {\em arc} (directed 
 edge) $(u, v)$ at the endpoint~$u$ and as an  arc $(v, u)$ at the endpoint $v$.
% Call those two arcs {\em mates} of each other.
% For many algorithms and graph transformations it is useful to access the mate 
% of an arc.
% 
% \begin{definition}{({\em mate function})}\info{Wo wird die function gebracht.}
% 	A graph interface of a graph $G$  supports the mate function if it provides 
% 	the 
% 	function $\op{mate}_G: A \rightarrow A$ that, for a given adjacency entry 
% 	$(v, 
% 	k) \in A$, returns its mate $(u, j) \in A$ such that $\op{head}_G(v, k) = u 
% 	\land 
% 	\op{head}_G(u, j) = v$. Note that $\op{mate}_G(\op{mate}_G(v, k)) = (v, k)$.
%\end{definition}%
We also use the two space-efficient data structures below. 

\begin{definition}{(subgraph stack~\cite[Theorem~4.7]{HagKL19})}
	The {\em subgraph stack} is a data structure initialized for an $n$-vertex
	$m$-edge undirected graph $G_0=(V,E)$ that manages a finite list 
	$(G_0,\ldots,G_\ell)$, called the {\em client list}, of ordered graphs such 
	that 
	$G_i$ is a proper subgraph of $G_{i-1}$ for $0 < i \leq \ell$. 
\end{definition}
	Each graph in the client list implements the graph-access interface
supporting the access operations from above in
	$O(\log^*n)$ time. Additionally, it provides the following operations:
	\begin{itemize}
		\item $\op{push}(B_V, B_A)$: Appends a new graph $G_{\ell+1}=(V_{\ell+1}, E_{\ell+1})$ 
		to the client list. The graph is represented by two bit arrays $B_V$ and $B_A$, with $B_V[i]=1$, $B_A[i]=1$
		exactly if the vertex $i$ (or arc) of $G_\ell$ is still contained in $G_{\ell+1}$. 
		It runs in $O((|B_V| + |B_A|)\log^*n)$.
		\item $\op{pop}$: Removes $G_\ell$ from the client list in constant time.
		\item $\op{toptune}$: Subsequent accesses to the interface of $G_{\ell}$
		can run 
		in $O(1)$ time. 
		Runs in $O((|V_{\ell}| + |E_{\ell}|)\log^*n)$ time and uses $O(n+m)$ bits additional space.
	\end{itemize}
	For the special case where we have a constant $\epsilon$ with $0 < \epsilon < 1$
	such that
        $|V_{i+1}| < \epsilon|V_{i}|$ and $|E_{i+1}| < \epsilon|E_{i}|$ for all $0 \leq i <
	\ell$,
         the entire subgraph stack uses 
	$\sum_{i=0}^{\infty}(\epsilon_{1}^{i}n+\epsilon_{2}^{i}m)=O(n+m)$ bits of space.
	We will use the subgraph stack only in this special case with $\epsilon = 2/3$.

Using fast rank-select data structures~\cite{BauH19} one can easily support the
following data structure. 
\begin{definition}{(static space allocation)}
For $n \in \Nat $, let $V = \{1, \ldots, n\}$ be the universe and $K\subset V$ a fixed 
set of keys with $|K| = O(n / \log b)$ for some $b \le n$. Assume that the
keys are stored within an array of $n$ bits so that a word RAM can read
$K$ in $O(n/\log n)$ time.
Then there is a data structure using $O(n)$ bits that allows us to
read and write $O(\log b)$ bits in constant time 
for each $k\in K$ and the initialization of the data structure can be done in
$O(n/\log n)$ time.
\end{definition}

We now formally define a tree decomposition.

\begin{definition}(tree decomposition~\cite{ROBERTSON199565}, bag)
A {\em tree decomposition} of a graph $G = (V, E)$ is a pair $(T, B)$ where
$T = (W, F)$ is a tree and $B$ is a mapping $W \rightarrow \{V'\ |\ V' 
\subseteq 
V\}$ such that the following properties hold: {\textbf(TD1)} $\bigcup_{w\in W} 
G[B(w)] = G$, and
{\textbf (TD2)} for each vertex $v \in V$, the nodes $w$ with $v \in B(w)$ 
induce 
	a 
	subtree in $T$.
For each node $w \in W$, $B(w)$ is called the {\em bag} of $w$.%
\end{definition}

We recall from the introduction that the {\em width}  of a tree decomposition is defined as the 
number of vertices 
in a largest 
bag minus 1 and the {\em treewidth} of a graph $G$ as the minimum width among all 
possible tree decompositions for $G$.
We subsequently also use the well-known fact that an $n$-vertex
graph $G$ with treewidth $k$ has $O(kn)$ edges~\cite{ROBERTSON199565}.

Our algorithms use space-efficient BFS and DFS.

\begin{theorem}{(BFS~\cite{ElmHK15})}\label{th:se-DFS}
Given an $n$-vertex	$m$-edge graph $G$, there is a BFS on $G$
that runs in $O(n + m)$ time and uses $O(n)$ bits.
\end{theorem}

By taking the randomized DFS of Choudhari et al.~\cite{ChoGS18} and
replacing a randomized dictionary~\cite[Lemma~3]{ChoGS18} by Hagerup's
dictionary~\cite[Corollary~5.3]{Hag18}, we obtain the following.

\begin{theorem}\label{th:dfs}
Given an $n$-vertex	$m$-edge graph $G$, there is a DFS on $G$
that runs in $O(m + n \log^* n)$ time and uses $O(n)$ bits.
\end{theorem}

The DFS of Theorem~\ref{th:dfs} assumes that we provide
a graph interface with adjacency arrays.
Later we run a DFS on an $n$-vertex graph $G$ with treewidth $k$, but can only provide a graph interface with $O(k)$ access-restricted vertices, i.e., $O(k)$ vertices are provided by a list interface.

\begin{lemma}\label{lem:spaceDFS}
Assume that a graph interface for an $n$-vertex, $m$-edge graph $G$ with $O(k)$
access-restricted vertices for some $k \in \Nat$ is given.
Assume further that we can iterate over the adjacency list of an
access-restricted vertex~$v$ in $O(f(n, m))$ time for some function~$f$,
whereas we can access an entry in the adjacency array of another vertex
still in $O(1)$ time.
Then, there is a DFS that runs in $O(k \cdot f(n, m) + m + n\log^* n)$ time
and uses $O(n + k \log n)$ bits.%,
%and a BFS that runs in $O(k \cdot f(n, m) + m + n)$ time
%and uses $O(n)$ bits.
\end{lemma}
\begin{proof}
To get an idea consider first a standard DFS. Basically it starts
from some vertex $v \in V$, visits each vertex of the connected component of $G$,
and iterates over the incident edges of each visited vertex once.
If the DFS visits a vertex $w$, then the current $v$-$w$ path 
is managed in a stack that
stores the status of the current iteration over the incident edges of each vertex on the path,
i.e., for each vertex on the path, it stores the index of the edges of the path
such that the DFS can return
to a vertex and proceed with the next edge to explore another path.
The space-efficient DFS (Theorem~\ref{th:dfs}) 
%stores on a stack tuples of a vertex $v \in V$ and 
%index pointers 
%into the adjacency array of $v$, which is used to store the status of the current
%iteration over the neighbors of $v$. The DFS 
stores only a small part of the entire
DFS stack and uses further stacks with approximations of the
index pointers 
to
reconstruct removed information from the stack. 
The reason for storing only
small parts of the stack as well as 
approximate
indices is 
to bound the space-usage by $O(n)$ bits.
Let $V'$ be the set of access-restricted vertices with $|V'| = k$.
For the access restricted
vertices $V'$ we are not able to use this strategy,
since we only provide a list interface
i.e, we are unable to use indices. 
Thus, for the vertices of $V'$ we simply ignore the space restrictions imposed by the storage scheme of the DFS and instead
store a pointer with $\Theta(\log n)$ bits into the adjacency list for each vertex $v' \in V'$ additionally to the
storage scheme already in place for the DFS, just as any regular DFS working with adjacency lists would need.
This of course uses extra space and thus results in the space-efficient DFS using $O(n+k \log n)$
bits instead of $O(n)$ bits.
Since there is one iteration over the incident edges of each 
of the~$O(k)$ access-restricted vertices,
the running time of the DFS increases to $O(k \cdot f(n, m) + m + n \log^* n)$ time.
%
%\red{A standard BFS explores $G$ round-wise starting at some vertex $v$.
%In each round it manages two sets $Q$ and $Q'$, representing the vertices
%of the current and next round, respectively.
%For each vertex $v \in Q$ it traverses the incident edges of $v$ and puts
%every unvisited neighbor of $v$ into $Q'$.
%A space-efficient BFS realizes $Q$ and $Q'$ as $O(n)$-bit arrays with
%$O(n)$-bit choice dictionaries~\cite{Hag18cd,HagK16,KamS18c} on top of them
%for linear time iteration over members of the sets.
%The BFS visits every vertex once and during the visit iterates over the
%complete neighborhood. 
%Since there is one iteration over the incident edges of each 
%of the~$O(k)$ access-restricted vertices,
%the running time of the BFS increases to $O(k \cdot f(n, m) + m + n)$ time,
%while the space bound of $O(n)$-bits remains.}
%
%without affecting the asymptotic runtime (the pointers we store can be accessed in constant time).
\end{proof}

\section{Finding $k$ Vertex-Disjoint Paths using $O(n + k^2 (\log k) \log n)$ 
	Bits}\label{sec:nBitVDP}

To compute $k$  vertex-disjoint $s$-$t$-paths we modify 
a standard network-flow technique 
where a so-called residual network~\cite{AhuMO93} and a standard DFS is used.

\begin{lemma}\label{lem:kpathsStd}
        (Network-Flow Technique~\cite{AhuMO93})
	Given an integer $k$ as well as an $n$-vertex $m$-edge graph $G = (V, E)$
	with $k$ pairwise internal vertex-disjoint $s$-$t$-paths for some
	vertices $s, t \in V$, there is an algorithm that can compute
	$\ell \le k$ internal pairwise vertex-disjoint $s$-$t$ paths using 
	$O((n + m) \log n)$ bits by executing~$\ell$ times a depth-first search, i.e., in 
	$O(\ell(n + m))$ time.
\end{lemma}

The well-known network-flow technique 
increases the size of an initially empty set of internal vertex-disjoint $s$-$t$-paths one by one in rounds.
It makes usage of 
(1) a so-called residual network for extending a set of edge-disjoint paths by another edge disjoint path
and (2) a simple reduction that allows us to find vertex-disjoint paths by constructing edge-disjoint paths in a modified graph.
We can neither afford storing any copy of a the original graph, a graph transformated 
nor the exact routing of several paths because it would take $\Theta(n \log n)$ bits.
Instead, we use graph interfaces that allow access to a modified graph on the fly
and we partition the paths in parts belonging to small so-called regions that allow us to store 
exact path membership on the boundaries and recompute
the paths in between the boundaries of the regions on the fly if required.
For a sketch of the idea, see Fig.\ref{fig:setting}.
Storing the exact path membership only for the boundaries has the drawback that
we cannot recompute the original paths inside a region,
 but using a deterministic 
algorithm we manage to compute always the same one.
It is important for us to remark that the network-flow technique
above also works even if we change the exact routing
of the $s$-$t$-paths (but not their number)
before or after each round.

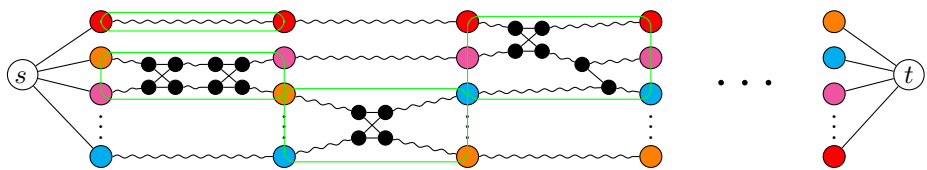
\begin{figure}[h]
	\centering
	\begin{tikzpicture}[
%  -{Stealth[length = 2.5pt]},
start chain = going right,
node distance = 10pt,
vertex/.style={minimum width=0.30em, minimum height=0.30em, circle, draw, on 
chain},
svertex/.style={minimum width=0.15em, minimum height=0.15em, circle, draw, 
on chain, inner sep=1pt, text width=3pt, align=center},
fix/.style={inner sep=1pt, text width=6pt, align=center},
graph/.style={minimum width=3em, minimum height=2em, ellipse, draw, on 
chain, decorate, decoration={snake,segment length=1mm,amplitude=0.2mm}},
decoration={
	markings,
	mark=at position 0.5 with {\arrow{>}}},
path/.style=
	{-,
		decorate,
		decoration={snake,amplitude=.2mm,segment length=2mm,}}
],

\node[vertex, fix] (s) {$s\phantom{t}$};

\node[vertex, fix, fill=red] at ($(s) + (15pt, 20pt)$) (a1) {};
\node[vertex, fix, below = 5pt of a1, fill=orange] (a2) {};
\node[vertex, fix, below = 5pt of a2, fill=Rhodamine] (a3) {};
\node[vertex, fix, below = 15pt of a3, fill=cyan] (a4) {};

\node[vertex, fix, right = 60pt of a1, fill=red] (b1) {};
\node[vertex, fix, right = 60pt of a2, fill=Rhodamine] (b2) {};
\node[vertex, fix, right = 60pt of a3, fill=orange] (b3) {};
\node[vertex, fix, right = 60pt of a4, fill=cyan] (b4) {};

\node[vertex, fix, right = 60pt of b1, fill=red] (c1) {};
\node[vertex, fix, right = 60pt of b2, fill=Rhodamine] (c2) {};
\node[vertex, fix, right = 60pt of b3, fill=cyan] (c3) {};
\node[vertex, fix, right = 60pt of b4, fill=orange] (c4) {};

\node[vertex, fix, right = 60pt of c1, fill=red] (d1) {};
\node[vertex, fix, right = 60pt of c2, fill=Rhodamine] (d2) {};
\node[vertex, fix, right = 60pt of c3, fill=cyan] (d3) {};
\node[vertex, fix, right = 60pt of c4, fill=orange] (d4) {};

\node[vertex, fix, right = 60pt of d1, fill=orange] (e1) {};
\node[vertex, fix, right = 60pt of d2, fill=cyan] (e2) {};
\node[vertex, fix, right = 60pt of d3, fill=Rhodamine] (e3) {};
\node[vertex, fix, right = 60pt of d4, fill=red] (e4) {};

\node[vertex, fix, right = 320pt of s] (t) {$t$};

\node[left = 37pt of t, yshift=-3pt] (d) {\scalebox{2}{$\dots$}};

\node[svertex, fill = black] at ($(a2) + (5pt, -2.5pt)$) (ca2a) {};
\node[svertex, fill = black] at ($(a2) + (15pt, -2.5pt)$) (ca2b) {};

\node[svertex, fill = black] at ($(a3) + (5pt, 2.5pt)$) (ca3a) {};
\node[svertex, fill = black] at ($(a3) + (15pt, 2.5pt)$) (ca3b) {};

\node[svertex, fill = black] at ($(a2) + (30pt, -2.5pt)$) (ca2a2) {};
\node[svertex, fill = black] at ($(a2) + (40pt, -2.5pt)$) (ca2b2) {};

\node[svertex, fill = black] at ($(a3) + (30pt, 2.5pt)$) (ca3a2) {};
\node[svertex, fill = black] at ($(a3) + (40pt, 2.5pt)$) (ca3b2) {};

\node[svertex, fill = black] at ($(b3) + (15pt, -7pt)$) (cb3a) {};
\node[svertex, fill = black] at ($(b3) + (25pt, -7pt)$) (cb3b) {};

\node[svertex, fill = black] at ($(b4) + (15pt, 7pt)$) (cb4a) {};
\node[svertex, fill = black] at ($(b4) + (25pt, 7pt)$) (cb4b) {};

\node[svertex, fill = black] at ($(c1) + (5pt, -2.5pt)$) (cc1a) {};
\node[svertex, fill = black] at ($(c1) + (15pt, -2.5pt)$) (cc1b) {};

\node[svertex, fill = black] at ($(c2) + (5pt, 2.5pt)$) (cc2a) {};
\node[svertex, fill = black] at ($(c2) + (15pt, 2.5pt)$) (cc2b) {};

\node[svertex, fill = black] at ($(c2) + (30pt, -2.5pt)$) (cc2a2) {};
%\node[svertex, fill = black] at ($(c2) + (40pt, -2.5pt)$) (cc2b2) {};

%\node[svertex, fill = black] at ($(c3) + (30pt, 2.5pt)$) (cc3a) {};
\node[svertex, fill = black] at ($(c3) + (40pt, 2.5pt)$) (cc3b) {};

%\draw [decorate,decoration={brace,amplitude=5pt, mirror},xshift=-4pt,yshift=0pt]
%($(a4.south) - (5pt, 0)$) -- ($(a4.south) + (5pt, 0)$) node [black,midway, yshift = 
%-15pt] (n) {$O(k)$};
%\node[below = -5pt of n] (n1) {border vertices};

%\draw [decorate,decoration={brace,amplitude=5pt, mirror},xshift=-4pt,yshift=0pt]
%($(b4.south) - (5pt, 0)$) -- ($(c4.south) + (5pt, 0)$) node [black,midway, yshift = 
%-15pt] (q) {region of $O(k \log k)$};
%\node[below = -5pt of q] (q1) {vertices};

\node[below = -6pt of a3](ad) {$\vdots$};
\node[below = -6pt of b3](ad) {$\vdots$};
\node[below = -6pt of c3](ad) {$\vdots$};
\node[below = -6pt of d3](ad) {$\vdots$};
\node[below = -6pt of e3](ad) {$\vdots$};

\path [draw, path] (a1) -- (b1);
\path [draw, path] (a2) -- (ca2a);
\path [draw, path] (a3) -- (ca3a);
\path [draw, path] (a4) -- (b4);
\path [draw, path] (ca2b2) -- (b2);
\path [draw, path] (ca3b2) -- (b3);
\path [draw, -] (ca2a) -- (ca2b);
\path [draw, -] (ca3a) -- (ca3b);
\path [draw, path] (ca2b) -- (ca2a2);
\path [draw, path] (ca3b) -- (ca3a2);
\path [draw, -] (ca2a2) -- (ca2b2);
\path [draw, -] (ca3a2) -- (ca3b2);
\path [draw, -] (ca2a) -- (ca3b);
\path [draw, -] (ca3a) -- (ca2b);
\path [draw, -] (ca2a2) -- (ca3b2);
\path [draw, -] (ca3a2) -- (ca2b2);

\path [draw, path] (b1) -- (c1);
\path [draw, path] (b2) -- (c2);
\path [draw, path] (b3) -- (cb3a);
\path [draw, path] (b4) -- (cb4a);
\path [draw, path] (cb3b) -- (c3);
\path [draw, path] (cb4b) -- (c4);

\path [draw, -] (cb3a) -- (cb3b);
\path [draw, -] (cb3b) -- (cb4a);
\path [draw, -] (cb3a) -- (cb4b);
\path [draw, -] (cb4a) -- (cb4b);

\path [draw, -] (cc1a) -- (cc1b);
\path [draw, -] (cc1b) -- (cc2a);
\path [draw, -] (cc1a) -- (cc2b);
\path [draw, -] (cc2a) -- (cc2b);

\path [draw, path] (c1) -- (cc1a);
\path [draw, path] (c2) -- (cc2a);
\path [draw, path] (c3) -- (cc3b);
\path [draw, path] (c4) -- (d4);
\path [draw, path] (cc1b) -- (d1);
\path [draw, path] (cc2b) -- (cc2a2);
\path [draw, path] (cc2a2) -- (d2);
%\path [draw, -] (cc2a2) -- (cc2b2);
\path [draw, path] (cc3b) -- (d3);
\path [draw, -] (cc2a2) -- (cc3b);

\path[draw, -] (s) -- (a1);
\path[draw, -] (s) -- (a2);
\path[draw, -] (s) -- (a3);
\path[draw, -] (s) -- (a4);

\path[draw, -] (t) -- (e1);
\path[draw, -] (t) -- (e2);
\path[draw, -] (t) -- (e3);
\path[draw, -] (t) -- (e4);

\draw[rounded corners, color=green] ($(a2) + (0, 2pt)$) rectangle ($(b3) + (0, -2pt)$) {};
\draw[rounded corners, color=green] ($(c1) + (0, 2pt)$) rectangle ($(d3) + (0, -2pt)$) {};
\draw[rounded corners, color=green] ($(a1) + (0, 3.5pt)$) rectangle ($(b1) + (0, -3.5pt)$) {};
\draw[rounded corners, color=green] ($(b3) + (0, 2pt)$) rectangle ($(c4) + (0, -2pt)$) {};

%\path[draw] (1) -- ($(1) - (3pt, 20pt)$);

\end{tikzpicture}
        \vspace{-2mm}
	\caption{Vertex-disjoint $s$-$t$-paths with colored boundaries. % vertices
%	and black cross edges. 
        Some regions are exemplary sketched by green rectangles.}\label{fig:setting}
\end{figure}

For the rest of this section, let $G = (V, E)$ be an $n$-vertex $m$-edge graph with treewidth~$k \in \Nat$
such that~$G$ has~$k$ internal pairwise vertex-disjoint $s$-$t$-paths for some vertices $s, t \in V$ and let $\mathcal P$
be a set of $\ell \le k$ pairwise internally vertex-disjoint $s$-$t$-paths.
Let $V'$ be the set that consists of the union of all vertices over all paths in~$\mathcal P$.

Before we can present our algorithm, we present two lemmas that consider properties of internal pairwise 
vertex-disjoint paths as given in~$\mathcal P$. 
To describe the properties we first define a special
path structure (with respect to~$G$). % for the computation of our storage scheme.
A {\em deadlock cycle} in $G$ with respect to $\mathcal P$ consists
of a sequence $P_1, \ldots, P_{\ell'}$ of paths in $\mathcal P$ for some $2 \le \ell' \le \ell$
such that there is 
a subpath $x_i, \ldots, y_i$ on every path $P_i$
($1 \le i \le \ell'$) and there is an edge $\{x_i, y_{(i \mod
\ell') + 1}\} \in E$.
We call a deadlock cycle {\em simple} if
every subpath consists of exactly two vertices, and otherwise {\em extended} (see Fig.~\ref{fig:deadlockexample}).
Moreover, we call an $s$-$t$-path $P$ {\em chordless}
if only subsequent vertices of the path are connected by an edge called chord.
In some sense, a chord of a path is a deadlock cycle over
only one path.

\begin{figure}[h]
			\centering
			\begin{subfigure}{.5\textwidth}
				\centering
				\begin{tikzpicture}[
%  -{Stealth[length = 2.5pt]},
%start chain = going down,
node distance = 10pt,
vertex/.style={minimum width=0.30em, minimum height=0.30em, circle, draw, on 
chain},
svertex/.style={minimum width=0.15em, minimum height=0.15em, circle, draw, 
on chain, inner sep=1pt, text width=3pt, align=center},
fix/.style={inner sep=1pt, text width=6pt, align=center},
graph/.style={minimum width=3em, minimum height=2em, ellipse, draw, on 
chain, decorate, decoration={snake,segment length=1mm,amplitude=0.2mm}},
decoration={
	markings,
	mark=at position 0.5 with {\arrow{>}}},
path/.style=
	{-,
		decorate,
		decoration={snake,amplitude=.2mm,segment length=2mm,}}
],

\node[] (p) {$P_{\ell}$};
\node[below= 0pt of p] (pl) {$P_{\ell'}$};
\node[below= 0pt of pl] (p2) {$P_2$};
\node[below= 0pt of p2] (p1) {$P_1$};

\node[svertex, fill=black, right= 55pt of p1] (x1) {};
\node[svertex, fill=black, right= 70pt of p1] (y1) {};

\node[svertex, fill=black, right= 40pt of p2] (x2) {};
\node[svertex, fill=black, right= 55pt of p2] (y2) {};

\node[svertex, fill=black, right= 55pt of pl] (xl) {};
\node[svertex, fill=black, right= 70pt of pl] (yl) {};

\path[draw, path]
($(p.east) + (6pt,0)$) -- ($(p) + (170pt, 0pt)$)
($(pl.east) + (5pt,0)$) -- (xl.center)
(yl.center) -- ($(pl) + (170pt, 0)$)
($(p2.east) + (5pt,0)$) -- (x2.center)
(y2.center) -- ($(p2) + (170pt, 0)$)
($(p1.east) + (5pt,0)$) -- (x1.center)
(y1.center) -- ($(p1) + (170pt, 0)$)
;

\path[draw]
(x1) -- (y1)
(x2) -- (y2)
(xl) -- (yl)
(x1) -- (y2)
(x2) -- (yl)
(xl) -- (y1)
;

%\path[draw, dotted, color=red, line width = 1pt]
%(pl1) -- (pl2)
%(p31) -- (p32)
%(p32) -- (p33)
%(p21) -- (p22)
%(p11) -- (p12)
%(w) -- (p11)
%(u) -- (pl2)
%;

%\path[draw] 
%(pl1) edge [-] (p33)
%(p31) edge [-] (p22)
%(p21) edge [-] (p12)
%;

\node[] at ($(p2) + (15pt, 11pt)$) {$\vdots$};
\node[] at ($(p2) + (170pt, 11pt)$) {$\vdots$};

\node[] at ($(pl) + (15pt, 11pt)$) {$\vdots$};
\node[] at ($(pl) + (170pt, 11pt)$) {$\vdots$};

\end{tikzpicture}
 			\end{subfigure}%
			\begin{subfigure}{.5\textwidth}
				\centering
				\begin{tikzpicture}[
%  -{Stealth[length = 2.5pt]},
%start chain = going down,
node distance = 10pt,
vertex/.style={minimum width=0.30em, minimum height=0.30em, circle, draw, on 
chain},
svertex/.style={minimum width=0.15em, minimum height=0.15em, circle, draw, 
on chain, inner sep=1pt, text width=3pt, align=center},
fix/.style={inner sep=1pt, text width=6pt, align=center},
graph/.style={minimum width=3em, minimum height=2em, ellipse, draw, on 
chain, decorate, decoration={snake,segment length=1mm,amplitude=0.2mm}},
decoration={
	markings,
	mark=at position 0.5 with {\arrow{>}}},
path/.style=
	{-,
		decorate,
		decoration={snake,amplitude=.2mm,segment length=2mm,}}
],

\node[] (p) {$P_{\ell}$};
\node[below= 0pt of p] (pl) {$P_{\ell'}$};
\node[below= 0pt of pl] (p2) {$P_2$};
\node[below= 0pt of p2] (p1) {$P_1$};

\node[svertex, fill=black, right= 55pt of p1] (x1) {};
\node[svertex, fill=black, right= 70pt of p1] (y1) {};

\node[svertex, fill=black, right= 25pt of p2] (x2) {};
\node[svertex, fill=red, right= 40pt of p2] (d) {};
\node[svertex, fill=black, right= 55pt of p2] (y2) {};

\node[svertex, fill=black, right= 55pt of pl] (xl) {};
\node[svertex, fill=black, right= 70pt of pl] (yl) {};

\path[draw, path]
($(p.east) + (6pt,0)$) -- ($(p) + (170pt, 0pt)$)
($(pl.east) + (5pt,0)$) -- (xl.center)
(yl.center) -- ($(pl) + (170pt, 0)$)
($(p2.east) + (5pt,0)$) -- (x2.center)
(y2.center) -- ($(p2) + (170pt, 0)$)
($(p1.east) + (5pt,0)$) -- (x1.center)
(y1.center) -- ($(p1) + (170pt, 0)$)
;

\path[draw]
(x1) -- (y1)
(xl) -- (yl)
(x1) -- (y2)
(x2) -- (yl)
(xl) -- (y1)
;

\path[draw, path, color=red]
(x2) -- (d)
(d) -- (y2)
;

%\path[draw, dotted, color=red, line width = 1pt]
%(pl1) -- (pl2)
%(p31) -- (p32)
%(p32) -- (p33)
%(p21) -- (p22)
%(p11) -- (p12)
%(w) -- (p11)
%(u) -- (pl2)
%;

%\path[draw] 
%(pl1) edge [-] (p33)
%(p31) edge [-] (p22)
%(p21) edge [-] (p12)
%;

\node[] at ($(p2) + (15pt, 11pt)$) {$\vdots$};
\node[] at ($(p2) + (170pt, 11pt)$) {$\vdots$};

\node[] at ($(pl) + (15pt, 11pt)$) {$\vdots$};
\node[] at ($(pl) + (170pt, 11pt)$) {$\vdots$};

\end{tikzpicture}
			\end{subfigure}%
			\caption{Illustration of a simple and extended deadlock cycle (left and right, respectively).}\label{fig:deadlockexample}
\end{figure}
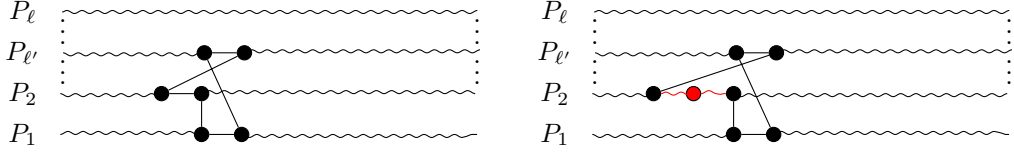

\begin{lemma}\label{lem:sameVertexSet}
Assume that the paths in $\mathcal P$ are chordless 
and that $G$ has no extended deadlock cycle with respect to $\mathcal P$.
Any other set $\mathcal Q$ of $\ell$
internal
pairwise vertex-disjoint $s$-$t$-paths in $G[V']$ uses all vertices of $V'$.
\end{lemma}
\begin{proof}
Note that each path in $\mathcal P$ uses another neighbor of $s$ and since
the paths are chordless, $s$ can not have another neighbor in $V'$. In
other words, $s$ has exactly $\ell$ neighbors in $V'$ and each path in $\mathcal
Q$ must use exactly one of these neighbors. This allows us to name the paths
in $\mathcal P$ and in $\mathcal Q$ with $P_1,\ldots,P_\ell$ and $Q_1,\ldots,Q_\ell$, respectively, such that the
second vertex of $P_i$ and $Q_i$ are the same for all $i\in \{1,\ldots,\ell\}$.

The following
description assumes suitable indices of the paths in $\mathcal P$.
Assume for a contradiction that the lemma is not true.
This means that 
$P_1$ uses a vertex $v$ whereas $Q_1$ leaves the vertices of $P_1$ at some
vertex $u$ to avoid using a vertex $v\in V'$.
Since the paths are
chordless, $Q_1$ must jump on a vertex of a path $P_2$. Since the paths are
chordless, this means that $Q_2$
must leave to a path $P_3$. At the end, there must be a path $Q_i$ ($i \ne 1$) that
jumps to a vertex $w$ on $P_1$. If $w$ is before $v$
(while walking on
$P_1$ from $s$ to $t$), then we rename the indices of the paths in $Q$ such that the paths
with the same index
in $\mathcal P$ and in $\mathcal Q$ use the same vertex strictly after the
jumps in $\mathcal Q$. Afterwards, we can reapply the paragraph. 
If $w$ is behind $v$, we have an extended deadlock
cycle in $G$ with respect to $\mathcal P$; a contradiction.
\end{proof}

\begin{lemma}\label{lem:probMaintain}
Assume that the paths in $\mathcal P$ are chordless 
and that $G$ has no extended deadlock cycle with respect to $\mathcal P$.
Any other set $\mathcal Q$ of $\ell$ internal
pairwise vertex-disjoint $s$-$t$-paths in $G[V']$ is chordless and  
$G$ has no extended deadlock cycle with respect to $\mathcal
Q$.
\end{lemma}
\begin{proof}
For a contradiction assume that $\mathcal Q$ has a path $Q$ with a chord $\{u, w\}$
such that some vertex $v$ is in between $u$ and $w$ on the path $Q$. 
This allows us to construct a path $Q'$ from $Q$ that uses 
the edge $\{u, w\}$ instead of the subpath of $Q$ between $u$ and $w$, leaving $v$ unused.
We so get $\ell$ paths $\{Q'\} \cup (\mathcal Q \setminus \{Q\})$ that does not use $v \in V'$, which is a contradiction to Lemma~\ref{lem:sameVertexSet}.

Assume now that there is an extended deadlock
cycle $Z$ in $G$ with respect to $\mathcal Q$ and some vertex $v$ on $Z$ such that $v$ is in the middle of
one of the subpaths of $Z$. 
By removing the common edges 
of $Z$ and the paths in $\mathcal Q$ from the paths in $\mathcal Q$ 
and adding the remaining edges of $Z$ to the paths in $\mathcal Q$, 
we get a set of pairwise vertex-disjoint $s$-$t$-paths $\mathcal Q'$ 
in $G[V' \setminus \{v\}]$.
An example is sketched in Fig.~\ref{fig:constQprime}.
% for some vertex $v$ on $Z$ such that $v$ is in the middle of
%one of the subpaths of $Z$. 
The existence of $\mathcal Q'$
is again a contraction to Lemma~\ref{lem:sameVertexSet}.
\end{proof}
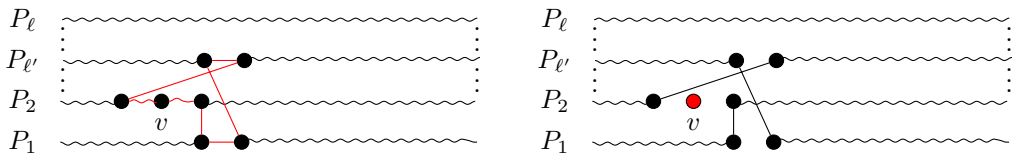
\begin{figure}[h]
			\centering
			\begin{subfigure}{.5\textwidth}
				\centering
				\begin{tikzpicture}[
%  -{Stealth[length = 2.5pt]},
%start chain = going down,
node distance = 10pt,
vertex/.style={minimum width=0.30em, minimum height=0.30em, circle, draw, on 
chain},
svertex/.style={minimum width=0.15em, minimum height=0.15em, circle, draw, 
on chain, inner sep=1pt, text width=3pt, align=center},
fix/.style={inner sep=1pt, text width=6pt, align=center},
graph/.style={minimum width=3em, minimum height=2em, ellipse, draw, on 
chain, decorate, decoration={snake,segment length=1mm,amplitude=0.2mm}},
decoration={
	markings,
	mark=at position 0.5 with {\arrow{>}}},
path/.style=
	{-,
		decorate,
		decoration={snake,amplitude=.2mm,segment length=2mm,}}
],

\node[] (p) {$P_{\ell}$};
\node[below= 0pt of p] (pl) {$P_{\ell'}$};
\node[below= 0pt of pl] (p2) {$P_2$};
\node[below= 0pt of p2] (p1) {$P_1$};

\node[svertex, fill=black, right= 55pt of p1] (x1) {};
\node[svertex, fill=black, right= 70pt of p1] (y1) {};

\node[svertex, fill=black, right= 25pt of p2] (x2) {};
\node[svertex, fill=black, right= 40pt of p2,label=below:$v$] (d) {};
\node[svertex, fill=black, right= 55pt of p2] (y2) {};

\node[svertex, fill=black, right= 55pt of pl] (xl) {};
\node[svertex, fill=black, right= 70pt of pl] (yl) {};

\path[draw, path]
($(p.east) + (6pt,0)$) -- ($(p) + (170pt, 0pt)$)
($(pl.east) + (5pt,0)$) -- (xl.center)
(yl.center) -- ($(pl) + (170pt, 0)$)
($(p2.east) + (5pt,0)$) -- (x2.center)
(y2.center) -- ($(p2) + (170pt, 0)$)
($(p1.east) + (5pt,0)$) -- (x1.center)
(y1.center) -- ($(p1) + (170pt, 0)$)
;

\path[draw, color=red]
(x1) -- (y1)
(xl) -- (yl)
(x1) -- (y2)
(x2) -- (yl)
(xl) -- (y1)
;

\path[draw, path, color=red]
(x2) -- (d)
(d) -- (y2)
;

%\path[draw, dotted, color=red, line width = 1pt]
%(pl1) -- (pl2)
%(p31) -- (p32)
%(p32) -- (p33)
%(p21) -- (p22)
%(p11) -- (p12)
%(w) -- (p11)
%(u) -- (pl2)
%;

%\path[draw] 
%(pl1) edge [-] (p33)
%(p31) edge [-] (p22)
%(p21) edge [-] (p12)
%;

\node[] at ($(p2) + (15pt, 11pt)$) {$\vdots$};
\node[] at ($(p2) + (170pt, 11pt)$) {$\vdots$};

\node[] at ($(pl) + (15pt, 11pt)$) {$\vdots$};
\node[] at ($(pl) + (170pt, 11pt)$) {$\vdots$};

\end{tikzpicture}
 			\end{subfigure}%
			\begin{subfigure}{.5\textwidth}
				\centering
				\begin{tikzpicture}[
%  -{Stealth[length = 2.5pt]},
%start chain = going down,
node distance = 10pt,
vertex/.style={minimum width=0.30em, minimum height=0.30em, circle, draw, on 
chain},
svertex/.style={minimum width=0.15em, minimum height=0.15em, circle, draw, 
on chain, inner sep=1pt, text width=3pt, align=center},
fix/.style={inner sep=1pt, text width=6pt, align=center},
graph/.style={minimum width=3em, minimum height=2em, ellipse, draw, on 
chain, decorate, decoration={snake,segment length=1mm,amplitude=0.2mm}},
decoration={
	markings,
	mark=at position 0.5 with {\arrow{>}}},
path/.style=
	{-,
		decorate,
		decoration={snake,amplitude=.2mm,segment length=2mm,}}
],

\node[] (p) {$P_{\ell}$};
\node[below= 0pt of p] (pl) {$P_{\ell'}$};
\node[below= 0pt of pl] (p2) {$P_2$};
\node[below= 0pt of p2] (p1) {$P_1$};

\node[svertex, fill=black, right= 55pt of p1] (x1) {};
\node[svertex, fill=black, right= 70pt of p1] (y1) {};

\node[svertex, fill=black, right= 25pt of p2] (x2) {};
\node[svertex, fill=red, right= 40pt of p2,label=below:$v$,] (d) {};
\node[svertex, fill=black, right= 55pt of p2] (y2) {};

\node[svertex, fill=black, right= 55pt of pl] (xl) {};
\node[svertex, fill=black, right= 70pt of pl] (yl) {};

\path[draw, path]
($(p.east) + (6pt,0)$) -- ($(p) + (170pt, 0pt)$)
($(pl.east) + (5pt,0)$) -- (xl.center)
(yl.center) -- ($(pl) + (170pt, 0)$)
($(p2.east) + (5pt,0)$) -- (x2.center)
(y2.center) -- ($(p2) + (170pt, 0)$)
($(p1.east) + (5pt,0)$) -- (x1.center)
(y1.center) -- ($(p1) + (170pt, 0)$)
;

\path[draw]
%(x1) -- (y1)
%(xl) -- (yl)
%
(x1) -- (y2)
(x2) -- (yl)
(xl) -- (y1)
;

%\path[draw, path, color=red]
%(x2) -- (d)
%(d) -- (y2)
%;

%\path[draw, dotted, color=red, line width = 1pt]
%(pl1) -- (pl2)
%(p31) -- (p32)
%(p32) -- (p33)
%(p21) -- (p22)
%(p11) -- (p12)
%(w) -- (p11)
%(u) -- (pl2)
%;

%\path[draw] 
%(pl1) edge [-] (p33)
%(p31) edge [-] (p22)
%(p21) edge [-] (p12)
%;

\node[] at ($(p2) + (15pt, 11pt)$) {$\vdots$};
\node[] at ($(p2) + (170pt, 11pt)$) {$\vdots$};

\node[] at ($(pl) + (15pt, 11pt)$) {$\vdots$};
\node[] at ($(pl) + (170pt, 11pt)$) {$\vdots$};

\end{tikzpicture}
			\end{subfigure}%
			\vspace{-4mm}
			\caption{On the left an extended deadlock cycle $Z$ is shown in red 
			with respect to $\ell'$ paths in $Q$.
			On the right a rerouting is sketched such that the new path $Q'$ do not use vertex $v$ anymore.}\label{fig:constQprime}
\end{figure}

%\change{Diese Information kommt zu frueh! Der leser weiss nocht gar nicht
%das spezielle Pfade brauchen, also kann er diese aussage an dieser Stelle nicht validieren.}

%
We begin with a high-level description of our approach to compute
 $\mathcal{P}$ space efficiently.
Afterwards, we present the missing details.
A chordless
$s$-$t$-path $P$ %, i.e., a path where only subsequent vertices of the path are connected by an edge,
can be computed by a slightly extended DFS and 
to identify the path it suffices to store 
its vertices in an $O(n)$-bit array~$\A$ (Lemma~\ref{lem:chordlessSTPath}). 
However, if we want to store $\ell > 1$ chordless pairwise internally vertex-disjoint $s$-$t$-paths in $\A$,
we cannot distinguish them without further information.
For an easier description assume that 
each path with its vertices is colored with a different color.
Using an array of $n$ fields of $O(\log \ell)$ bits each we can easily store the coloring
of each vertex and thus $\ell$ chordless paths with $O(n \log \ell)$ bits in total. 
However, our goal is to store $\ell$ paths with $O(n + f(\ell) \log n)$ bits for some polynomial function~$f$.
We therefore define a so-called {\em path data scheme} that 
stores partial information about the $\ell$ paths using $O(n)$ bits (Def.~\ref{def:pathDatascheme}).
%
%Our idea is to color a fraction of vertices in $V'$ and use this coloring to 
%deterministically compute (new) paths in between those vertices whenever 
%we are interested in the routing of the paths.
The idea of our scheme is to select and color a set $\B \subseteq V'$ of vertices along the paths with the property that $|\B| = O(n  / \log k)$ and $G[V' \setminus \B]$ consists only of small
connected components of $O(k \log k)$ vertices. We later call these
components {\em regions}.
%\red{We call the set of colored vertices that are directly connected with such a connected component
%its {\em boundary} and the set consisting of the vertices of the connected component 
%and its two boundaries a {\em region}.
%See Fig.~\ref{fig:setting} for a sketch.}

Such a set $\B$ exists if the paths have so-called {\em good} properties (Def.~\ref{def:good-paths}).
Once a path data scheme is created we lose the information about the exact routing of the original paths,
but are able to realize operations to construct the same number of paths.
These operations are summarized
in a so-called {\em path data structure} (Def.~\ref{def:pathDataStructure}). 
In particular, it allows us to
determine the color of a vertex $v$ on a path $P$ or to
run along the paths
by finding
the successor or predecessor of $v$
in $O(\op{deg}(v) + f(k))$ time for some polynomial $f$ (Lemma~\ref{lem:recompute}).
The idea here is to explore the region containing $v$
and use Lemma~\ref{lem:kpathsStd} to get a fixed set of 
$\ell$ paths connecting the equally colored vertices. By
Lemma~\ref{lem:probMaintain}, the new paths still have our good properties.
%\unsure{We need quite the same paths again since our paths have properties, like: they are chordless.}
However, assuming that we have a set $\mathcal P$ of good paths and 
some new chordless path~$P^* \notin \mathcal P$, whose computation is shown in Lemma~\ref{lem:newPath}, $\mathcal P \cup \{P^*\}$
is not necessarily good.
Our approach to make the paths good is to find a rerouting $R(\mathcal P \cup \{P^*\})$ that outputs
a set $\mathcal P'$ of good paths where $|\mathcal P'| = |\mathcal P| + 1$ (Lemma~\ref{lem:rerouting}).
Let $V^*$ be the vertices of $\mathcal P$ and $P^*$.
We realize $R$ by a mapping $\Rs: V^* \rightarrow \{0, \ldots, k\}$ that defines
the successor of a vertex $v$ as a vertex $u \in N(v)$ of another path with color
$\Rs(v)$.
Intuitively speaking, we cut the paths in $\mathcal P \cup \{P^*\}$ in pieces and reconnect them as defined in $\Rs$.

The rerouting may consist of $O(n)$ successor information, each of $O(\log k)$ bits.
Therefore, we 
compute a rerouting by considering only the $s$-$p$-subpath
 $P'$ 
of $P^*$ for some vertex $p$. 
We compute
the rerouting in $O(\log k)$ batches by moving $p$ closer to $t$ with each
batch. Each batch consists of $\Theta(n / \log k)$ successor information (or less for the last batch).
Using the rerouting $\Rs$ and
 path data structures storing $\mathcal P$ and $P^*$, we realize a so-called {\em weak path data structure}
for the paths $R(P \cup \{P'\}).$ 
We call it a weak path data structure because it can not be used to determine the color of a vertex.
After that we use it to compute a new valid path data scheme 
for the paths $R(P \cup \{P'\})$ 
such that we need neither the rerouting~$\Rs$ nor~$V^*$ anymore (Lemma~\ref{lem:storageScheme}).
However, since $P'$ is only a subpath of $P^*$ we repeat the whole process described in this paragraph with another
subpath $P''$ of $P^*$, i.e., the subpath of $P^*$ whose first vertex is adjacent to the last vertex of $P'$.
Finally, we so get a valid path data structure storing 
$|\mathcal P| + 1$ good $s$-$t$ paths (Corollary~\ref{cor:bathwise-inclusion}).

We now describe our approach in detail.
To store a single $s$-$t$-path we number the internal vertices along the path with $1, 2, 3, 1, 2, 3,$ etc. and store the numbers of all vertices inside an array $\A$ of $2n$ bits.
For a vertex $v$ outside the paths, we set $\A[v] = 0$. 
We refer to this technique as {\em path numbering}.
By the numbering, we know the direction from $s$ to $t$ even if we hit a vertex in the middle of the path.
Note that $\A$ completely defines $V'$.

To follow a path stored in $\A$, begin at a vertex $v$ and look for a neighbor $w$ of $v$ with $\A[w] = (\A[v] \mod 3) + 1$.
Note that this approach only works if~$v$ has only one such neighbor.
To avoid ambiguities when following a path we require that the path is chordless.

\begin{lemma}\label{lem:chordlessSTPath}
	Assume that we are given access to a DFS that uses a stack to store a path from $s$ to the current vertex and that runs in $f_{\rm{t}}(n, m)$ time using $f_{\rm{b}}(n, m)$ bits. 
	Then, there is an algorithm to compute a 
	path numbering for a
	chordless $s$-$t$-path~$P$ in time $O((n + m) + f_{\rm{t }}(n,m))$ by using $O(n + f_{\rm{b}}(n,m))$ bits.
\end{lemma}
\begin{proof}
	We call the internal subpath $P' \subseteq P$ from
	a vertex $v_i$ to a vertex $v_j$
	{\em skippable} exactly if $P'$ consists of at least~$3$ vertices and there exists a {\em chord}, i.e., 
	an edge from $v_i$ to $v_j$. 
	The idea is to construct first an $s$-$t$-path $P$ in $G$ 
	with the DFS and, directly after finding $t$,
	we stop the DFS and then pop the vertices from 
	the DFS stack and remove skippable subpaths until we arrive at $s$.
	
	Let $P = (s = v_1, v_2, \ldots, v_x = t)$ 
	be the vertices on the DFS stack after reaching $t$
	and let $M$ be an $n$-bit array marking all vertices with a one 
	that are currently on the stack.
	(We say that a vertex $v$ is {\em marked in $M$} exactly if $M[v] = 1$.)
	The algorithm moves backwards from~$t$ to~$s$, i.e., by 
	first popping $t$, setting $v_{i + 1} = t$ and then
	popping the remaining vertices from the DFS stack as follows.
	
	Pop the next vertex $v_i$ from the stack, unmark $v_i$ in $A$,
	and let $v_{i-1}$ be the vertex that is now on top of the DFS stack.
	We check $v_i$'s neighborhood $N(v_i)$ in $G$ for all vertices $u \notin \{v_{i-1}, v_{i+1}\}$ that are 
	marked in $M$, mark them in an $n$-bit array $M'$ as chords 
	and remember their number in $x$.
    The internal subpaths between $v_i$ and all vertices marked in $M'$ are skippable.
    We remove the skippable paths as follows:
    As long as $x \ne 0$, we pop a vertex $u$ from the DFS stack.
    If $u$ is marked in $M'$, we unmark $u$ in $M'$ and reduce $x$ by one.
    If $x \ge 1$, we unmark $u$ in $M$.
    When $x = 0$, the edge $\{u, v_i\}$ is the last chord, 
    we set $v_i := u$ and $v_{i + 1} := v_i$
    and continue our algorithm as described in this paragraph until $v_i = s$.
	
    After removing skippable paths from $P$, we can output the path as follows.
    Allocate an array $\A$ with $2$ bits for each vertex and initialize all entries with~$0$.
	Let $i = 1$.
	Starting from $u = s$ go to its only neighbor $w \ne u$ that is marked
	in $M$.
	Set $i: = (i \mod 3) + 1$ and $\A[w] := i$. 
	Unmark $u$ in $M$, set $u := w$, and continue assigning the path numbering until $w = t$.
	Finally return $\A$.
		
		Observe that the DFS runs once to find a path $P$ from $s$ to $t$
		and then we only pop the vertices from the stack. 
		During the execution of the DFS we manage
		membership of vertices using only arrays of $O(n)$ bits.
		The construction of $\A$ uses a single run over the path
		by exploring the neighborhood of each marked vertex.
		In total our algorithm runs in $O(n + m)$
		time and 
		uses $O(n)$ bits in addition to the time and space needed by the given DFS.
\end{proof}

Recall that the path numbering 
does not uniquely define $\ell > 1$ chordless vertex-disjoint 
paths 
since there can be edges, called {\em cross edges},
between vertices of two different paths in~$\mathcal P$.
Due to the cross edges, the next vertex of a path can be ambiguous (see again Fig.~\ref{fig:setting}).
Coloring these vertices would solve the problem, but we may have to many to store their coloring.

Our idea is to select a set $\B \subseteq V'$ of vertices that we call {\em boundary
vertices}
such that $|\B| = O(n / \log k)$ holds and by removing $\B$ from $G[V']$
we get a graph $G[V' \setminus \B]$ that consists of connected components
of size $O(k \log k)$.
In particular, $\B$ contains all vertices with {\em large degree}, i.e., vertices with a degree
strictly greater than $k \log k$.
Note that a graph with treewidth $k$ has at most $kn$ edges,
and so we have only $O(n / \log k)$ vertices of large degree.
%
%\red{The reason for having all vertices or large degree is that 
%
% Due to performance reasons we want 
%%that the later shown recoloring
%%of the paths in between the boundary vertices depends on $k$ and not on $n$.
%%To achieve that we avoid the recoloring for boundary vertices
%%since they may connect several connected components, and 
%to avoid iterating over the neighborhood of
%vertices with {\em large degree}, i.e., vertices with a degree}
%strictly greater than $k \log k$.
%Let $L \subseteq V'$ be the set of vertices with large degree.
%
We store the color of all boundary vertices $\B$. % and 
%all vertices $L$ of $V'$ with large degree.
This allows us to answer their color directly.
A {\em region} is a connected component in $G[V' \setminus \B]$.
Note that a boundary vertex can connect several regions. 
To avoid exploring several regions when searching 
%To find 
the predecessor/successor on the same path
of a vertex $v\in \B$, %quickly without exploring several regions, %---which is of interest especially if 
%$v$ is large---
we additionally
store the color of the predecessor/successor $w\notin\{s,t\}$
of each $v \in \B$. % unique predecessor and unique successor of each vertices in $\B \cup L$.
%
%
%
%In detail, our goal is to color $O(n/\log k)$ vertices
%such that the following condition holds:
%by removing all colored 
%vertices in $G[V']$ we get a graph $G[V' \setminus \B]$ consisting of $O(n / (k \log k))$ 
%components called {\em regions}, each of size $O(k\log k)$.
%Let us assume that the regions are numbered such that we visit Region $1$, Region $2$, Region $3$, etc.
%in this order while moving over a path from $s$ to $t$.
%We call the set of colored vertices that partitions the graph as described
%above (a set of) {\em border vertices} with respect to~$\mathcal P$.
%Moreover, the $k$ border vertices that disconnect Region $i \in \Nat$ and Region $i + 1$
%are called a {\em boundary}.
%
According to our described setting %and the condition 
above, we formally define our scheme.

\begin{definition}{(Path Data Scheme)}\label{def:pathDatascheme}
	A {\em path data scheme} for $\mathcal P$ in $G$ is a triple $(\A, \B, \C)$ where 
	\begin{itemize}
		\item $\A$ is an array storing the path numbering of all paths in $\mathcal P$,
		\item $\B$ is a set of all boundary vertices with
		$|\B| = O(n / \log k)$ and $\B$ defines regions of size $O(k \log k)$, and
		\item $\C$ stores the color of every vertex in $\B$ and of each of their
		predecessors and successors. % on the same path,
	\end{itemize}
\end{definition}
We realize $\A$ as well as $\B$ as arrays of $O(n)$ bits,
and $\C$ as a $O(n)$-bit data structure using static space allocation.
Altogether our path data scheme uses $O(n)$ bits.

A further crucial part of our approach is 
to (re-)compute paths fast
(in particular, we do not want to use a $k$-disjoint-path algorithm),
but we also want to guarantee that vertices of the same color get connected with a
path constructed by a deterministic network-flow algorithm.
In other words, 
our path data scheme must store the same color 
for each pair of colored vertices
that are connected by our fixed algorithm of Lemma~\ref{lem:kpathsStd}.
We call a path data scheme that has this property {\em valid}
(with respect to our fixed algorithm).

To motivate the stored information of our path data scheme,
we first show how it can be used to realize a path data structure
that allows us to answer queries on all vertices $V'$ and not
only a fraction of $V'$.
Afterwards, we show the computation of a path data scheme.

\begin{definition}{(Path Data Structure)}\label{def:pathDataStructure}
A path data structure supports the following operations where $v \in V'$. 
	\begin{itemize}
		\item $\op{prev}/\op{next}(v)$:
Returns the predecessor and the successor, respectively, of $v$.
		\item $\op{color}(v)$: Returns the color $i$ of a path $P_i \in \mathcal P$ to
which $v$ belongs to.
	\end{itemize}
\end{definition}

To implement the path data structure, 
our idea is first to explore the region in $G[V' \setminus \B]$ 
containing the given vertex $v$ by
using a BFS and second to construct a graph
of the vertices and edges visited by the BFS.
We partition the visited colored vertices inside the region into two sets
$S'$ (successors of $\B$) and $T'$ (predecessors of $\B$), as well as extend the graph by two vertices $s'$ and $t'$
that are connected to $S'$ and $T'$, respectively.
Then we sort the vertices and its adjacency arrays by vertex {\textsc{id}}s
and run the deterministic network-flow algorithm of Lemma~\ref{lem:kpathsStd}
to construct always the same fixed set of paths (but not necessarily the original paths).
The construction of the paths in a region consisting of a set $U$ of vertices 
is summarized in the next lemma, which we use subsequently to 
support the operations of our path data structure.
We use the next lemma also 
to make a path data scheme valid. We so guarantee that our 
network-flow algorithm 
connects equally colored vertices 
and no $k$-disjoint path algorithm is necessary.
%We so get a valid path data scheme.
%
This is the reason why the following lemma 
is stated in a generalized manner and does not simply assume that a 
%Note that, given a valid 
path data scheme is given.
%scheme~$(\A, \B, \C)$,
%we easily can call the next lemma.

\begin{lemma}\label{lem:recomputeRegion}
Assume that we are given 
the vertices $U$ of a region in $G$ as well as
two sets $S',T'\subseteq U$ of all
successors and predecessors vertices
of vertices on the boundary, respectively. 
Take $n'=|U|$.
Then $O(n' k^2 \log^2 k)$ time and 
$O(k^2 (\log k) \log n)$ bits suffice to
compute paths connecting each vertex in $S'$ with another vertex in $T'$.
\end{lemma}
\begin{proof}
	We construct a graph $G' = G[U]$, 
	add two new vertices
	$s'$ and $t'$ and connect them with the vertices of $S'$ and $T'$, respectively.
	To structurally get the same graph independent of
	the permutation of the vertices
	in the representation of the set $U$
	 we sort the vertices and the adjacency arrays of the graph representation for $G'$. 
	The details to do that are described in the three subsequent paragraphs.
	Finally we run the algorithm of Lemma~\ref{lem:kpathsStd}
	to compute all $s'$-$t'$-paths 
	in the constructed graph $G'$. Since $S'$ and $T'$ are
    the endpoints of disjoints subpaths of paths in
	$\mathcal P$, we can indeed connect each vertex 
	in~$S'$ with another vertex in~$T'$.

	We now show the construction of $G'$.
	We choose an arbitrary $v \in U$ and run
	a BFS in graph $G[U]$ starting
	at~$v$ three times.
	(We use a standard BFS with the restriction
	that it ignores vertices of
	$G$ that are not in~$U$.)

	In the first run we count the number $n' = O(k \log k)$ 
	of explored vertices.
	Knowing the exact number of vertices and knowing 
	that all explored vertices can have a degree 
	at most $k \log k$ in $G[U]$,
	 we allocate an array $D$ of $n' + 2$ fields
	and, for each $D[i]$ ($i = 1, \ldots, n' + 2$), an
	array of $\lceil k \log k \rceil$ fields,
	each of $\lceil \log (k \log k) \rceil$
	bits.
	We will use $D$ to store the adjacency arrays 
	for $G'$ isomorphic to $G[U]$.
	For reasons of performance, we want to use indirect addressing 
	when operating on $G'$.
	Since the space requirement to realize indirect addressing depends on the largest value in $U$, we 
	use a bidirectional mapping from $U$ to $\{1, \ldots, n'\}$ and use vertex names out of $\{1, \ldots, n'\}$ for $G'$.
	More exactly, to translate the vertex names of $G'$ to the vertex names of $G$, 
	we use a translation table $M: \{1, \ldots, n'\} \rightarrow U$, and we realize the 
	reverse direction $M^{-1}: U \rightarrow \{1, \ldots, n'\}$ by using binary search in $M$, 
        which can be done in $O(\log n')$ time per access.
	For the table we allocate an array $M$ of $n'$ fields, each of $\lceil \log n \rceil$ bits.

	In a second run of the BFS we fill $M$ with the vertices explored by the BFS and sort $M$ by the vertex names.
	In a third run, for each vertex $v$ explored by the BFS, we 
	determine the neighbors $u_1, \ldots, u_x$ of $v$ sorted by
	their $\textsc{id}$s 
	and store $M^{-1}(u_1), \ldots, M^{-1}(u_x)$ in $D[M^{-1}(v)]$.
	During the third BFS run we want to store also two sets $S'$ and $T'$ (using a standard balanced tree representation).
	Now, $D$ allows constant time accesses to a graph $G'$ isomorphic to $G[U]$.
	Afterwards, we are able to compute the paths in $G[U]$ 
	and using the mapping~$M$ translate 
	the vertex $\textsc{id}$s in $G'$ back to
	vertex $\textsc{id}$s in $G$.
	
	{\bfseries Efficiency:}
	We now analyze the space consumption and the runtime of our algorithm.
	Since $G$ has treewidth $k$ 
	and $G'$ consists of vertices of a region
	of $G$, we can follow
	that $G'$ has $n' = O(k \log k)$
	vertices and $O(n' k)$ edges.
	$D$ uses $O(n' \log n)$ bits to store $P(n')$ pointers,
	each of $O(\log n)$ bits, that point at adjacency arrays.
	The adjacency arrays use $O(n'k \log k)$ bits to store
	$O(n'k)$ vertex $\textsc{id}$s out of $\{1, \ldots, k\}$.
	The translation table $M$ and the sets $S'$ and $T'$ 
	use $O(n' \log n)$ bits. 
	The BFS can use $O(n'k \log k)$ bits.
	The space requirement of an in-place sorting algorithm is negligible.
	Lemma~\ref{lem:kpathsStd} runs with 
	$O((n' + kn') \log n') = O(n'k \log k)$ bits.
	In total our algorithm uses 
	$O(n' k \log k) + O(n' \log n) = O(k^2 (\log k) \log n)$ bits.
	
	To compute the graph $G'$ and fill $M$ we run a BFS.
	Running the BFS comes with an extra time 
	factor of $O(\log n')$ to translate a vertex
	$\textsc{id}$s  of $G$ to a vertex 
	$\textsc{id}$s  of $G'$ 
	(i.e., to access values in $M^{-1}$) 
	and costs us $O(n'k \log n')$ time in total. 
    Since $G'$ has $O(n'k)$ edges $D$ has 
        $O(n' k)$ non-zero entries and a sorting of~$D$ can be done in $O(n' k)$ time.
    Sorting $M$ can be done in the same time.
	Finally, Lemma~\ref{lem:kpathsStd} has to be executed once for up to 
        $O(k \log k)$ paths (every vertex
        of the region can be part of $S' \cup T'$), 
        which can be done in 
        $O((n'k) (k \log k) \log k)$ time since
        $O(n'k)$ bounds the edges of the region
        and since we get an extra factor
        of $O(\log k)$ by using the translation table~$M$. 
	In total our algorithms runs in time 
	$O(n' k^2 \log^2 k)$.
\end{proof}

Since the number of vertices in a region is bound by $O(k \log k)$,
we can use Lemma~\ref{lem:recomputeRegion} and 
the coloring $\C$
to support the operations of the path data structure
in the bounds mentioned below.

\begin{lemma}\label{lem:recompute}
	Given a valid path data scheme $(\A, \B, \C)$
	we can realize $\op{prev}$ and $\op{next}$ %of the path data structure 
	in time $O(\op{deg}(v) + k^3 \log^3 k)$ as well as $\op{color}(v)$ in
	$O(k^3 \log^3 k)$ time.
	All operations use $O(k^2 (\log k) \log n)$ bits.
\end{lemma}
\begin{proof}
	In the case that a vertex $v$ is in $\B$, 
        we find its color
	and the color of its predecessor and successor in $\C$.
	(Vertices neighbored to $s$ or $t$ have $s$ as the predecessor and $t$ as the successor, respectively.)
	Thus, we can answer $\op{prev}$ and $\op{next}$ by iterating over $v$'s neighborhood 
	and determining the two neighbors $u, w \in V'$ that are colored the same as $v$.
	By using the numbering $\A$ we know the incoming and 
	outgoing edge of the path through $v$ and so know 
	which of the vertices $u$ and $w$ is the predecessor or successor of $v$.
	
	For a vertex $v \notin \B$, we explore $v$'s
	region in $G[V' \setminus \B]$ by running 
	a BFS in $G[V']$ with the restriction that we do not visit vertices in $\B$.
	We use a balanced heap
	to store the set~$U$ of visited vertices.
	Moreover, we partition all colored vertices of $U$ 
    into the set $S'$ (of successors of $\B$) 
    and the set $T'$ (of predecessors of $\B$)
    with respect to the information in $C$.
    In detail, if a colored vertex $u \in U$
    has an equally colored neighbor $w \in \B$,
    then $u$ is the successor of $w$
    if $\A[v] = (\A[w] \mod 3) + 1$ holds
    and the predecessor of $v$
    if $\A[w] = (\A[v] \mod 3) + 1$ holds.
    
    Having $U$, $S'$ and $T'$ we call 
    Lemma~\ref{lem:recomputeRegion} to
    get the paths in the region.
	Note that a path within a region can be disconnected by vertices
	with large degree so that 
	we can have more than two vertices of each color in $S'$ and $T'$.
	E.g., a path may visit $s_1, t_1, s_2, t_2, s_3, t_3, \ldots \in U$
	with $s_1, s_2, s_3 \in S'$ and $t_1, t_2, t_3 \in T'$.
	Since our path data scheme is
	valid, we can conclude the following:
	assume that we connect the computed subpaths using their common equally
	colored boundary vertices.
	By Lemma~\ref{lem:sameVertexSet} each solution has to use all
	vertices and the network-flow algorithm
	indeed connects $s_1$ with $t_1$,
	$s_2$ with $t_2$, etc.
	By Lemma~\ref{lem:probMaintain} the paths and thus our
	computed subpaths are chordless and have no extended deadlock cycles in $G$
	with respect to~$\mathcal P'$.
    Using the paths we can move along 
    the path of $v$ to a colored vertex
    and so answer $v$'s color
    and both of $v$'s neighbors on the path.
    
    {\bfseries Efficiency:}
	We now analyze the runtime of our algorithm
	and its space consumption.
	Realizing the operations $\op{prev}$
	and $\op{next}$ for
	a vertex of $\B$ can be done in 
	$O(\op{deg}(v))$ time
	 since both colored neighbors have to be found,
	 and their color can be accessed
	in constant time using $\C$.
	The operation $\op{color}$ runs in constant time.
	For the remaining vertices we have
	to explore the region in $G[V' \setminus \B]$.
	Since every region has $n' = O(k \log k)$ vertices
	and also treewith~$k$, the exploration can be done
	in linear time per region, i.e.,
	in $O(n'k)$ time using a standard BFS.
	Filling $U$, $S'$ and $T'$ requires 
	us to add $O(n')$ vertices into
	balanced heaps, which uses 
	$O(n' \log n')$ time.
	The execution of Lemma~$\ref{lem:recomputeRegion}$
	can be done in $O(n' k^2 \log^2 k)$ time
	and running along the paths 
	to find the actual color runs in $O(n')$
	time.
	Thus, in total $\op{prev}$ and $\op{next}$ run
	in time $O(\op{deg}(v) + k^3 \log^3 k)$ and $\op{color}(v)$ in
	time $O(k^3 \log^3 k)$.
	
	A BFS to explore a region in
	$G[V' \setminus \B]$ uses $O(n'k \log n) = O(k^2 (\log k) \log n)$ bits,
	while $U$, $S'$ and $T'$ use $O(n' \log n)$ 
	bits in total.
	The execution of Lemma~\ref{lem:recomputeRegion}
    uses $O(k^2 (\log k) \log n)$ bits, which is also an upper bound 
	for all remaining operations.
\end{proof}

To be able to compute a storage scheme for $\ell > 1$ paths,
$\mathcal P$ must satisfy
certain so-called good properties that we summarize in Definition~\ref{def:good-paths}.
%To describe the properties we consider first a 
%problematic path structure (with respect to~$G$) for the computation of our storage scheme.
%A {\em deadlock cycle} in $G$ with respect to $\mathcal P$ consists
%of a sequence $P_1, \ldots, P_{\ell'}$ of paths in $\mathcal P$ for some $2 \le \ell' \le \ell$
%such that there is 
%a subpath $x_i, \ldots, y_i$ on every path $P_i$
%($1 \le i \le \ell'$) and there is an edge $\{x_i, y_{(i \mod
%\ell') + 1}\} \in E$.
%We call a deadlock cycle {\em simple} if
%every subpath consists of exactly two vertices, and otherwise {\em extended}.

\begin{definition}{(good paths)}\label{def:good-paths}
We call a set $\mathcal P'$ of $s$-$t$-paths {\em good} if and only if
\begin{enumerate}
	\item all paths in $\mathcal P'$ are pairwise internal vertex disjoint,
	\item each path in $\mathcal P'$ is chordless, and
	\item there is no extended deadlock cycle in $G$ with respect to~$\mathcal P'$.
\end{enumerate}
\end{definition}

The next lemma shows the computation of a valid path data scheme for
one path $P$. The computation is straightforward. Since there is only one path,
this path is uniquely defined by the path numbering in $\A$ and we can run
along it and select every $k \lceil \log k \rceil$ vertex as a boundary vertex
and compute and store the remaining information required for a valid path data scheme.
After the lemma we always assume 
that a chordless path $P$ is given together with a valid path data scheme and thus
a path data structure so that we can easily access the predecessor and
successor of a vertex on the path.

\begin{lemma}\label{lem:pathdatastructurefromsinglepath}
Given a path numbering $\A$ for one chordless $s$-$t$ path $P$, we can computes 
a valid path data scheme $(\A, \B, \C)$ for $P$ in $O(n)$ time using $O(n)$ bits.
\end{lemma}
\begin{proof}
Since $\A$ stores a single chordless path $P$, we can run along it from $s$ to $t$ 
while adding every $\lceil k \log k \rceil$th vertex and every vertex of large degree into
an initially empty set $\B$. 
We then can easily determine the predecessor and successor.
Thus, we can store $\C$ using static space allocation.
A path data scheme storing a single path is valid by default since there
is no second path to that an algorithm could switch.
We so get a path storage scheme storing $\ell = 1$ good $s$-$t$-path in $O(n)$ time
using $O(n)$ bits.
\end{proof}

Assume that, for some $\ell \le k$, we have already computed $\ell$ good $s$-$t$-paths, 
which are stored in a valid path data scheme. 
Our approach to compute an $(\ell + 1)$th 
$s$-$t$-path is based on the well-known network-flow technique~\cite{AhuMO93} described in the beginning of the section,
which we now modify to make it space efficient.

By the standard network-flow technique, we do not
search for vertex-disjoint paths in $G$. Instead, we modify $G$ into a directed graph $G'$
and search for edge-disjoint paths in $G'$.
To make the approach space efficient, we do not
store $G'$ separately, but we provide a graph interface that realizes
$G'$. 

\begin{lemma}\label{lem:transform}
Given an $n$-vertex $m$-edge graph $G = (V, E)$ there is a graph interface representing an
directed $n'$-vertex $m'$-edge graph
$G' = (V', E')$ where $V' = \{ v', v''\ |\ v \in V \}$ and $E' = \{ (u'', v'), (v'', u') | \{u, v\} \in E \} \cup \{(v', v'') | v \in
V\}$ 
and thus $n' = 2n$ and $m' = 2m + n$.
The graph interface allows us to access
outgoing edges and incoming edges of $G'$, respectively, by
supporting the operations 
$\op{head}^{\rm{out}}_{G'}$, $\op{deg}^{\rm{out}}_{G'}$
and $\op{head}^{\rm{in}}_{G'}$, $\op{deg}^{\rm{in}}_{G'}$.
The graph interface can be initialized in constant time and the operations have an
overhead of 
constant time and $O(\log n)$ bits.
\end{lemma}
\begin{proof}
	We now show how to compute the operations for $G'$ from $G$ on the fly.
	For each vertex $v$ in $V$, we define two vertices $v' = v$ and $v'' = v + n$ for $V'$.
	By the transformation we can see that every vertex $v'$ (with $v' \le n$) has exactly 
	one outgoing edge to $v'' = v' + n$
	and $\op{deg}_G(v')$ incoming edges from some vertices $u'' = \op{head}_G(v', j) + n$ ($j \le \op{deg}_G(v')$).
	Moreover, every vertex $v''$ (with $v'' > n$) has $\op{deg}_G(v'' - n)$ outgoing edges 
	to some vertices $u' = \op{head}_G(v'' - n, j)$ ($j \le \op{deg}_G(v'' - n)$) and one incoming edge from $v' = v'' - n$.
	With these information we can provide our stated operations for $G'$.
\noindent\begin{minipage}{.5\linewidth}
	\begin{align*}
	\op{deg}^{\rm{out}}_{G'}(v) =
	\begin{cases}
		1 & v \le n\\
		\op{deg}_G(v) & v > n
	\end{cases}
	\end{align*}
\end{minipage}%
\begin{minipage}{.5\linewidth}
	\begin{align*}
	\op{head}^{\rm{out}}_{G'}(v, j) =
	\begin{cases}
		v + n & v \le n\\
		\op{head}_G(v - n, j) & v > n
	\end{cases}
	\end{align*}
\end{minipage}

\noindent\begin{minipage}{.5\linewidth}
	\begin{align*}
	\op{deg}^{\rm{in}}_{G'}(v) =
	\begin{cases}
		\op{deg}_G(v) & v \le n\\
		1 & v > n
	\end{cases}
	\end{align*}
\end{minipage}%
\begin{minipage}{.5\linewidth}
	\begin{align*}
	\op{head}^{\rm{in}}_{G'}(v, j) =
	\begin{cases}
		\op{head}_G(v, j) + n & v \le n\\
		 v - n & v > n
	\end{cases}
	\end{align*}
\end{minipage}
\vspace{10pt}

\noindent The operations $\op{deg}^{\rm{out}}_{G'}$, $\op{deg}^{\rm{in}}_{G'}$ and $\op{head}^{\rm{out}}_{G'}$,
$\op{head}^{\rm{in}}_{G'}$ have the same asymptotic bounds as $\op{deg}_G$ and $\op{head}_G$, respectively.
\end{proof}

We next show that we can compute an $(\ell+1) $st paths.

\begin{lemma}\label{lem:newPath}
Given a valid path data scheme for a set $\mathcal P$ of $\ell\!<\!k$ good $s$-$t$-paths in $G$,
	$O(n (k + \log^* n) k^3 \log^3 k)$ time
	and $O(n + k^2 (\log k) \log n)$ bits
	suffice to compute an array of $2n$ bits storing a path numbering of an $(\ell+1)$th chordless $s$-$t$-path $P^*$,
	which can be a dirty path with respect to  $\mathcal P$.
\end{lemma}
\begin{proof}
To solve the lemma, we adapt the standard network-flow approach as follows. 
We first use the graph interface of the Lemma~\ref{lem:transform} 
to obtain a graph $G''$ in that we can search for edge-disjoint paths instead of vertex-disjoint
paths in $G$.
We provide the operations $\op{prev}$, $\op{next}$, $\op{color}$ and a virtual array $\A$
that gives access to the paths $\mathcal P$ adjusted to $G'$ %numbering 
by using a simple translation to the corresponding data structure of $G$.

According to the construction of $G'$ a vertex $v$ in $G$ is equivalent to two vertices,
an in-vertex $v' = v$ and an out-vertex $v'' = v + n$ in $G'$.
Moreover, every in-vertex $u'$ in $G$ has a single outgoing edge and this edge
points to its out-vertex $u'' = u' + n$, which has only outgoing edges that all
point to some in-vertices in $G'$.
Taken our stored $\ell$ good paths into account a path $(v_1, v_2, v_3, \ldots) \in \mathcal P$
in $G$ translates into a path $(v_1', v_1'', v_2', v_2'', v_3', v_3'', \ldots)$ in $G'$
and the edges between them must be reversed in $G'$.

Next, we want to build the residual network of
$G'$ and $\mathcal P$. This means that we have to reverse the edges on the
paths of $\mathcal P$. To identify the edges incident to a vertex $v$,
 it seems to be natural to use $\op{prev}(v)$ and $\op{next}(v)$.
Unfortunately, we cannot effort to query these two operations many times to
find out which of the incident edges of $v$ is reversed since
$v$ can be a large vertex and the runtime
of both operations depends on $\op{deg}_{G'}(v)$.
This issue becomes especially important when using a space-efficient DFS that
may query $\op{prev}(v)$
more than a constant number of
times (e.g., by running several restorations of a stack segment including~$v$).

To avoid querying $\op{prev}(v)$ for any vertex $v$ of $G'$, 
we present the DFS a graph $G''$ where $v$ has as outgoing edges 
first all outgoing edges of $v$ in $G'$ and 
then all incoming edges of $v$ in $G'$.
In detail, we present the DFS 
a graph $G''$ with all vertices of $G'$ and one further 
vertex $d$ with no outgoing edges. We also make sure that the DFS has 
colored $d$ black (e.g., start the DFS on $d$ before doing anything else).
A sketch of the graph $G'$ with a blue path, graph $G''$ and the reversal of the edges in $G''$
can be seen in Fig.~\ref{fig:transgraph}.
As we see in the next paragraph, we will use $d$ as
a kind of {\em dead end}.

\begin{figure}[h]
	\centering
	\begin{subfigure}{.5\textwidth}
		\centering
		\begin{tikzpicture}[
scale=1.1,
node distance = 10pt,
vertex/.style={minimum width=0.30em, minimum height=0.30em, circle, draw, on 
chain},
svertex/.style={minimum width=0.15em, minimum height=0.15em, circle, draw, 
on chain, inner sep=1pt, text width=5pt, align=center, color=red},
graph/.style={minimum width=3em, minimum height=2em, ellipse, draw, on 
chain, decorate, decoration={snake,segment length=1mm,amplitude=0.2mm}},
decoration={
	markings,
	mark=at position 0.5 with {\arrow{>}}},
],

\node[svertex] (a) at (0, 0) {};
\node[svertex, below=20pt of a] (b) {};
\node[svertex, right=35pt of b] (c) {};
\node[svertex, below=20pt of b] (d) {};

\draw[-]
(a) -- (b)
(b) edge[color=blue] (c)
(b) edge[color=blue] (d)
;

\end{tikzpicture}
 	\end{subfigure}%
	\begin{subfigure}{.5\textwidth}
		\centering
		\begin{tikzpicture}[
scale=1.1,
node distance = 10pt,
vertex/.style={minimum width=0.30em, minimum height=0.30em, circle, draw, on 
chain},
svertex/.style={minimum width=0.15em, minimum height=0.15em, circle, draw, 
on chain, inner sep=1pt, text width=5pt, align=center},
graph/.style={minimum width=3em, minimum height=2em, ellipse, draw, on 
chain, decorate, decoration={snake,segment length=1mm,amplitude=0.2mm}},
decoration={
	markings,
	mark=at position 0.5 with {\arrow{>}}},
],

\node[svertex] (a) at (0, 0) {};
\node[svertex, right=15pt of a] (a2) {};

\draw[color=red] ($(a) + (-4pt, 4pt)$) rectangle ($(a2) + (4pt, -4pt)$) {};

\node[svertex, below=20pt of a] (b) {};
\node[svertex, right=15pt of b] (b2) {};

\draw[color=red] ($(b) + (-4pt, 4pt)$) rectangle ($(b2) + (4pt, -4pt)$) {};

\node[svertex, right=20pt of b2] (c) {};
\node[svertex, right=15pt of c] (c2) {};

\draw[color=red] ($(c) + (-4pt, 4pt)$) rectangle ($(c2) + (4pt, -4pt)$) {};

\node[svertex, below=20pt of b] (d) {};
\node[svertex, right=15pt of d] (d2) {};

\draw[color=red] ($(d) + (-4pt, 4pt)$) rectangle ($(d2) + (4pt, -4pt)$) {};

\path[->, >=stealth]
(a) edge (a2)
(b) edge[color=blue] (b2)
(c) edge[color=blue] (c2)
(d) edge[color=blue] (d2)
(a2) edge (b)
(b2) edge[color=blue] (c)
(b2) edge (d)
(d2) edge[color=blue] (b)
(b2) edge (a)

($(b) + (-10pt, 10pt)$) edge [->, -latex, dashed, bend right=15] (b)
(c2) edge [->, -latex, dashed, bend right=15] ($(c2) + (10pt, 10pt)$)
;

\end{tikzpicture}
 	\end{subfigure}\vspace{2mm}
 	
	\begin{subfigure}{.5\textwidth}
		\centering
		\begin{tikzpicture}[
scale=1.1,
node distance = 10pt,
vertex/.style={minimum width=0.30em, minimum height=0.30em, circle, draw, on 
chain},
svertex/.style={minimum width=0.15em, minimum height=0.15em, circle, draw, 
on chain, inner sep=1pt, text width=5pt, align=center},
bvertex/.style={minimum width=0.15em, minimum height=0.15em, circle, draw, 
on chain, inner sep=1pt, text width=.5pt, align=center, fill=black},
graph/.style={minimum width=3em, minimum height=2em, ellipse, draw, on 
chain, decorate, decoration={snake,segment length=1mm,amplitude=0.2mm}},
decoration={
	markings,
	mark=at position 0.5 with {\arrow{>}}},
],

\node[svertex] (a) at (0, 0) {};
\node[svertex, right=15pt of a] (a2) {};

\node[svertex, below=20pt of a] (b) {};
\node[svertex, right=15pt of b] (b2) {};

\node[svertex, right=20pt of b2] (c) {};
\node[svertex, right=15pt of c] (c2) {};

\node[svertex, below=20pt of b] (d) {};
\node[svertex, right=15pt of d] (d2) {};

\node[bvertex, left=6pt of a2] (1) {};
\node[bvertex, left=6pt of d2] (2) {};
\node[bvertex, left=6pt of b2] (3) {};
\node[bvertex, left=6pt of c2] (4) {};
\node[bvertex, left=9pt of c] (5) {};

\node[bvertex, below=6pt of a, xshift=9pt] (6) {};
\node[bvertex, above=6pt of b, xshift=9pt] (7) {};

\node[bvertex, below=6pt of b, xshift=9pt] (8) {};
\node[bvertex, above=6pt of d, xshift=9pt] (9) {};

\node[bvertex, left=5pt of b, yshift=3pt] (10) {};

\path[->, >=stealth]
(a.45) edge[bend left=10] (a2.135)
(b.38) edge[color=blue, bend left=10] (b2.135)
(c.45) edge[color=blue, bend left=10] (c2.135)
(d.38) edge[color=blue, bend left=10] (d2.135)
(a2.south) edge[bend left=10] (b.50)
(b2.45) edge[color=blue, bend left=10] (c.135)
(b2.south) edge[bend left=10] (d.50)
(d2.north) edge[color=blue, bend right=10] (b.315)
(b2.north) edge[bend right=10] (a.315)

($(b) + (-10pt, 10pt)$) edge [->, -latex, dashed, bend right=15] (b)
(c2) edge [->, -latex, dashed, bend right=15] ($(c2) + (10pt, 10pt)$)

(b.225) edge[bend left=10] (10)

(a2.225) edge[bend left=10] (1.315)
(b2.225) edge[bend left=10] (3.315)
(c2.225) edge[bend left=10] (4.315)
(d2.225) edge[bend left=10] (2.315)
(c.225) edge[bend left=10] (5.315)

(a.south) edge[bend right=10] (6)
(b.north) edge[bend left=10] (7)

(b.south) edge[bend right=10] (8)
(d.north) edge[bend left=10] (9)
;

\end{tikzpicture}
 	\end{subfigure}%
	\begin{subfigure}{.5\textwidth}
		\centering
		\begin{tikzpicture}[
scale=1.1,
node distance = 10pt,
vertex/.style={minimum width=0.30em, minimum height=0.30em, circle, draw, on 
chain},
svertex/.style={minimum width=0.15em, minimum height=0.15em, circle, draw, 
on chain, inner sep=1pt, text width=5pt, align=center},
bvertex/.style={minimum width=0.15em, minimum height=0.15em, circle, draw, 
on chain, inner sep=1pt, text width=.3pt, align=center, fill=black},
graph/.style={minimum width=3em, minimum height=2em, ellipse, draw, on 
chain, decorate, decoration={snake,segment length=1mm,amplitude=0.2mm}},
decoration={
	markings,
	mark=at position 0.5 with {\arrow{>}}},
],

\node[svertex] (a) at (0, 0) {};
\node[svertex, right=15pt of a] (a2) {};

\node[svertex, below=20pt of a] (b) {};
\node[svertex, right=15pt of b] (b2) {};

\node[svertex, right=20pt of b2] (c) {};
\node[svertex, right=15pt of c] (c2) {};

\node[svertex, below=20pt of b] (d) {};
\node[svertex, right=15pt of d] (d2) {};

\node[bvertex, left=6pt of a2] (1) {};
\node[bvertex, left=6pt of d2] (2) {};
\node[bvertex, left=6pt of b2] (3) {};
\node[bvertex, left=6pt of c2] (4) {};
\node[bvertex, left=9pt of c] (5) {};

\node[bvertex, below=6pt of a, xshift=9pt] (6) {};
\node[bvertex, above=6pt of b, xshift=9pt] (7) {};

\node[bvertex, above=5pt of d2, xshift=-7pt] (8) {};
\node[bvertex, above=5pt of d, xshift=7pt] (9) {};

\node[bvertex, left=5pt of b, yshift=3pt] (10) {};

\path[->, >=stealth]
(a.45) edge[bend left=10] (a2.135)
(b.38) edge[color=blue, bend left=10] (3)
(c.45) edge[color=blue, bend left=10] (4)
(d.38) edge[color=blue, bend left=10] (2)
(a2.south) edge[bend left=10] (b.50)
(b2.45) edge[color=blue, bend left=10] (5)
(b2.south) edge[bend left=10] (d.50)
(d2.north) edge[color=blue, bend right=10] (8)
(b2.north) edge[bend right=10] (a.315)

($(b) + (-10pt, 10pt)$) edge [->, -latex, dashed, bend right=15] (b)
(c2) edge [->, -latex, dashed, bend right=15] ($(c2) + (10pt, 10pt)$)

(b.225) edge[bend left=10] (10)

(a2.225) edge[bend left=10] (1.315)
(b2.225) edge[bend left=10] (b.315)
(c2.225) edge[bend left=10] (c.315)
(d2.225) edge[bend left=10] (d.315)
(c.225) edge[bend left=10] (b2.315)

(a.south) edge[bend right=10] (6)
(b.north) edge[bend left=10] (7)

(b.south) edge[bend right=10] (d2.135)
(d.north) edge[bend left=10] (9)
;

\end{tikzpicture}
 	\end{subfigure}%
	\caption{From top left to buttom right, a sketch of $G$, $G'$, $G''$ and $G''$ with reversed edges.
	Each black vertex is the same vertex $b$. The blue edges are edges of a path.
	The dashed lines in the graphs sketches 
	an edge connecting the middle rightmost with the middle leftmost vertex.}\label{fig:transgraph}
\end{figure}
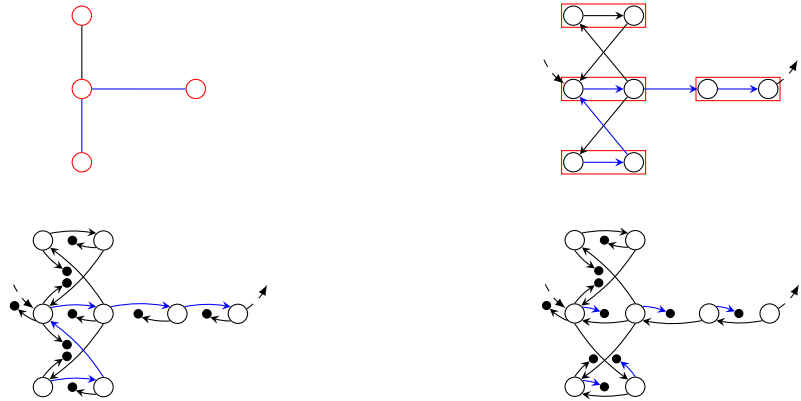

Every vertex $v \neq d$ of $G''$ has
$\op{deg}^{\rm{out}}_{G'}(v) + \op{deg}^{\rm{in}}_{G'}(v)$
many outgoing edges defined as follows.
The heads of the first $\op{deg}^{\rm{out}}_{G'}(v)$ outgoing edges of $v$ in $G''$ are the same as in $G'$
with the exception that we present edge $(v, d)$ for possibly one outgoing edge $(v, z)$ that is reversed (i.e.,
$\op{color}(v) = \op{color}(z) \land \A[z ] = (\A[v ] - 1) \mod 3)$). 
The heads of the next $\op{deg}^{\rm{in}}_{G'}(v)$ incoming edges of $v$ in $G''$ are 
presented as an edge $(v, d)$ with the exception that we present possibly
one incoming edge $(z, v)$
that is reversed
 as an outgoing edge $(v, z)$.

Intuitively speaking, this ensures that a DFS backtracks immediately after using
edges that do not exist in the residual graph. 
Note that we so can decide the
head of an edge in $G''$ by only calling $\op{color}$ twice.
The graph $G''$ has asymptotically
the same amount of vertices and edges as $G$.
To guarantee our space bound we compute a new $s$-$t$-path $P$ in $G''$ using a space-efficient DFS on $G''$
and consecutively make $P$ chordless by Lemma~\ref{lem:chordlessSTPath}.
Finally, we transform $P$ into a %possibly dirty
path $P^*$ with respect to $\mathcal P$ by merging in and out vertices together.

{\bfseries Efficiency:}
Concerning the running time we use
a space-efficient DFS (Lemma~\ref{th:dfs})
to find a new $s$-$t$-path in $G''$.
To operate in $G''$ we need to query the color of vertices.
We pay for that with an extra factor of $O(k^3 \log^3 k)$ 
(Lemma~\ref{lem:recompute}) in the running time.
Since our graph has~$O(n)$ vertices and~$O(kn)$ edges
we find a new $s$-$t$ path in~$G''$ in time $O(n (k + \log^* n) k^3 \log^3 k)$.
The time to transform the found path into~$P^*$
and store it can be done in the same bound.
The total space used for the algorithm is 
$O(n + k^2 (\log k) \log n)$ bits---$O(n)$ bits for the space-efficient DFS,
$O(k^2 (\log k) \log n)$ bits for Lemma~\ref{lem:recompute},
and $O(n)$ bits for storing $P^*$ and its transformation.
\end{proof}

Let us call an $s$-$t$-path $P^*$ a {\em clean path} with respect to $\mathcal P$
if $\mathcal P \cup \{P^*\}$ is a set of good paths, and a {\em dirty path} with respect to $\mathcal P$
if for all $P \in \mathcal P$, 
the common vertices and edges of $P^*$ and $P$ build subpaths each consisting of at least two
vertices and used by the two paths in opposite direction when running from $s$
to $t$ over both paths.
Intuitively, if we construct an $(\ell + 1)$st path $P^*$
in the residual network of $G$ with $\ell$ paths
in~$\mathcal P$,
then $P^*$ can run backwards on paths in $\mathcal P$,
and we call $P^*$ a dirty path (see Fig.~\ref{fig:cleanAndDirtyPaths}).
After a rerouting, we obtain $\ell + 1$ internal vertex disjoint $s$-$t$ paths.
To store the paths $\mathcal P \cup \{P^*\}$ with our path data structure, 
the paths must be good, which we can guarantee through another rerouting.
The details are described below.

\begin{figure}[h]
	\centering
	\begin{tikzpicture}[
%  -{Stealth[length = 2.5pt]},
%start chain = going down,
node distance = 10pt,
vertex/.style={minimum width=0.30em, minimum height=0.30em, circle, draw, on 
chain},
svertex/.style={minimum width=0.15em, minimum height=0.15em, circle, draw, 
on chain, inner sep=1pt, text width=3pt, align=center},
fix/.style={inner sep=1pt, text width=6pt, align=center},
graph/.style={minimum width=3em, minimum height=2em, ellipse, draw, on 
chain, decorate, decoration={snake,segment length=1mm,amplitude=0.2mm}},
decoration={
	markings,
	mark=at position 0.5 with {\arrow{>}}},
path/.style=
	{-,
		decorate,
		decoration={snake,amplitude=.2mm,segment length=2mm,}}
],

\node[] (p) {$P^*$};
\node[below= 2pt of p] (p1) {$P_\ell$};
\node[below= 2pt of p1] (p2) {$P_3$};
\node[below= 2pt of p2] (p3) {$P_2$};
\node[below= 2pt of p3] (p4) {$P_1$};

%\node[svertex, fill=ForestGreen, right = 0pt of p] (s) {};
\node[svertex, fill=black, right = 30pt of p1] (v1) {};
\node[svertex, fill=black, right = 100pt of p1] (v1b) {};
\node[svertex, fill=black, right = 30pt of p2] (v2) {};
\node[svertex, fill=black, right = 100pt of p2] (v2b) {};
\node[svertex, fill=black, right = 100pt of p3] (v3) {};
\node[svertex, fill=black, right = 81pt of p1] (v1q) {};

\node[svertex, fill=black, right = 60pt of p2] (v2v) {};
\node[svertex, fill=ForestGreen, right = 120pt of p3] (v3u) {};

\node[svertex, fill=ForestGreen, above = 5pt of v2v] (v) {};
\node[svertex, fill=ForestGreen, right = 5pt of v3u] (u) {};

\node[svertex, fill=black, right = 90pt of p4] (b1) {};
\node[svertex, fill=black, right = 70pt of p4] (b2) {};

\path[draw, path] ($(p1.east) + (5pt,0)$) -- (v1);
\path[draw, path] ($(p2.east) + (5pt,0)$) -- (v2);
\path[draw, path] ($(p3.east) + (5pt,0)$) -- (v3);
\path[draw, path] ($(p4.east) + (5pt,0)$) -- (b2);
\path[draw, path] (b1) -- ($(p4) + (170pt, 0)$);

\path[draw, path] (v1) -- ($(p1) + (170pt, 0)$);
\path[draw, path] (v2b) -- ($(p2) + (170pt, 0)$);
\path[draw, path] (u) -- ($(p3) + (170pt, 0)$);

\path[draw, path] (v2) -- (v2v);
\path[draw, path] (v2v) -- (v2b);
\path[draw, path, color=red] (u) -- (v3u);

\path[draw, -] (v1) -- (v2);
\path[draw, -] (v2b) -- (v3);
\path[draw, -] (v2b) -- (v1b);

\path[draw, -] (v) -- (v2v);
\path[draw, path] (v3) -- (v3u);

\node[] at ($(u) + (8pt, 5pt)$) (lu) {$u$};

\node[] at ($(p4) + (15pt, 11pt)$) {$\vdots$};
\node[] at ($(p4) + (170pt, 11pt)$) {$\vdots$};

\node[svertex, fill=ForestGreen] at ($(p1) + (80pt, 10pt)$) (y) {};
\path[draw, -] (y) -- (v1q);

\path[draw, path, color=ForestGreen]
($(s) + (15pt,0)$) -- ($(s) + (90pt, 0)$)
($(s) + (90pt, 0)$) edge [path, bend left=30] (u)
;
\path[draw, path, color=red]
(v3u) edge [dashed, bend left=15] (b1) 
(b2) edge [dashed, bend left=50] (v)
(b2) -- (b1)
(v) edge [path, bend left=15] (y)
(y) edge [->, bend left=0] ($(y) + (10pt, 0pt)$)
;

\end{tikzpicture}
	\caption{The green subpath of $P^*$ is clean until it hits a common vertex $u$ with $P_3$,
	then it is a dirty path.}\label{fig:cleanAndDirtyPaths}
\end{figure}

Let $P^*$ be an $s$-$t$-path returned by Lemma~\ref{lem:newPath}.
We first consider the case where $P^*$ is a dirty path with respect to $\mathcal P$.
To make the paths good we cut the paths into subpaths 
and by using the subpaths we construct a set of new $\ell + 1$ good $s$-$t$-paths.
To achieve that we have to get rid of 
common vertices and extended deadlock cycles.
We handle common vertices and extended deadlock cycles in a single process, 
but we first briefly sketch the standard network-flow technique
to remove 
common vertices of $P^*$ with a path $P^c \in \mathcal P$.
By the construction of the paths, the common vertices of an induced subpath are ordered on
$P^c$ as $v_{\sigma_1}, v_{\sigma_2}, \ldots, v_{\sigma_x}$ (for some function $\sigma$ and $x \ge 2$) and on $P^*$ as $v_{\sigma_x}, v_{\sigma_{x - 1}}, \ldots, v_{\sigma_1}$.
$P^c$ and $P^*$ can be split into vertex disjoint paths by 
(1) removing the vertices $v_{\sigma_2}, \ldots, v_{\sigma_{x - 1}}$ from $P^*$ as well as from $P^c$, 
(2) by rerouting path $P^c$ at $v_{\sigma_1}$ to follow $P^*$,
and (3) by rerouting $P^*$ at $v_{\sigma_x}$ to follow~$P^c$.

We denote by $R(\mathcal P \cup \{P^*\})$ a {\em rerouting}
function that returns a new set of $|\mathcal P \cup \{P^*\}|$ good paths.
We so change the successor and predecessor information of some vertices of the paths $\mathcal P \cup \{P^*\}$.
Let~$V^*$ be a set of all vertices in $R(\mathcal P \cup \{P^*\})$.
To define successor and  predecessor information for the vertices we use
the mappings $\Rs, \Rp: V^* \rightarrow \{0, \ldots, k\}$
where for each vertex $u \in V^*$
with $\Rs(u) \ne 0$, $u$'s new successor is a vertex $v \in N(v)
\cap V^*$ with $\Rs(u) = \op{color}(v)$.
Similarly we use $\Rp$ 
to define a new predecessor of some $u \in V^*$.
The triple $(V^*, \Rs, \Rp)$ realizes a rerouting $R$.

To avoid using too much space by storing rerouting information for too many vertices
our approach is to 
compute a rerouting $R(\mathcal P \cup \{P'\})$ only 
for an $s$-$p$-subpath $P'$ of $P^*$ 
where $p$ is a vertex of $P^*$.
This means that $R(\mathcal P \cup \{P'\})$ consists of $\ell$ $s$-$t$-paths
and one $s$-$p$-path.
Moreover, we make the
rerouting in such a way that, for all paths $P_i \in R(\mathcal P \cup \{P'\})$,
there is a vertex $v_i$ with the following property: replacing the $v_i$-$t$-subpath of $P_i$ by an
edge $\{v_i,t\}$ for all paths $P_i$ we get $\ell+1$ good paths. 
Let $\mathcal P^{\rm{c}}$ consist of the $s$-$v_i$-subpaths for each $P_i$. 
We then call the set of vertices part of the paths in $\mathcal P^{\rm{c}}$ 
a {\em clean area} $\Q$ for $R(\mathcal P \cup \{P'\})$.

The idea is to repeatedly compute a rerouting in $O(\log k)$ batches and thereby extend $\Q$ with each
batch by $\Theta(n/\log k)$ vertices---the last batch may be smaller---such 
that we store $O(n / \log k)$ entries in $\Rs$ and $\Rp$ with each batch.
After each batch we free space for the next one
by computing a valid path storage scheme storing good $s$-$t$ paths
$R(\mathcal P \cup \{P'\}) \setminus P''$,
where $P''$ is the only path in $R(\mathcal P \cup \{P'\})$ not ending in $t$,
 and a valid path storage scheme for the $s$-$t$ path
 $\mathcal P^{**}$ obtained from $P^*$ by 
 replacing~$P'$ trough~$P''$.
(For an example consider
the $p$-$w$-subpath of~$P^*$ in Fig.~\ref{fig:monotone}a, which is replaced
by a beginning of $P_2$ in Fig.~\ref{fig:monotone}d.)

%\red{We now sketch our ideas of the rerouting that removes extended deadlock cycles in $\mathcal P \cup \{P^*\}$.
%
%If such a path exists, it can be found from $v$ using a DFS
%in $G[V']$. It is an extended deadlock cycle
%if one of the subpaths that was used by the DFS to reach $w$ from $v$
%consists of at least three vertices.
%As we show in the proof of the next lemma, we can remove 
%an extended deadlock cycle by rerouting the paths
%using the cross edges on the path from $v$ to $w$ that was used
%by the DFS.}\info{Vieleicht sowas hier, anstatt das folgende.}

We now sketch our ideas of the rerouting that removes extended deadlock cycles in $\mathcal P \cup \{P^*\}$.
Recall that every extended deadlock cycle
must contain parts of $P^*$ since all remaining 
paths are good. %and that a deadlock cycle consists of 
%subpaths of our paths and the cross edges between them.
Hence by moving along $P^*$ we look for an vertex $u$ of $P^*$ 
that is an endpoint of a cross edge $\{u, v\}$.
Since a deadlock cycle is a cycle, there must be
some vertex~$w$ after $u$ on $P^*$ that is connected
by subpaths of $\mathcal P$ and cross edges connecting
the paths in $\mathcal P \cup \{P^*\}$.

Since common vertices and deadlock cycles may
intersect, vertices $u$ and $w$ can have 
a cross edge or be common vertex on a path in~$\mathcal P$
(Fig.~\ref{fig:monotone}a).
To find $u$ and $w$,
our idea is to move over $P^*$ starting from $s$ and stop at the first
such vertex $u$ of $P^*$.
Then use a modified DFS at $u$ that runs over
paths in $\mathcal P$ only in reverse direction and over 
cross edges---but never over two subsequent cross edges---and so 
explores subpaths of $\mathcal P$ (marked orange in
Fig.~\ref{fig:monotone}b). 
Whenever a vertex $v$ of the clean area is reached,
the DFS backtracks, i.e., the DFS assumes that $v$ has no outgoing edges.
We so guarantee that the clean area is not explored again and again.
Afterwards we determine the latest vertex $w$ on $P^*$ that 
is a common vertex or has a cross edge with one of the explored subpaths.
If such a vertex $w$ is found, 
we either have a common subpath from $u$ to $w$ between a path $P^c \in \mathcal P$ and $P^*$
(which is removed as described above) or a 
deadlock cycle. If it is 
an extended deadlock cycle, we have 
a subpath on the cycle consisting of at least three vertices.
As the proof of the following lemma shows
the extended deadlock cycle can be destroyed by
removing the inner vertices of the subpath
and rerouting the paths via cross edges part of the extended deadlock cycle.
We find an extended deadlock cycle by an additional run of the modified DFS from $u$
to $w$ (Fig.~\ref{fig:monotone}c) and reroute 
it along the cross edges. Fig.~\ref{fig:monotone}d shows a rerouting where the path in $R(\mathcal P \cup \{P'\})$
starting with the vertices of the old path $P_2$ becomes the ``new'' path~$P^*$.

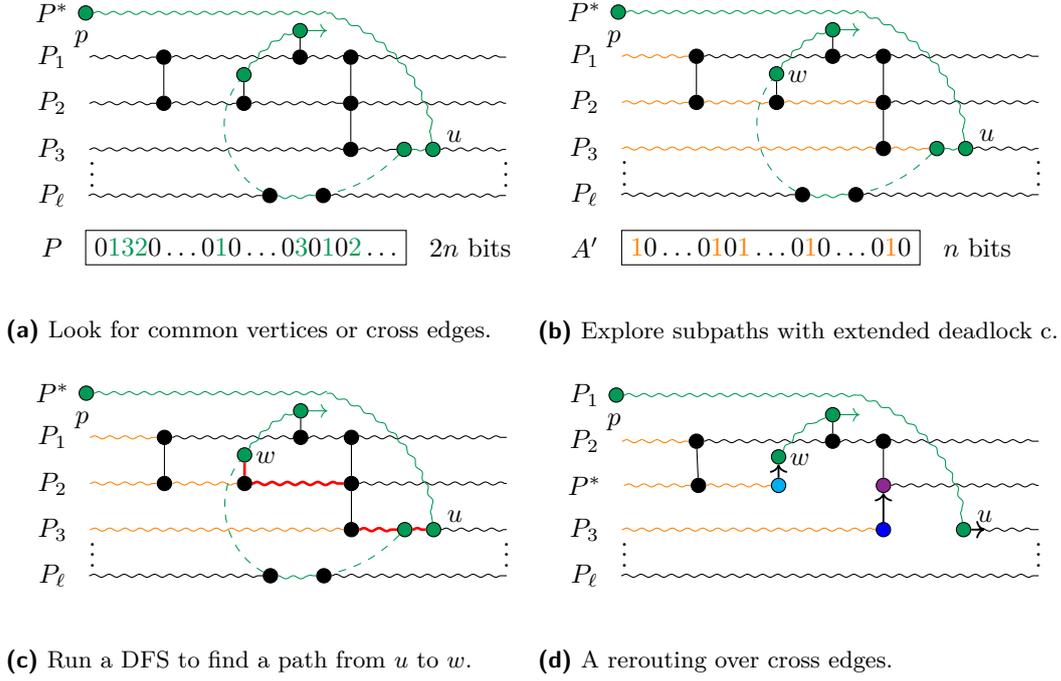
\begin{figure}[h]
			\centering
			\begin{subfigure}{.5\textwidth}
				\centering
				\begin{tikzpicture}[
%  -{Stealth[length = 2.5pt]},
%start chain = going down,
node distance = 10pt,
vertex/.style={minimum width=0.30em, minimum height=0.30em, circle, draw, on 
chain},
svertex/.style={minimum width=0.15em, minimum height=0.15em, circle, draw, 
on chain, inner sep=1pt, text width=3pt, align=center},
fix/.style={inner sep=1pt, text width=6pt, align=center},
graph/.style={minimum width=3em, minimum height=2em, ellipse, draw, on 
chain, decorate, decoration={snake,segment length=1mm,amplitude=0.2mm}},
decoration={
	markings,
	mark=at position 0.5 with {\arrow{>}}},
path/.style=
	{-,
		decorate,
		decoration={snake,amplitude=.2mm,segment length=2mm,}}
],

\node[] (p) {$P^*$};
\node[below= 2pt of p] (p1) {$P_1$};
\node[below= 2pt of p1] (p2) {$P_2$};
\node[below= 2pt of p2] (p3) {$P_3$};
\node[below= 2pt of p3] (p4) {$P_\ell$};

\node[svertex, fill=ForestGreen, right = 0pt of p] (s) {};
\node[svertex, fill=black, right = 30pt of p1] (v1) {};
\node[svertex, fill=black, right = 100pt of p1] (v1b) {};
\node[svertex, fill=black, right = 30pt of p2] (v2) {};
\node[svertex, fill=black, right = 100pt of p2] (v2b) {};
\node[svertex, fill=black, right = 100pt of p3] (v3) {};
\node[svertex, fill=black, right = 81pt of p1] (v1q) {};

\node[svertex, fill=black, right = 60pt of p2] (v2v) {};
\node[svertex, fill=ForestGreen, right = 120pt of p3] (v3u) {};

\node[svertex, fill=ForestGreen, above = 5pt of v2v] (v) {};
\node[svertex, fill=ForestGreen, right = 5pt of v3u] (u) {};

\node[svertex, fill=black, right = 90pt of p4] (b1) {};
\node[svertex, fill=black, right = 70pt of p4] (b2) {};

\path[draw, path] ($(p1.east) + (5pt,0)$) -- (v1);
\path[draw, path] ($(p2.east) + (5pt,0)$) -- (v2);
\path[draw, path] ($(p3.east) + (5pt,0)$) -- (v3);
\path[draw, path] ($(p4.east) + (5pt,0)$) -- (b2);
\path[draw, path] (b1) -- ($(p4) + (170pt, 0)$);

\path[draw, path] (v1) -- ($(p1) + (170pt, 0)$);
\path[draw, path] (v2b) -- ($(p2) + (170pt, 0)$);
\path[draw, path] (u) -- ($(p3) + (170pt, 0)$);

\path[draw, path] (v2) -- (v2v);
\path[draw, path] (v2v) -- (v2b);
\path[draw, path, color=ForestGreen] (u) -- (v3u);

\path[draw, -] (v1) -- (v2);
\path[draw, -] (v2b) -- (v3);
\path[draw, -] (v2b) -- (v1b);

\path[draw, -] (v) -- (v2v);
\path[draw, path] (v3) -- (v3u);

\node[] at ($(p) + (11pt, -10pt)$) (ppp) {$p$}; %{$s'$};

\node[] at ($(u) + (8pt, 5pt)$) (lu) {$u$};
%\node[] at ($(v) + (8pt, -1pt)$) (lv) {$w$};

\node[] at ($(p4) + (15pt, 11pt)$) {$\vdots$};
\node[] at ($(p4) + (170pt, 11pt)$) {$\vdots$};

\node[svertex, fill=ForestGreen] at ($(p1) + (80pt, 10pt)$) (y) {};
\path[draw, -] (y) -- (v1q);

\path[draw, path, color=ForestGreen]
(s) -- ($(s) + (90pt, 0)$)
($(s) + (90pt, 0)$) edge [path, bend left=30] (u)
(v3u) edge [dashed, bend left=15] (b1) 
(b2) edge [dashed, bend left=50] (v)
(b2) -- (b1)
(v) edge [path, bend left=15] (y)
(y) edge [->, bend left=0] ($(y) + (10pt, 0pt)$)
;

\long\def\cg#1{%
	{{\color{ForestGreen}{#1}}}%
	% #1%
}

\node[below = 5pt of p4] (dsp) {$P$};
\node[draw, right = 5pt of dsp] (sdpe) 
{$0\cg{13}\cg{2}0\ldots0\cg{1}0\ldots0\cg{3}0\cg{1}0\cg{2}\ldots$};
\node[right= 5pt of sdpe] (sp) {$2n$ bits};
\end{tikzpicture}
				\caption{Look for common vertices or cross edges.}
 			\end{subfigure}%
			\begin{subfigure}{.5\textwidth}
				\centering
				\begin{tikzpicture}[
%  -{Stealth[length = 2.5pt]},
%start chain = going down,
node distance = 10pt,
vertex/.style={minimum width=0.30em, minimum height=0.30em, circle, draw, 
on 
	chain},
svertex/.style={minimum width=0.15em, minimum height=0.15em, circle, 
draw, 
	on chain, inner sep=1pt, text width=3pt, align=center},
fix/.style={inner sep=1pt, text width=6pt, align=center},
graph/.style={minimum width=3em, minimum height=2em, ellipse, draw, on 
	chain, decorate, decoration={snake,segment 
	length=1mm,amplitude=0.2mm}},
decoration={
	markings,
	mark=at position 0.5 with {\arrow{>}}},
path/.style=
{-,
	decorate,
	decoration={snake,amplitude=.2mm,segment length=2mm,}}
],

\node[] (p) {$P^*$};
\node[below= 2pt of p] (p1) {$P_1$};
\node[below= 2pt of p1] (p2) {$P_2$};
\node[below= 2pt of p2] (p3) {$P_3$};
\node[below= 2pt of p3] (p4) {$P_\ell$};

\node[svertex, fill=ForestGreen, right = 0pt of p] (s) {};
\node[svertex, fill=black, right = 30pt of p1] (v1) {};
\node[svertex, fill=black, right = 30pt of p2] (v2) {};
\node[svertex, fill=black, right = 100pt of p2] (v2b) {};
\node[svertex, fill=black, right = 100pt of p3] (v3) {};
\node[svertex, fill=black, right = 100pt of p1] (v1b) {};
\node[svertex, fill=black, right = 81pt of p1] (v1q) {};

\node[svertex, fill=black, right = 60pt of p2] (v2v) {};
\node[svertex, fill=ForestGreen, right = 120pt of p3] (v3u) {};

\node[svertex, fill=ForestGreen, above = 5pt of v2v] (v) {};
\node[svertex, fill=ForestGreen, right = 5pt of v3u] (u) {};

\node[svertex, fill=black, right = 90pt of p4] (b1) {};
\node[svertex, fill=black, right = 70pt of p4] (b2) {};

\path[draw, path, color= orange] ($(p1.east) + (5pt,0)$) -- (v1);
\path[draw, path, color= orange] ($(p2.east) + (5pt,0)$) -- (v2);
\path[draw, path, color= orange] ($(p3.east) + (5pt,0)$) -- (v3);
%\path[draw, path] ($(p4.east) + (5pt,0)$) -- ($(p4) + (170pt, 0)$);
\path[draw, path] ($(p4.east) + (5pt,0)$) -- (b2);
\path[draw, path] (b1) -- ($(p4) + (170pt, 0)$);

\path[draw, path] (v1) -- ($(p1) + (170pt, 0)$);
\path[draw, path] (v2b) -- ($(p2) + (170pt, 0)$);
\path[draw, path] (u) -- ($(p3) + (170pt, 0)$);

\path[draw, path, color= orange] (v2) -- (v2v);
\path[draw, path, color= orange] (v2v) -- (v2b);
\path[draw, path, color= orange] (v3) -- (v3u);

\path[draw, -] (v1) -- (v2);
\path[draw, -] (v2b) -- (v3);
\path[draw, -] (v2b) -- (v1b);

\path[draw, -] (v) -- (v2v);
\path[draw, path, color=ForestGreen] (u) -- (v3u);

\node[] at ($(u) + (8pt, 5pt)$) (lu) {$u$};
\node[] at ($(v) + (8pt, -1pt)$) (lv) {$w$};

\node[] at ($(p4) + (15pt, 11pt)$) {$\vdots$};
\node[] at ($(p4) + (170pt, 11pt)$) {$\vdots$};

\node[] at ($(p) + (11pt, -10pt)$) (ppp) {$p$}; %{$s'$};

\node[svertex, fill=ForestGreen] at ($(p1) + (80pt, 10pt)$) (y) {};
\path[draw, -] (y) -- (v1q);

\path[draw, path, color=ForestGreen]
(s) -- ($(s) + (90pt, 0)$)
($(s) + (90pt, 0)$) edge [path, bend left=30] (u)
(v3u) edge [dashed, bend left=15] (b1) 
(b2) edge [dashed, bend left=50] (v)
(b2) -- (b1)
(v) edge [path, bend left=15] (y)
(y) edge [->, bend left=0] ($(y) + (10pt, 0pt)$)
;

\long\def\co#1{%
	{{\color{orange}{#1}}}%
	% #1%
}

\node[below = 5pt of p4] (dsp) {$A'$};
\node[draw, right = 5pt of dsp] (sdpe) 
{$\co{1}0\ldots0\co{1}0\co{1}\ldots0\co{1}0\ldots0\co{1}0$};
\node[right= 5pt of sdpe] (sp) {$n$ bits};

\end{tikzpicture}
				\caption{Explore subpaths with extended deadlock c.}
			\end{subfigure}\vspace{3mm}%
			
			\begin{subfigure}{.5\textwidth}
				\centering
				\begin{tikzpicture}[
%  -{Stealth[length = 2.5pt]},
%start chain = going down,
node distance = 10pt,
vertex/.style={minimum width=0.30em, minimum height=0.30em, circle, draw, 
	on 
	chain},
svertex/.style={minimum width=0.15em, minimum height=0.15em, circle, 
	draw, 
	on chain, inner sep=1pt, text width=3pt, align=center},
fix/.style={inner sep=1pt, text width=6pt, align=center},
graph/.style={minimum width=3em, minimum height=2em, ellipse, draw, on 
	chain, decorate, decoration={snake,segment 
		length=1mm,amplitude=0.2mm}},
decoration={
	markings,
	mark=at position 0.5 with {\arrow{>}}},
path/.style=
{-,
	decorate,
	decoration={snake,amplitude=.2mm,segment length=2mm,}}
],

\node[] (p) {$P^*$};
\node[below= 2pt of p] (p1) {$P_1$};
\node[below= 2pt of p1] (p2) {$P_2$};
\node[below= 2pt of p2] (p3) {$P_3$};
\node[below= 2pt of p3] (p4) {$P_\ell$};

\node[svertex, fill=ForestGreen, right = 0pt of p] (s) {};
\node[svertex, fill=black, right = 30pt of p1] (v1) {};
\node[svertex, fill=black, right = 30pt of p2] (v2) {};
\node[svertex, fill=black, right = 100pt of p2] (v2b) {};
\node[svertex, fill=black, right = 100pt of p3] (v3) {};
\node[svertex, fill=black, right = 100pt of p1] (v1b) {};
\node[svertex, fill=black, right = 81pt of p1] (v1q) {};

\node[svertex, fill=black, right = 60pt of p2] (v2v) {};
\node[svertex, fill=ForestGreen, right = 120pt of p3] (v3u) {};

\node[svertex, fill=ForestGreen, above = 5pt of v2v] (v) {};
\node[svertex, fill=ForestGreen, right = 5pt of v3u] (u) {};

\node[svertex, fill=black, right = 90pt of p4] (b1) {};
\node[svertex, fill=black, right = 70pt of p4] (b2) {};

\path[draw, path, color= orange] ($(p1.east) + (5pt,0)$) -- (v1);
\path[draw, path, color= orange] ($(p2.east) + (5pt,0)$) -- (v2);
\path[draw, path, color= orange] ($(p3.east) + (5pt,0)$) -- (v3);
\path[draw, path] ($(p4.east) + (5pt,0)$) -- (b2);
\path[draw, path] (b1) -- ($(p4) + (170pt, 0)$);

\path[draw, path] (v1) -- ($(p1) + (170pt, 0)$);
\path[draw, path] (v2b) -- ($(p2) + (170pt, 0)$);
\path[draw, path] (u) -- ($(p3) + (170pt, 0)$);

\path[draw, path, color= orange] (v2) -- (v2v);
\path[draw, path, color= red, line width = 1pt] (v2v) -- (v2b);
\path[draw, path, color= red, line width = 1pt] (v3) -- (v3u);

\path[draw, -] (v1) -- (v2);
\path[draw, -] (v2b) -- (v3);
\path[draw, -] (v2b) -- (v1b);

\path[draw, -, color=red, line width = 1pt] (v) -- (v2v);
\path[draw, path, color=red, line width = 1pt] (v3u) -- (u);

\node[] at ($(u) + (8pt, 5pt)$) (lu) {$u$};
\node[] at ($(v) + (8pt, -1pt)$) (lv) {$w$};

\node[] at ($(p4) + (15pt, 11pt)$) {$\vdots$};
\node[] at ($(p4) + (170pt, 11pt)$) {$\vdots$};

\node[] at ($(p) + (11pt, -10pt)$) (ppp) {$p$}; %{$s'$};

\node[svertex, fill=ForestGreen] at ($(p1) + (80pt, 10pt)$) (y) {};
\path[draw, -] (y) -- (v1q);

\path[draw, path, color=ForestGreen]
(s) -- ($(s) + (90pt, 0)$)
($(s) + (90pt, 0)$) edge [path, bend left=30] (u)
(v3u) edge [dashed, bend left=15] (b1) 
(b2) edge [dashed, bend left=50] (v)
(b2) -- (b1)
(v) edge [path, bend left=15] (y)
(y) edge [->, bend left=0] ($(y) + (10pt, 0pt)$)
;
\long\def\cr#1{%
	{{\color{red}{#1}}}%
	% #1%
}
%\node[below = 5pt of p4] (dsp) {$D$};
%\node[draw, right = 5pt of dsp] (sdpe) 
%{$0\ldots\cr{1}\cr{1}\ldots\cr{11}00\ldots0\cr{11}0\ldots$};
%\node[right= 5pt of sdpe] (sp) {$n$ bits};

\end{tikzpicture}
				\caption{Run a DFS to find a path from $u$ to $w$.}
			\end{subfigure}%
			\begin{subfigure}{.5\textwidth}
				\centering
				\begin{tikzpicture}[
%  -{Stealth[length = 2.5pt]},
%start chain = going down,
node distance = 10pt,
vertex/.style={minimum width=0.30em, minimum height=0.30em, circle, draw, 
	on 
	chain},
svertex/.style={minimum width=0.15em, minimum height=0.15em, circle, 
	draw, 
	on chain, inner sep=1pt, text width=3pt, align=center},
fix/.style={inner sep=1pt, text width=6pt, align=center},
graph/.style={minimum width=3em, minimum height=2em, ellipse, draw, on 
	chain, decorate, decoration={snake,segment 
		length=1mm,amplitude=0.2mm}},
decoration={
	markings,
	mark=at position 0.5 with {\arrow{>}}},
path/.style=
{-,
	decorate,
	decoration={snake,amplitude=.2mm,segment length=2mm,}}
],

\node[] (p) {$P_1$};
\node[below= 2pt of p] (p1) {$P_2$};
\node[below= 2pt of p1] (p2) {$P^*$};
\node[below= 2pt of p2] (p3) {$P_3$};
\node[below= 2pt of p3] (p4) {$P_\ell$};

\node[svertex, fill=ForestGreen, right = 0pt of p] (s) {};
\node[svertex, fill=black, right = 30pt of p1] (v1) {};
\node[svertex, fill=black, right = 30pt of p2] (v2) {};
\node[svertex, fill=Plum, right = 99.25pt of p2] (v2b) {};
\node[svertex, fill=blue, right = 100pt of p3] (v3) {};
\node[svertex, fill=black, right = 100pt of p1] (v1b) {};
\node[svertex, fill=black, right = 81pt of p1] (v1q) {};

\node[svertex, fill=cyan, right = 60pt of p2] (v2v) {};
%\node[svertex, fill=red, right = 120pt of p3] (v3u) {};

\node[svertex, fill=ForestGreen, above = 5pt of v2v] (v) {};
\node[svertex, fill=ForestGreen, right = 130pt of p3] (u) {};

\path[draw, path, color= orange] ($(p1.east) + (5pt,0)$) -- (v1);
\path[draw, path, color= orange] ($(p2.east) + (5pt,0)$) -- (v2);
\path[draw, path, color= orange] ($(p3.east) + (5pt,0)$) -- (v3);
\path[draw, path] ($(p4.east) + (5pt,0)$) -- ($(p4) + (170pt, 0)$);

\path[draw, path] (v1) -- ($(p1) + (170pt, 0)$);
\path[draw, path] (v2b) -- ($(p2) + (170pt, 0)$);
\path[draw, path] ($(u) + (5pt, 0)$) -- ($(p3) + (170pt, 0)$);

\path[draw, path, color= orange] (v2) -- (v2v);
%\path[draw, dotted, color= red, line width = 1pt] (v2v) -- (v2b);
%\path[draw, dotted, color= red, line width = 1pt] (v3) -- (v3u);
%\path[draw, dotted, color= red, line width = 1pt] (v3u) -- (u);

\node[] at ($(p) + (11pt, -10pt)$) (ppp) {$p$}; %{$s'$};

\path[draw, -] (v1) -- (v2);
\path[draw, -] (v2b) -- (v1b);

\path[draw, <-, line width=0.7pt] (v2b) -- (v3);

\path[draw, <-, line width= 0.7pt] (v) -- (v2v);
\path[draw, ->, line width= 0.7pt] (u) -- ($(u.east) + (5pt, 0)$);

\node[] at ($(u) + (8pt, 5pt)$) (lu) {$u$};
\node[] at ($(v) + (8pt, -1pt)$) (lv) {$w$};

\node[] at ($(p4) + (15pt, 11pt)$) {$\vdots$};
\node[] at ($(p4) + (170pt, 11pt)$) {$\vdots$};

\node[svertex, fill=ForestGreen] at ($(p1) + (80pt, 10pt)$) (y) {};
\path[draw, -] (y) -- (v1q);

\path[draw, path, color=ForestGreen]
(s) -- ($(s) + (90pt, 0)$)
(v) edge [path, bend left=15] (y)
($(s) + (90pt, 0)$) edge [path, bend left=30] (u)
(y) edge [->, bend left=0] ($(y) + (10pt, 0pt)$)
;
%\path[draw, path, color=red]
%(v3u) edge [dotted, bend left=50] ($(p4) + (80pt,10pt)$) 
%($(p4) + (80pt,10pt)$) edge [dotted, bend left=50] (v)
%;
\long\def\cr#1{%
	{{\color{red}{#1}}}%
}
%\node[below = 5pt of p4] (dsp) {$D$};
%\node[draw, right = 5pt of dsp] (sdpe) 
%{$0\ldots\cr{1}\cr{1}\ldots\cr{11}00\ldots0\cr{11}0\ldots$};
%\node[right= 5pt of sdpe] (sp) {$n$ bits};
\end{tikzpicture}
				\caption{A rerouting over cross edges.}
			\end{subfigure}%
			\caption{Steps and data structures of the algorithm
to create good paths.}\label{fig:monotone}
\end{figure}

An $s$-$p$-path $\CP$ is called a {\em clean subpath} with respect
to a set of good paths $\mathcal P'$ if, for the extension $P^{\mathrm{ext}}$ of $\CP$ by an edge
$\{p,t\}$ and a vertex $t$, $\mathcal P' \cup \{P^{\mathrm{ext}}\}$ is a set of good paths.
The green path in Fig.~\ref{fig:cleanAndDirtyPaths} is a clean subpath.
The next lemma summarizes the computation of our rerouting $R$.
We initially call the lemma with a clean subarea $\Q = \{s\}$ and a clean path that consists
only of vertex $p = s$.

\begin{lemma}\label{lem:rerouting}
	There is an algorithm that accepts as input
	a valid path data scheme for $\ell$ good $s$-$t$-paths $\mathcal P$,
	a valid path data scheme for a possibly dirty $s$-$t$ path $P^*$,
	as well as a clean area $\Q$ for $\mathcal P$ including a clean $s$-$p$ subpath 
	$\CP$ of $P^*$
	and outputs a rerouting $R$,
	a clean area $\Q'$ including a clean $s$-$p'$ subpath $\CP'$
	with the properties that
	$\CP$ is a subpath of $\CP'$, and
	$|\Q'| = |\Q| + \Omega(n/\log k) \lor p' = t$.
	The algorithm runs in $O(n k^3 \log^3 k)$
	time and uses $O(n + k^2 (\log k) \log n)$~bits.
\end{lemma}
	\begin{proof}
		We begin to describe our algorithm to detect common vertices and extended deadlock cycles.
		Afterwards we compute the rerouting $R$ which is realized by $(V^*, \Rs, \Rp)$.
		Initially, $V^*$ is an array of bits consisting of all vertices of $\mathcal P \cup \{P^*\}$.
		During the algorithm we remove vertices from $V^*$ whenever we change our paths.
		We want to use static space allocation to store the mappings $\Rs$ and $\Rp$,
		which requires us to know the key set of the mapping in advance.
		We solve it by running the algorithm twice and compute the key set
		in the first run, reset $\Q$ and $p$ from a backup, and compute the values and store them in a second run.
		For simplicity, we omit these details below and assume that we 
		can simply store the values inside $\Rs$ and $\Rp$. 
		
		Starting from vertex $p$
		we run along path $P^*$.
		We stop at the first vertex $u$ if $u$ is a common vertex of $P^*$ and a path in $\mathcal P$
		(Fig.~\ref{fig:monotone}a)
		or it is an endpoint of a cross edge.
                For a simpler notion, we call such a vertex $u$ of $P^*$ a vertex
		{\em touching} the vertices in $\mathcal P$.
                Next we use a DFS without restorations (due to our
                details described below, the DFS stack consists of only $O(\ell)$
		vertices, which are incident
                to cross edges of the current DFS path). We so explore the vertices $\Q'$ from $u$ that are
		reachable from $u$ via edges on the paths $\mathcal P$ used in
		reverse direction and cross edges, but ignore
		a cross edge if it immediately follows after another cross edge and ignore
		the
		vertices in the clean area $\Q$.
		To allow an economical way of
		storing the stack, the DFS prefers reverse edges (edges on a path $P \in \mathcal P$) compared
		to cross edges when iterating over the outgoing edges of a
		vertex. 
		This guarantees that the vertices on the DFS-stack consist of
		at most $\ell$ subpaths of paths in $\mathcal P$.
        We store the vertices $\Q'$ processed by the DFS in a choice 
        dictionary~\cite{Hag18cd,HagK16,KamS18c}  since
        a choice dictionary allows us to compute $\Q := \Q \cup \Q'$
        in a time linear to the amount of elements in $\Q'$.
        		Moreover, we count the number $q$ of vertices of $P^*$ that touch $\Q'$.
		 Running over $P^*$ starting from $u$ we can
		determine the last such vertex $w$
		(Fig.~\ref{fig:monotone}b).

		Then we run a modified DFS to construct a $u$-$w$ path $\widetilde{P}$
                that (1) consists only of vertices in~$\Q'$, 
                that (2) uses no subsequent cross edges 
                and that (3)	 uses every path in $\mathcal P$ only once (Fig.~\ref{fig:monotone}c).
        Note that these restrictions (1) -- (3) still allow us to reach all vertices of~$\Q'$.
		To ensure restriction (3) we maintain a bit array $F$ of {\em forbidden}
		paths where $i \in F$ exactly if a subpath of $P_i\in \mathcal P$
		is part of the currently constructed path.
		To construct $\widetilde{P}$ the DFS starts at vertex $u$ and processes a vertex $v$ as follows:
		\begin{enumerate}
	           \item Stop the DFS if $w$ is reached.
                   \item {\color{gray}// To guarantee restriction (1):}\\
                    If $v \notin  \Q'$, then
                   backtrack the DFS 
                   and color $v$ white.
            \item Iterate over all reversed edges $(v, v')$: recursively process $v'$.
            \item {\color{gray}// To guarantee restriction (2) and (3):}\\ 
                 Iterate over all cross edges $\{v, v'\}$: \\
                 If $\op{color}(v') \notin F$ and $v'$ was not discovered by a cross edge, 
                 recursively process $v'$. 
            \item Backtrack the DFS to the predecessor $\widetilde{v}$ of $v$ on $\widetilde{P}$.
                  If $\{\widetilde{v}, v\}$ is a cross-edge, 
                   color $v$ white.\\
                  {\color{gray}// Coloring $v$ white guarantees that we can
explore outgoing cross-edges of $v$ if we reach $v$ again over a reverse edges.}
		\end{enumerate}

		Let $1,\ldots,\ell$ be the colors of the paths in $\mathcal P$ and $(\ell  + 1)$ be the color of $P^*$.
		In the case that $u$ and $w$ are 
        both on one path of 
 $\mathcal P$,
		we reroute the paths as follows. 
		Set $\Rs(u) = \Rp(w) = \op{color}(u)$,
        $\Rp(u) = \Rs(w) = \ell + 1$ and update $V^*$ accordingly, i.e.,
        remove all internal vertices of the $u$-$w$ subpath of $P^*$ from $V^*$ (Fig.~\ref{fig:rerouting}a).
        Afterwards, we set $\Q := \Q \cup \Q'$ 
        and $p = \op{next}(w)$ and repeat the whole algorithm above.

                Recall that the construction of the $u$-$w$ path uses only
                one subpath of every path in~$\mathcal P$.
		Without loss of generality, let $\widetilde{P}$ be
		$w, x_1, \ldots, y_1,$ $x_2, \ldots, y_2,$ $x_3, \ldots,y_{\ell'}, u$,
		where $(y_i, x_{i + 1})$ are cross edges and $x_i, \ldots, y_i$ are subpaths of a path $P_i \in \mathcal P$
		(for some order of $\mathcal P$ and some $\ell'$ with $0 < i < \ell' \le \ell$).
				
		We first assume that $(w, x_1)$ and $(y_{\ell'}, u)$ are both cross edges (Fig.~\ref{fig:rerouting}b).
		If the $x_i$-$y_i$-subpath of every $P_i$ as well as
		the $u$-$w$-subpath of $P^*$ consists of exactly two vertices, then we found a simple deadlock cycle, in particular, $u$ and $w$ are the only vertices of $P'$ that touch~$\Q'$.
                We keep the simple deadlock cycle by setting $\Q := \Q \cup \Q'$ and $p = \op{next}(w)$ and repeat the whole algorithm
		above.

		Otherwise we have an extended deadlock cycle and we compute a rerouting as follows.
		For every cross edge $(y_i, x_{i + 1})$ with $0 < i < \ell'$ on the DFS stack,
		we set $\Rs(x_{i + 1}) = i$, $\Rp(y_i) = i + 1$, remove all vertices of each path $P_i$ in between $x_i$ and $y_i$
		as well as all vertices of path $P^*$ in between $u$ and $w$ from $V^*$.
		We now describe the rerouting at $w$.
		Vertex $u$ is handled analogously. 
		If as assumed $(w, x_1)$ is a cross edge, 
		set $\Rs(x_1) = \ell + 1$ and $\Rp(w) = 1$ (Fig.~\ref{fig:rerouting}c).
		Otherwise,
		$w$ is on path $P_1$ before vertex $y_1$, or $w = x_1$. 
		Then remove all vertices in between $w$ and $y_1$ on path $P_1$ from $V^*$
		and set $\Rs(w) = \ell + 1$ and $\Rp(w) = 1$ (Fig.~\ref{fig:rerouting}d).
		For the case where neither $(w, x_1)$ nor $(y_{\ell'}, u)$ is a cross edge,
		as well as $u$ and $w$ are on different paths, see Fig.~\ref{fig:rerouting}e
		for the rerouting.
        After all cases described
        in this paragraph set $\Q := \Q \cup \Q'$, $p = \op{next}(w)$ and repeat the whole algorithm
        above.
		
		Note that in all cases above, we only add information for $O(k)$ predecessors and successors
		with each extension of $\Q$ to our rerouting.
		Thus, our rerouting information increases in small pieces.
		Whenever $\Rs$ is defined for $\Theta(n / \log k)$ vertices ($R$ uses $O(n)$ bits)
		or we reached $t$ while moving over $P^*$, we 
                break the algorithm above. 
        By applying Lemma~\ref{lem:chordlessSTPath} on each rerouted path
		(the graph induced by the vertices of the paths is taken as input graph), we can make all paths chordless
		without introducing further cross edges or extended deadlock cycles.
                Finally, return the rerouting $(V^*, \Rs, \Rp)$ as well as the
		new clean area $\Q$ and the end vertex $p$ of a new clean path in $R(\mathcal P \cup {P'})$. 

\begin{figure}[h]
			\centering
			\begin{subfigure}{0.5\textwidth}
				\centering
				\begin{tikzpicture}[
%  -{Stealth[length = 2.5pt]},
%start chain = going down,
node distance = 10pt,
vertex/.style={minimum width=0.30em, minimum height=0.30em, circle, draw, on 
chain},
svertex/.style={minimum width=0.15em, minimum height=0.15em, circle, draw, 
on chain, inner sep=1pt, text width=3pt, align=center},
fix/.style={inner sep=1pt, text width=6pt, align=center},
graph/.style={minimum width=3em, minimum height=2em, ellipse, draw, on 
chain, decorate, decoration={snake,segment length=1mm,amplitude=0.2mm}},
decoration={
	markings,
	mark=at position 0.5 with {\arrow{>}}},
path/.style=
	{-,
		decorate,
		decoration={snake,amplitude=.2mm,segment length=2mm,}}
],

\node[] (p) {$P_{\ell + 1}$};
\node[below= 2pt of p] (pl) {$P_{\ell'}$};
\node[below= 2pt of pl] (p3) {$P_3$};
\node[below= 2pt of p3] (p2) {$P_2$};
\node[below= 2pt of p2] (p1) {$P_1$};

\node[svertex, fill=black, right = 30pt of pl] (pl1) {};
\node[svertex, fill=black, right = 45pt of pl, color=red] (pl2) {};

\node[svertex, fill=black, right = 110pt of pl, color=red] (pl3) {};
\node[svertex, fill=black, right = 125pt of pl] (pl4) {};

\node[svertex, fill=black, right = 50pt of p3] (p31) {};
\node[svertex, fill=black, right = 65pt of p3] (p32) {};

\node[svertex, fill=black, right = 60pt of p2] (p21) {};
\node[svertex, fill=black, right = 75pt of p2] (p22) {};

\node[svertex, fill=black, right = 120pt of p2] (p23) {};
\node[svertex, fill=black, right = 135pt of p2] (p24) {};

% to let the picture be larger
\node[svertex, fill=white, color=white, below=24pt of p21] (spacer) {};
\path[draw, path, color=white, ->, bend left=15]
(spacer) -- ($(spacer) + (0pt, -17pt)$) 
;
% END to let the picture be larger

\path[draw, path] ($(pl.east) + (5pt,0)$) -- (pl1);
\path[draw, dotted, color=red, line width = 1pt] (pl2) -- (pl3);
\path[draw, path] (pl4) -- ($(pl) + (170pt, 0)$);

\path[draw, path] ($(p3.east) + (5pt,0)$) -- (p31);
\path[draw, path] (p32) -- ($(p3) + (170pt, 0)$);

\path[draw, path] ($(p2.east) + (5pt,0)$) -- (p21);
\path[draw, path] (p22) -- (p23);
\path[draw, path] (p24) -- ($(p2) + (170pt, 0)$);

\path[draw, path] ($(p1.east) + (5pt,0)$) -- ($(p1) + (170pt, 0)$);

\path[draw, path]
($(p.east) + (1pt,0)$) edge [path, color=ForestGreen] ($(p.east) + (110pt, 0)$)
($(p.east) + (110pt, 0)$) edge [path, bend left=30, color=ForestGreen] (pl4)
(pl3) edge [dotted, color=red, line width = 1pt,  line width = 1pt] (pl4)
(pl1) edge [dotted, color=red, line width = 1pt,  line width = 1pt] (pl2)
(p31) -- (p32)
(p21) -- (p22)
(p23) -- (p24)
(pl3) edge [dashed, bend left=15, color=ForestGreen,  line width = 1pt] (p24)
(p23) edge [dashed, bend left=45, color=ForestGreen,  line width = 1pt] (p22)
(p21) edge [dashed, bend left=15, color=ForestGreen,  line width = 1pt] (p32)
(p31) edge [dashed, bend left=15, color=ForestGreen,  line width = 1pt] (pl2)
(p31) edge [dashed, bend left=15, color=ForestGreen,  line width = 1pt] (pl2)
(pl1) edge [->, bend left=15, color=ForestGreen] ($(pl1) + (10pt, 10pt)$);

\node[] at ($(pl1) + (0pt, -7pt)$) (lw) {$w$};
\node[] at ($(pl4) + (0pt, -7pt)$) (lu) {$u$};

\node[] at ($(p3) + (15pt, 11pt)$) {$\vdots$};
\node[] at ($(p3) + (170pt, 11pt)$) {$\vdots$};

\end{tikzpicture}
				\vspace{-0.65cm}
				\caption{$u, w$ are on same path $P \in \mathcal P$.}
                                \vspace{0.25cm}
 			\end{subfigure}%
			\begin{subfigure}{.5\textwidth}
				\centering
				\begin{tikzpicture}[
%  -{Stealth[length = 2.5pt]},
%start chain = going down,
node distance = 10pt,
vertex/.style={minimum width=0.30em, minimum height=0.30em, circle, draw, on 
chain},
svertex/.style={minimum width=0.15em, minimum height=0.15em, circle, draw, 
on chain, inner sep=1pt, text width=3pt, align=center},
fix/.style={inner sep=1pt, text width=6pt, align=center},
graph/.style={minimum width=3em, minimum height=2em, ellipse, draw, on 
chain, decorate, decoration={snake,segment length=1mm,amplitude=0.2mm}},
decoration={
	markings,
	mark=at position 0.5 with {\arrow{>}}},
path/.style=
	{-,
		decorate,
		decoration={snake,amplitude=.2mm,segment length=2mm,}}
],

\node[] (p) {$P_{\ell + 1}$};
\node[below= 2pt of p] (pl) {$P_{\ell'}$};
\node[below= 2pt of pl] (p3) {$P_3$};
\node[below= 2pt of p3] (p2) {$P_2$};
\node[below= 2pt of p2] (p1) {$P_1$};

\node[svertex, fill=black, right = 110pt of pl] (pl1) {};
\node[svertex, fill=black, right = 130pt of pl] (pl2) {};

\node[svertex, fill=black, right = 70pt of p3] (p31) {};
\node[svertex, fill=black, right = 90pt of p3, color=red] (p32) {};
\node[svertex, fill=black, right = 111pt of p3] (p33) {};

\node[svertex, fill=black, right = 50pt of p2] (p21) {};
\node[svertex, fill=black, right = 70pt of p2] (p22) {};

\node[svertex, fill=black, right = 30pt of p1] (p11) {};
\node[svertex, fill=black, right = 50pt of p1] (p12) {};

\node[svertex, fill=black, above= 12pt of pl2] (u) {};
\node[svertex, fill=black, below= 12pt of p11] (w) {};

\path[draw, path] 
($(pl.east) + (5pt,0)$) -- (pl1)
($(p3.east) + (5pt,0)$) -- (p31)
($(p2.east) + (5pt,0)$) -- (p21)
($(p1.east) + (5pt,0)$) -- (p11)
(pl2) -- ($(pl) + (170pt, 0)$)
(p33) -- ($(p3) + (171pt, 0)$)
(p22) -- ($(p2) + (170pt, 0)$)
(p12) -- ($(p1) + (170pt, 0)$)
;

\path[draw, dotted, color=red, line width = 1pt] 
(pl1) -- (pl2)
(p31) -- (p32)
(p32) -- (p33)
(p21) -- (p22)
(p11) -- (p12)
;

\path[draw, dashed, color=ForestGreen, line width = 1pt] 
(u) -- ($(p.east) + (145pt, 0)$)
($(w) - (15pt, 0)$) -- (w)
;

\path[draw, color=red, line width = 1pt] 
(pl1) edge [->] (p33)
(p31) edge [->] (p22)
(p21) edge [->] (p12)
(p11) edge [->] (w)
(u) edge [->] (pl2)
;

\path[draw, path, color=ForestGreen]
($(p.east) + (1pt,0)$) -- (u)
(w) -- ($(w) + (20pt, 0)$) edge [->, bend left=0] ($(w) + (25pt, 0)$)
;

\node[] at ($(u) + (6pt, -6pt)$) (lu) {$u$};
\node[] at ($(w) + (0pt, -7pt)$) (lw) {$w$};

\node[] at ($(p11) + (0pt, 7pt)$) (lx1) {$x_1$};
\node[] at ($(p12) + (0pt, -7pt)$) (ly1) {$y_1$};

\node[] at ($(p21) + (0pt, 7pt)$) (lx2) {$x_2$};
\node[] at ($(p22) + (0pt, -7pt)$) (ly2) {$y_2$};

\node[] at ($(p31) + (0pt, 7pt)$) (lx3) {$x_3$};
\node[] at ($(p33) + (0pt, -7pt)$) (ly3) {$y_3$};

\node[] at ($(pl1) + (0pt, 7pt)$) (lxl) {$x_{\ell'}$};
\node[] at ($(pl2) + (0pt, -7pt)$) (lyl) {$y_{\ell'}$};

\node[] at ($(p3) + (15pt, 11pt)$) {$\vdots$};
\node[] at ($(p3) + (170pt, 11pt)$) {$\vdots$};

\end{tikzpicture}
				\vspace{-0.65cm}
				\caption{$(w, x_1)$ and $(y_{\ell'}, u)$ are cross edges.}
                                \vspace{0.25cm}
			\end{subfigure}%
			
			\begin{subfigure}{.5\textwidth}
				\centering
				\begin{tikzpicture}[
%  -{Stealth[length = 2.5pt]},
%start chain = going down,
node distance = 10pt,
vertex/.style={minimum width=0.30em, minimum height=0.30em, circle, draw, on 
chain},
svertex/.style={minimum width=0.15em, minimum height=0.15em, circle, draw, 
on chain, inner sep=1pt, text width=3pt, align=center},
fix/.style={inner sep=1pt, text width=6pt, align=center},
graph/.style={minimum width=3em, minimum height=2em, ellipse, draw, on 
chain, decorate, decoration={snake,segment length=1mm,amplitude=0.2mm}},
decoration={
	markings,
	mark=at position 0.5 with {\arrow{>}}},
path/.style=
	{-,
		decorate,
		decoration={snake,amplitude=.2mm,segment length=2mm,}}
],

\node[] (p) {$P_{\ell + 1}$};
\node[below= 2pt of p] (pl) {$P_{\ell'}$};
\node[below= 2pt of pl] (p3) {$P_3$};
\node[below= 2pt of p3] (p2) {$P_2$};
\node[below= 2pt of p2] (p1) {$P_1$};

\node[svertex, fill=black, right = 110pt of pl] (pl1) {};
\node[svertex, fill=black, right = 130pt of pl, color=red] (pl2) {};

\node[svertex, fill=black, right = 70pt of p3] (p31) {};
\node[svertex, fill=black, right = 90pt of p3, color=red] (p32) {};
\node[svertex, fill=black, right = 111pt of p3] (p33) {};

\node[svertex, fill=black, right = 50pt of p2] (p21) {};
\node[svertex, fill=black, right = 70pt of p2] (p22) {};

\node[svertex, fill=black, right = 30pt of p1] (p11) {};
\node[svertex, fill=black, right = 50pt of p1] (p12) {};

\node[svertex, fill=black, right= 145pt of pl] (u) {};
\node[svertex, fill=black, below= 12pt of p11] (w) {};

\path[draw, path] 
($(pl.east) + (5pt,0)$) -- (pl1)
($(p3.east) + (5pt,0)$) -- (p31)
($(p2.east) + (5pt,0)$) -- (p21)
($(p1.east) + (5pt,0)$) -- (p11)
(u) -- ($(pl) + (170pt, 0)$)
(p33) -- ($(p3) + (171pt, 0)$)
(p22) -- ($(p2) + (170pt, 0)$)
(p12) -- ($(p1) + (170pt, 0)$)
;

\path[draw, dotted, color=red, line width = 1pt] 
(pl1) -- (pl2)
(p31) -- (p32)
(p32) -- (p33)
(p21) -- (p22)
(p11) -- (p12)
(u) -- (pl2)
;

\path[draw, dashed, color=ForestGreen, line width = 1pt]
($(w) - (15pt, 0)$) -- (w)
(pl2) edge [dashed, bend right=15] ($(pl2) + (10pt, -10pt)$)
;

\path[draw, color=red, line width = 1pt] 
(pl1) edge [->] (p33)
(p31) edge [->] (p22)
(p21) edge [->] (p12)
(p11) edge [->] (w)
;

\path[draw, path, color=ForestGreen]
($(p.east) + (1pt,0)$) -- ($(p.east) + (130pt, 0pt)$) edge [path, bend left = 25] (u)
(w) -- ($(w) + (20pt, 0)$) edge [->, bend left=0] ($(w) + (25pt, 0)$)
;

\node[] at ($(u) + (0pt, -7pt)$) (lu) {$u$};
\node[] at ($(w) + (0pt, -7pt)$) (lw) {$w$};

\node[] at ($(p11) + (0pt, 7pt)$) (lx1) {$x_1$};
\node[] at ($(p12) + (0pt, -7pt)$) (ly1) {$y_1$};

\node[] at ($(p21) + (0pt, 7pt)$) (lx2) {$x_2$};
\node[] at ($(p22) + (0pt, -7pt)$) (ly2) {$y_2$};

\node[] at ($(p31) + (0pt, 7pt)$) (lx3) {$x_3$};
\node[] at ($(p33) + (0pt, -7pt)$) (ly3) {$y_3$};

\node[] at ($(pl1) + (0pt, 7pt)$) (lxl) {$x_{\ell'}$};
%\node[] at ($(pl2) + (0pt, -7pt)$) (lyl) {$y_{\ell'}$};

\node[] at ($(p3) + (15pt, 11pt)$) {$\vdots$};
\node[] at ($(p3) + (170pt, 11pt)$) {$\vdots$};

\end{tikzpicture}
				\vspace{-0.65cm}
				\caption{$(w, x_1)$ is a cross edge and $y_{\ell'} = u$.}
                                \vspace{0.25cm}
			\end{subfigure}%
			\begin{subfigure}{.5\textwidth}
				\centering
				\begin{tikzpicture}[
%  -{Stealth[length = 2.5pt]},
%start chain = going down,
node distance = 10pt,
vertex/.style={minimum width=0.30em, minimum height=0.30em, circle, draw, on 
chain},
svertex/.style={minimum width=0.15em, minimum height=0.15em, circle, draw, 
on chain, inner sep=1pt, text width=3pt, align=center},
fix/.style={inner sep=1pt, text width=6pt, align=center},
graph/.style={minimum width=3em, minimum height=2em, ellipse, draw, on 
chain, decorate, decoration={snake,segment length=1mm,amplitude=0.2mm}},
decoration={
	markings,
	mark=at position 0.5 with {\arrow{>}}},
path/.style=
	{-,
		decorate,
		decoration={snake,amplitude=.2mm,segment length=2mm,}}
],

\node[] (p) {$P_{\ell + 1}$};
\node[below= 2pt of p] (pl) {$P_{\ell'}$};
\node[below= 2pt of pl] (p3) {$P_3$};
\node[below= 2pt of p3] (p2) {$P_2$};
\node[below= 2pt of p2] (p1) {$P_1$};

\node[svertex, fill=black, right = 110pt of pl] (pl1) {};
\node[svertex, fill=black, right = 130pt of pl] (pl2) {};

\node[svertex, fill=black, right = 70pt of p3] (p31) {};
\node[svertex, fill=black, right = 90pt of p3, color=red] (p32) {};
\node[svertex, fill=black, right = 111pt of p3] (p33) {};

\node[svertex, fill=black, right = 50pt of p2] (p21) {};
\node[svertex, fill=black, right = 70pt of p2] (p22) {};

\node[svertex, fill=black, right = 30pt of p1] (p11) {};
\node[svertex, fill=black, right = 50pt of p1] (p12) {};

\node[svertex, fill=black, above= 12pt of pl2] (u) {};
\node[svertex, fill=black, right= 15pt of p1] (w) {};

\path[draw, path] 
($(pl.east) + (5pt,0)$) -- (pl1)
($(p3.east) + (5pt,0)$) -- (p31)
($(p2.east) + (5pt,0)$) -- (p21)
($(p1.east) + (5pt,0)$) -- (w)
(pl2) -- ($(pl) + (170pt, 0)$)
(p33) -- ($(p3) + (171pt, 0)$)
(p22) -- ($(p2) + (170pt, 0)$)
(p12) -- ($(p1) + (170pt, 0)$)
;

\path[draw, dotted, color=red, line width = 1pt] 
(pl1) -- (pl2)
(p31) -- (p32)
(p32) -- (p33)
(p21) -- (p22)
(p11) -- (p12)
(w) -- (p11)
;

\path[draw, dashed, color=ForestGreen, line width = 1pt]
(u) -- ($(p.east) + (145pt, 0)$)
($(w) - (0pt, 30pt)$) edge [bend right = 25] (p11)
;

\path[draw, color=red, line width = 1pt] 
(pl1) edge [->] (p33)
(p31) edge [->] (p22)
(p21) edge [->] (p12)
%(p11) edge [->] (w)
(u) edge [->] (pl2)
;

\path[draw, path, color=ForestGreen]
($(p.east) + (1pt,0)$) -- (u)
(w) edge [->, bend left=15] ($(w) + (10pt, 10pt)$)
;

\node[] at ($(u) + (6pt, -6pt)$) (lu) {$u$};
\node[] at ($(w) + (0pt, -7pt)$) (lw) {$w$};

%\node[] at ($(p11) + (0pt, 7pt)$) (lx1) {$x_1$};
\node[] at ($(p12) + (0pt, -7pt)$) (ly1) {$y_1$};

\node[] at ($(p21) + (0pt, 7pt)$) (lx2) {$x_2$};
\node[] at ($(p22) + (0pt, -7pt)$) (ly2) {$y_2$};

\node[] at ($(p31) + (0pt, 7pt)$) (lx3) {$x_3$};
\node[] at ($(p33) + (0pt, -7pt)$) (ly3) {$y_3$};

\node[] at ($(pl1) + (0pt, 7pt)$) (lxl) {$x_{\ell'}$};
\node[] at ($(pl2) + (0pt, -7pt)$) (lyl) {$y_{\ell'}$};

\node[] at ($(p3) + (15pt, 11pt)$) {$\vdots$};
\node[] at ($(p3) + (170pt, 11pt)$) {$\vdots$};

\end{tikzpicture}
				\vspace{-0.65cm}
				\caption{$(y_{\ell'}, u)$ is a cross edge and $x_1 = w$.}
                                \vspace{0.25cm}
			\end{subfigure}%	

			\begin{subfigure}{1\textwidth}
				\begin{tikzpicture}[
%  -{Stealth[length = 2.5pt]},
%start chain = going down,
node distance = 10pt,
vertex/.style={minimum width=0.30em, minimum height=0.30em, circle, draw, on 
chain},
svertex/.style={minimum width=0.15em, minimum height=0.15em, circle, draw, 
on chain, inner sep=1pt, text width=3pt, align=center},
fix/.style={inner sep=1pt, text width=6pt, align=center},
graph/.style={minimum width=3em, minimum height=2em, ellipse, draw, on 
chain, decorate, decoration={snake,segment length=1mm,amplitude=0.2mm}},
decoration={
	markings,
	mark=at position 0.5 with {\arrow{>}}},
path/.style=
	{-,
		decorate,
		decoration={snake,amplitude=.2mm,segment length=2mm,}}
],

\node[] (p) {$P_{\ell + 1}$};
\node[below= 2pt of p] (pl) {$P_{\ell'}$};
\node[below= 2pt of pl] (p3) {$P_3$};
\node[below= 2pt of p3] (p2) {$P_2$};
\node[below= 2pt of p2] (p1) {$P_1$};

\node[svertex, fill=black, right = 110pt of pl] (pl1) {};
\node[svertex, fill=black, right = 130pt of pl, color=red] (pl2) {};
\node[svertex, fill=black, right = 69pt of pl] (pl0) {};

\node[svertex, fill=black, right = 70pt of p3] (p31) {};
\node[svertex, fill=black, right = 90pt of p3, color=red] (p32) {};
\node[svertex, fill=black, right = 111pt of p3] (p33) {};

\node[svertex, fill=black, right = 50pt of p2] (p21) {};
\node[svertex, fill=black, right = 70pt of p2] (p22) {};

\node[svertex, fill=black, right = 30pt of p1] (p11) {};
\node[svertex, fill=black, right = 50pt of p1] (p12) {};

\node[svertex, fill=black, right= 145pt of pl] (u) {};
\node[svertex, fill=black, right= 15pt of p1] (w) {};

\draw[-] (pl0) -- (p31);

\path[draw, path] 
($(pl.east) + (5pt,0)$) -- (pl1)
($(p3.east) + (5pt,0)$) -- (p31)
($(p2.east) + (5pt,0)$) -- (p21)
($(p1.east) + (5pt,0)$) -- (w)
(u) -- ($(pl) + (170pt, 0)$)
(p33) -- ($(p3) + (171pt, 0)$)
(p22) -- ($(p2) + (170pt, 0)$)
(p12) -- ($(p1) + (170pt, 0)$)
;

\path[draw, dotted, color=red, line width = 1pt] 
(pl1) -- (pl2)
(p31) -- (p32)
(p32) -- (p33)
(p21) -- (p22)
(p11) -- (p12)
(w) -- (p11)
(u) -- (pl2)
;

\path[draw, dashed, color=ForestGreen , line width = 1pt]
(pl2) edge [dashed, bend right=15] ($(pl2) + (10pt, -10pt)$)
($(w) - (0pt, 30pt)$) edge [bend right = 25] (p11)
;

\path[draw, color=red, line width = 1pt] 
(pl1) edge [->] (p33)
(p31) edge [->] (p22)
(p21) edge [->] (p12)
%(p11) edge [->] (w)
%(u) edge [->] (pl2)
;

\path[draw, path, color=ForestGreen]
($(p.east) + (1pt,0)$) -- ($(p.east) + (130pt, 0pt)$) edge [path, bend left = 25] (u)
(w) edge [->, bend left=15] ($(w) + (10pt, 10pt)$)
;

\node[] at ($(u) + (6pt, -6pt)$) (lu) {$u$};
\node[] at ($(w) + (0pt, -7pt)$) (lw) {$w$};

%\node[] at ($(p11) + (0pt, 7pt)$) (lx1) {$x_1$};
\node[] at ($(p12) + (0pt, -7pt)$) (ly1) {$y_1$};

\node[] at ($(p21) + (0pt, 7pt)$) (lx2) {$x_2$};
\node[] at ($(p22) + (0pt, -7pt)$) (ly2) {$y_2$};

\node[] at ($(p31) + (-6pt, 7pt)$) (lx3) {$x_3$};
\node[] at ($(p33) + (0pt, -7pt)$) (ly3) {$y_3$};

\node[] at ($(pl1) + (0pt, 7pt)$) (lxl) {$x_{\ell'}$};
%\node[] at ($(pl2) + (0pt, -7pt)$) (lyl) {$y_{\ell'}$};

\node[] at ($(p3) + (15pt, 11pt)$) {$\vdots$};
\node[] at ($(p3) + (170pt, 11pt)$) {$\vdots$};

\end{tikzpicture}
				\vspace{-0.65cm}
				\caption{Neither $(w, x_1)$ nor $(y_{\ell'}, u)$ is a cross
				edge.}%
                                \vspace{0.25cm}
			\end{subfigure}%	
			
			\caption{Sketch of rerouted paths to remove common vertices
and extended deadlock cycles.
The dotted and dashed subpaths are removed from the paths and the red lines with an arrow
show a switch to another path.
}\label{fig:rerouting}
\end{figure}
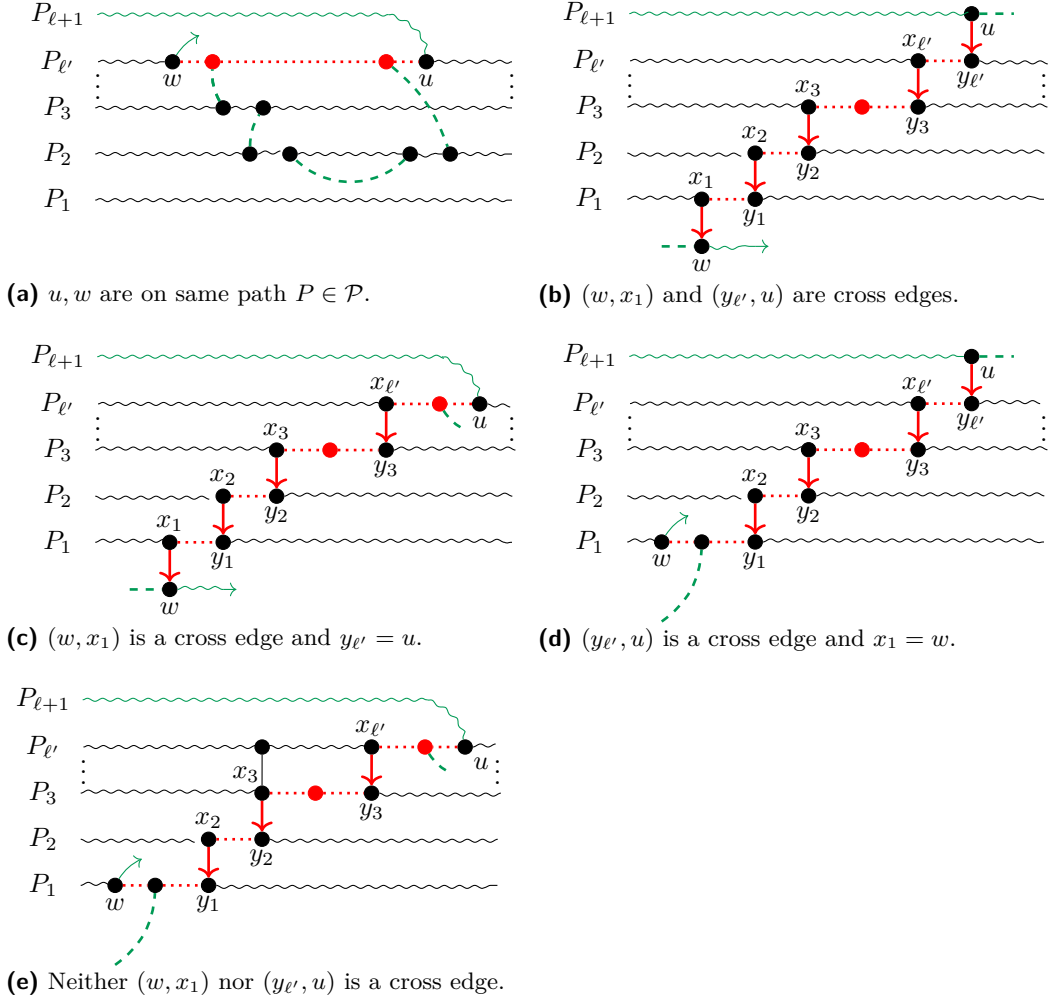

		{\bfseries Correctness.}
                One can see that our algorithm discovers every common
subpath of $P^*$ with a path $P^c \in \mathcal P$ as well as every extended
deadlock cycle since, in both cases, we have a vertex $u$ followed by $w$ on $P^*$
such that $w$ is discovered from $u$ while running backwards on~$\mathcal
P$. Since $V^*$ only shrinks (no new cross edges can occur), since
by construction 
the new paths have no common vertices accept $s$ and $t$ (here it
is important that the constructed $u$-$w$ path has no subsequent cross
edges) and 
since we choose $w$ as the latest vertex of $P^*$ touching $\Q'$,
$\Q$ maintains a clean area when adding $\Q'$ to it.
Note that each vertex with a rerouting information becomes part of $\Q$.
Thus, the size of $\Q$ increases by $\Omega(n / \log k)$ when our algorithm stops unless we reach $t$
while running on $P^*$.
		
		{\bfseries Efficiency.}
		Concerning the running time, 
        note that we
        travel along $P^*$ as well as construct
		a set $\Q'$ with a DFS
		on the paths in $\mathcal P$. 
		Note that we process all vertices of $\Q'$ at most once 
		by each of the two standard DFS runs
		(once to construct $\Q'$ as well as once to find a $u$-$w$ path 
		for computing the rerouting)
		before $\Q'$ is added into $\Q$. Within the processing of each vertex
		$v$, we call $\op{prev}(v)$ only once.
        Since each access to a vertex $v$ 
                can be realized in $O(\deg(v) + k^3 \log^3 k)$ time
                (Lemma~\ref{lem:recompute}),
		the total running time is $O(n k^3 \log^3 k)$.

               Concerning the space consumption observe that $\Rs$ and
$\Rp$, a backup of $\Q$ and $p$
as well as all arrays use $O(n)$ bits. To run the DFS,
$O(n + \ell \log n)$ bits suffice 
(storing the colors white, grey and black in $O(n)$ bits and $O(\ell \log n)$ bits
for the DFS stack since we only have to store the endpoints of $O(\ell)$ subpaths that are 
part of the stack).
Together with the space used by Lemma~\ref{lem:recompute}, our algorithm 
can run with $O(n + k^2 (\log k) \log n)$ bits.
\end{proof}
 
We can not store our rerouted paths directly in a valid path data scheme.
%To compute one we want to follow the paths with respect to a rerouting~$R$.
To follow the path with respect to the rerouting we next
provide a so-called weak path data structure for
all $s$-$t$ paths of $R(\mathcal P \cup \{P'\})$.

Recall that we have a path data scheme storing
the single path $P^*$ 
and a path data scheme storing the set~$\mathcal P$
of paths.
Let $\Q$, $p$ as well as 
$R$ (realized by $(V^*, \Rs, \Rp)$) be returned by the previous lemma.
Our valid path data scheme
can only store $s$-$t$-paths, but $R(\mathcal P \cup \{P'\})$ 
may contain an $s$-$p'$ path $P''$ with $p' \ne t$.
In the case that $p' \ne t$, we compute
a path numbering and a valid path data scheme
for the path $P^{**}$ obtained from
$P^*$ by replacing the subpath $P'$
by $P''$. (The path $P^{**}$ replaces the path $P^*$
in the computation of the next batch.)
Furthermore, we remove the path $P''$ from $R(\mathcal P \cup \{P'\})$
by removing its vertices from $V^*$.
It is easy to see that this 
modification of the paths can be done in the same bounds
as stated in Lemma~\ref{lem:rerouting}.

After the modification of the paths we show how to follow the paths
with respect to $R$.
Since we can find out in $O(1)$ time if a vertex $u$ belongs to $P^*$, a call of $\op{prev}$
and $\op{next}$ below is forwarded to the correct path data scheme.
We now
overload the $\op{prev}$ and $\op{next}$ operations
of our path data structure as follows, and so get a {\em weak path data structure}
for $R(\mathcal P \cup \{P'\}) \setminus \{P''\}$
that supports $\op{prev'}$ and $\op{next'}$ in time $O(\op{deg}(v) + k^3 \log^3 k)$
by using $O(n + k^2 (\log k) \log n)$ bits and that supports access to $V^*$ in constant time.
To obtain the runtime we realize
$\op{prev'}$ and $\op{next'}$
by computing the paths within at most 
one region,
 from which we know the color
of all neighbors of~$u$.
In detail, the weak path data structure supports the following operations where $u \in V^*$.

	\begin{itemize}
		\item $\op{prev'}(u)$:
		If $\Rp(u) = 0$, return $\op{prev}(u)$.
		Else, return $v \in N(u) \cap V^*$ with $\op{color}(v) = \Rp(u)$.
		\item $\op{next'}(u)$: 
		If $\Rs(u) = 0$, return $\op{next}(u)$.
		Else, return $v \in N(u) \cap V^*$ with $\op{color}(v) = \Rs(u)$.
		\item \op{inVStar}$(u)$: Returns true exactly if $u \in V^*$.
	\end{itemize}

%In the case that $\CP$ is an $s$-$t$ path, i.e., $p = t$, we
%compute a valid path data scheme storing
%$\ell + 1$ good $s$-$t$ paths.
%Otherwise, we replace the $s$-$p$ subpath~$P'$ of~$P^*$ by~$\CP$
%and compute a valid path data scheme for the path obtained as well as for the $\ell$ good $s$-$t$ paths $R(\mathcal P \cup \{P'\}) \setminus \CP$.
%In both cases, after having the new data scheme(s),
%we do not require $V^*$, $\Rs$, and $\Rp$ anymore.

To compute a valid path data scheme for $1 < \ell \le k$ good $s$-$t$-paths we have to partition $G[V^*]$ into regions. 
More exactly we have to find a set $\B' \subseteq V^*$
of boundary vertices such that $G[V^* \setminus \B']$ consists
of regions (i.e., connected components) of size $O(k \log k)$.
To find the boundary we start with the set~$S$ of neighbors of $s$ that belong to $V^*$
and move $S$ to the right ``towards'' $t$ such that $S$ 
always is a separator that
disconnects $s$ and $t$ in $G[V^*]$.
Whenever there are $O(k \log k)$ vertices ``behind'' the previous separator that was added to our boundary,
we add $S$ to our boundary---see also Fig.~\ref{fig:compute-storage-scheme}.
While moving $S$ towards $t$, we always make sure that the endpoints of a cross edge are not in different components in $G[V^* \setminus S]$.
This means, we do not move $S$ over an endpoint of a cross edge unless the other endpoint
is also in $S$.
We show in the next lemma that this is
possible if we have no extended deadlock cycles and if we use the following trick for
simple deadlock cycles: move $S$ one vertex further along the path for all vertices part of the simple deadlock cycle.

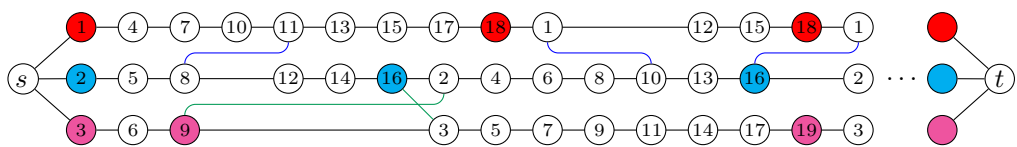
\begin{figure}[h]
			\centering
			\begin{tikzpicture}[
%  -{Stealth[length = 2.5pt]},
start chain = going right,
node distance = 8pt,
vertex/.style={minimum width=1.1em, minimum height=1.1em, circle, draw, 
on chain, inner sep=0pt, text width=8pt, align=center, fill=white},
fix/.style={inner sep=1pt, text width=6pt, align=center},
graph/.style={minimum width=3em, minimum height=2em, ellipse, draw, on 
chain, decorate, decoration={snake,segment length=1mm,amplitude=0.2mm}},
decoration={
	markings,
	mark=at position 0.5 with {\arrow{>}}},
path/.style=
	{-}
],

\node[vertex, fix] (s) {$s\phantom{t}$};

\node[vertex, fill=red] at ($(s) + (8pt, 19.5pt)$) (a1) {\scriptsize{1}}; 			% 1
\node[vertex, fill=cyan, below = 8pt of a1] (b1) {\scriptsize2}; %2
\node[vertex, fill=Rhodamine, below = 8pt of b1] (c1) {\scriptsize3}; %3

\node[vertex, right = 8pt of a1] (a2) {\scriptsize4}; % 4
\node[vertex, below = 8pt of a2] (b2) {\scriptsize5};
\node[vertex, below = 8pt of b2] (c2) {\scriptsize6};

\node[vertex, right = 8pt of a2] (a3) {\scriptsize7}; 
\node[vertex, below = 8pt of a3] (b3) {\scriptsize8}; % Kante
\node[vertex, below = 8pt of b3, fill=Rhodamine] (c3) {\scriptsize9}; % Kante

\node[vertex, right = 8pt of a3] (a4) {\scriptsize10};

\node[vertex, right = 8pt of a4] (a5) {\scriptsize11};
\node[vertex, below = 8pt of a5] (x1) {\scriptsize12};

\path [-, draw, blue, rounded corners]
(b3.north) -- ++(0, 4pt) -- ($(a5.south) - (0, 4pt)$) -- (a5.south);

\node[vertex, right = 8pt of a5] (a6) {\scriptsize13};
\node[vertex, below = 8pt of a6] (a7) {\scriptsize14};

\node[vertex, right = 8pt of a6] (a8) {\scriptsize15};
\node[vertex, below = 8pt of a8, fill=cyan] (a9) {\scriptsize16};

\node[vertex, right = 8pt of a8] (a10) {\scriptsize17};
\node[vertex, below = 8pt of a10] (a11) {\scriptsize2}; % Kante
\node[vertex, below = 8pt of a11] (a12) {\scriptsize3}; % Kante

\path [-, draw, ForestGreen] (a9) -- (a12);
%\path [-, draw, blue, rounded corners]
%(a9.south) -- ++(0, -3pt) -- ($(a12.north) + (0, 3pt)$) -- (a12.north);

\path [-, draw, ForestGreen, rounded corners]
(c3.north) -- ++(0, 4pt) -- ($(a11.south) - (0, 4pt)$) -- (a11.south);

\node[vertex, right = 8pt of a10, fill=red] (a13) {\scriptsize18};
\node[vertex, below = 8pt of a13] (a14) {\scriptsize4};
\node[vertex, below = 8pt of a14] (a15) {\scriptsize5};

\node[vertex, right = 8pt of a13] (a16) {\scriptsize1}; % Kante
\node[vertex, below = 8pt of a16] (a17) {\scriptsize6};
\node[vertex, below = 8pt of a17] (a18) {\scriptsize7};

\node[vertex, right = 8pt of a17] (a19) {\scriptsize8};
\node[vertex, below = 8pt of a19] (a20) {\scriptsize9};

\node[vertex, right = 8pt of a19] (a21) {\scriptsize10}; % Kante
\node[vertex, below = 8pt of a21] (a22) {\scriptsize11};

\path [-, draw, blue, rounded corners]
(a16.south) -- ++(0, -4pt) -- ($(a21.north) + (0, 4pt)$) -- (a21.north);

\node[vertex, right = 8pt of a21] (a23) {\scriptsize13};
\node[vertex, above = 8pt of a23] (a24) {\scriptsize12};
\node[vertex, below = 8pt of a23] (a25) {\scriptsize14};

\node[vertex, right = 8pt of a24] (a26) {\scriptsize15};
\node[vertex, fill=cyan, below = 8pt of a26] (a27) {\scriptsize16}; % Kante
\node[vertex, below = 8pt of a27] (a28) {\scriptsize17};

\node[vertex, fill=red, right = 8pt of a26] (a29) {\scriptsize18};
\node[vertex, fill=Rhodamine, right = 8pt of a28] (a30) {\scriptsize19};

\node[vertex, right = 8pt of a29] (a31) {\scriptsize1}; % Kante
\node[vertex, below = 8pt of a31] (a32) {\scriptsize2};
\node[vertex, below = 8pt of a32] (a33) {\scriptsize3};

\path [-, draw, blue, rounded corners]
(a27.north) -- ++(0, 4pt) -- ($(a31.south) - (0, 4pt)$) -- (a31.south);

\node[vertex, fill=red, right = 20pt of a31] (a34) {};
\node[vertex, fill=cyan, below = 8pt of a34] (a35) {};
\node[vertex, fill=Rhodamine, below = 8pt of a35] (a36) {};

\begin{scope}[on background layer]
	\path [draw, path](a1) -- (a31);
	\path [draw, path](b1) -- (a32);
	\path [draw, path](c1) -- (a33);
\end{scope}

%\node[vertex, below = 5pt of a20, fill=Rhodamine] (a30) {};			%3

\node[vertex, right = 354pt of s] (t) {$t$};

\path [draw, path](s) -- (a1) (s) -- (b1) (s) -- (c1);
\path [draw, path](t) -- (a34) (t) -- (a35) (t) -- (a36);

\node[left = 20pt of t] (d) {\scalebox{1}{$\dots$}};

%\path [draw, path](a10) -- (a11);
\end{tikzpicture}
			\caption{
			Our algorithm determines a boundary by moving from $s$ along all paths with respect to the cross edges (colored edges), i.e.,
			it cannot move from a vertex that is an endpoint of a cross edge if the other endpoint is not
			yet reached.
			In the example, we assume that $k \log k = 18$.
			The number in the vertices denotes the seen vertices since adding vertices to the boundary the last time.
			}\label{fig:compute-storage-scheme}
\end{figure}

In the case that the previous lemma returned paths containing
an $s$-$p'$ path with $p' \ne t$, then our weak 
path data structure gives access to $\ell$
good $s$-$t$ paths. Otherwise, $p' = t$ and
we have access to $\ell := \ell + 1$ good $s$-$t$ paths.

\begin{lemma}\label{lem:storageScheme}
	Given a weak path data structure for accessing $\ell$ good
	$s$-$t$ paths $\mathcal P'$ 
%	and
%	for an $s$-$t$ path $P^*$ that contains a clean $s$-$p$ subpath $\CP$ with respect to $\mathcal P'$,
	there is an algorithm that computes a valid path data scheme storing $\ell$ good $s$-$t$ paths.
%	and a valid path data scheme for $P^*$ if $p \ne t$, and
%		a valid path data scheme storing $\ell + 1$ good $s$-$t$ paths if $p = t$.
	The algorithm runs
	in %$O(n\ell^2 k^2 \log^2 k)$ 
	$O(n k^3 \log^3 k)$ time
%\info{Check, please. In der Laufzeit kam ein $\ell$ als Faktor for. Warum?}
	and uses $O(n + k^2 (\log k) \log n)$ bits.
\end{lemma}
\begin{proof}
	Let $V^*$ be the vertices of $P^*$ and $\mathcal P'$ obtained
	from the weak path data structure.
	As in the last proof, we assume that the paths in $\mathcal P'$
	are colored from $1, \ldots, \ell$.
	For the case that $p' \ne t$, we compute a new path data scheme for 
	$P^*$ by Lemma~\ref{lem:pathdatastructurefromsinglepath}.
	Moreover, we remove the vertices of $P^*$ from $V^*$
	and start to compute our path data scheme for~$\mathcal P'$.
	Our algorithm works in three steps.
	First we compute the boundary,
	then a path data scheme, and finally make it valid.

	{\bfseries Compute the boundary.} 
	Start with the set $S$ of all $|\mathcal P'|$ neighbors $u$ of $s$ in $G[V^*]$.
	Note that the vertices in $S$ are candidates for 
	a boundary because~$S$ 
	has the property that it cuts all $s$-$t$-paths;
	in other words,
	there is no path in $G[V^* \setminus S]$ that connects $s$ and $t$.
	For performance reasons, we cache $\op{next}'(u)$ for each vertex
	$u$ in $S$ as long as $u$ is in $S$ so that we have to compute it
	only once.
	From $S$ we can compute the next set $S'$ of candidates for a boundary.
	The idea is to try to put for each vertex $u \in S$ the vertex
	$v = \op{next}'(u)$ of the same path into $S'$.
	However, we cannot simply do this if $u$ has a cross edge
	that can bypass $S'$. In that case we take $u$ into $S'$.
	If $S = S'$, then one (or multiple) deadlock cycles restrain us to add new vertices to $S'$.
	We are either able to take the next vertex of every vertex of $S$ 
	that is part of a a simple deadlock cycle into $S'$ or show that
	an extended deadlock cycle exists, which is a contradiction to our property that
	our paths are good.
	
	We now focus on the structures that we use in the subsequent algorithm.
	Let $D$ (initially $D := S \cup \{s\}$) be a set 
	consisting of all vertices of the connected
	component of $G[V^* \setminus S]$ that contains $s$ as well as the vertices of $S$.
	(Intuitively speaking, $D$ is the set of vertices for which
	the membership in a boundary was already decided.)
	We associate the index of a vertex with its path membership
	(not its color since a weak path data scheme does not provide color information).
	Technically we realize $S$ and $S'$
	as arrays of $|\mathcal P'|$ fields
	such that $S[i]$ and $S'[i]$ store vertices $u_i$ and $v_i$, respectively, of $P_i \in \mathcal P'$ ($1 \le i \le |\mathcal P'|$).
	Moreover, for each $u_i \in S$ on path $P_i \in \mathcal P'$ 
 %       \blue{and for $v_i=\op{next}'(u_i)$,}
	we store the number $\op{d}(i) = |(N(u_i) \cap V^*) \setminus D|$ 
 %        \blue{and $\op{d}'(i) = |(N(v_i) \cap V^*) \setminus D|$}
	in an array of $|\mathcal P'|$ fields.
	These numbers represents the number of edges whose head does point at a
	vertex outside of $D$. If $\op{d}(i) \ne 1$, then $u_i$ has a cross edge
	to some vertex not in $D$ and we cannot simply put $v_i := \op{next}'(u_i)$ inside $S'$,
	because we might get a cross edge 
	that can be used by an $s$-$t$ path in $G[V^* \setminus S']$.
	Otherwise, if $\op{d}(i) = 1$, then $u_i$ has no  cross edge 
	possibly bypassing $S'$ and we can take $v_i$ into $S'$.
	Let $\B'$ be an initially empty set of boundary vertices.
	To know when the next set $S$ can be added to the boundary $\B'$ such that
	the property that the boundary defines regions of size $O(k \log k)$ holds,
	we use a counter $c$ (initially $c := 0$) to 
	count the number of vertices
	that we have seen until extending the boundary 
	the last time.
	
	We now show our algorithm to compute $S'$ given $S$.
	Take $\B' := S$, initialize $\op{d}(i)$ as $|(N(u_i) \cap V^*) \setminus D|$
	for each $u_i \in S$, and compute a set $S'$ in phases as follows.
	For each $u_i \in S$, if $\op{d}(i) = 1$,
	take $v_i = \op{next'}(u_i)$ into $S'$, increment $c$ by one,
	and set $\op{d}(i) = |(N(v_i) \cap V^*) \setminus D|$.
	Moreover, for each $w_j \in (N(v_i) \cap S )$ decrement $\op{d}(j)$ by one.
%        \blue{Similarily, update $\op{d}'(j)$.}
	If $\op{d}(i) \ne 1$, take $u$ into $S'$.
	If $S \ne S'$, add $S'$ into $D$, set $S := S'$ and repeat the phase.
	
	Otherwise (i.e., $S = S'$) we run into the {\em special case} where we have encountered a deadlock cycle (possible a combination of deadlock cycles),
	which we can handle as follows.
	Since $\mathcal P'$ is good and $\CP$ is clean, we have one or more simple deadlock cycles
	with vertices in $S$.
	Set $c = c + |\mathcal P'|$ and take $\til{S} = \{\op{next'}(u)\ |\ u \in S\}$.
	Note that all these simple deadlock cycles have only vertices in~$S \cup \til{S}$, but
	not all vertices of $\til{S}$ have to be part of a simple deadlock cycle.
	To compute $S'$ from $\til{S}$, we have to remove the vertices from~$\til{S}$ 
	that are not part of the simple deadlock cycle, which we can do as follows.
	Take $S' := \til{S}$.
	Check in rounds
	if~$S'$ separates $S$ from $t$, i.e.,
	for each $v_i \in S'$ with $1 \le i \le |\mathcal P'|$, 
        choose $u_i = S$ and check if $\op{d}(j)>\ell$. If not,
	additionally check 
        if a $w_j \in (N(u_i) \cap V^*) \setminus (D \cup \{S'[j]\} )$ 
        exists. If one of the check passes, some edge of $u_i$ bypasses $S'$
	and we replace $v_i$ by $u_i$ in $S'$ and decrement $c$ by one. 
	Then repeat the round until no such $w$ can be found for any vertex.
	Afterwards, set $S'$ as the new $S$, add $S$ into $D$, and compute 
	the $d$-values for the new vertices in $S$, 
	and proceed with the next phase.
	
	If $c > k \log k$ after adding $S$ into $D$, add the vertices of $S$ to~$\B'$,
	and reset $c := 0$.	
	Repeat the algorithm until $S = N(t) \cap V^*$ and add the vertices of $S$ to $\B'$.
	Finally, add all vertices of $V^*$ with large degree to $\B'$.
	
	{\bfseries Correctness of the boundary.}
        We show the following invariant whenever we update $S$: 
        the removal of $S$ in $G[V^*]$ disconnects $s$
        and $t$. Initially, this is true for the 
        neighbors of $s$ 
	because their removal disconnects each chordless path in $G[V^*]$.
	If we have a vertex $u \in S$ on $P_i \in \mathcal P$ with only one neighbor in $G[V^* \setminus D]$, i.e., we have $\op{d}(i) = 1$,
	then $S' = (S \setminus \{u\}) \cup \{\op{next'}(u)\}$ is again a valid separator.
	If we have no such vertex~$u$, then we construct a set $S'$ and we
        check for each vertex in $S$ if it has a cross edge that destroys $S'$ being
        a separator. 
        Thus, we set $S:=S'$ only if $S'$ is again a separator.

        We next show that
          we make progress with each separator $S$ over all phases, i.e., we find a
          new separator so that $D$ increases.
          For a simpler notation, let us say that a vertex on a path is
           {\em behind} another vertex on the path if the latter is closer to $s$.
          As long as we do not run into the special case having a deadlock cycle, we can clearly make progress. 
          Let $S'$ be defined as in the beginning of the special case, i.e.,
          $S'$ consists of the successors
          on the paths of each vertex in $S$.
          Let us consider a deadlock cycle $Z$ whose subpaths all
          begin with a vertex 
          in $S$. Let $S^Z$ be the endpoints of the subpaths. 
          Choose $Z$ such that  
          no cross edge of a vertex $v \in S^Z$ ends at a vertex $w$ on a path that is part of another deadlock cycle
          $Z'$
          such that $w$ is behind the subpath part of $Z'$.

          Assume for a contradiction that no such $Z$ exists. Then we must
          have a cyclic sequence of $x\ge 2$ deadlock cycles where each deadlock cycle has a vertex in $S$
           with a cross
          edge $e'$ to a vertex of the next deadlock cycle
          (Fig.~\ref{fig:deadlock}a shows an example for $x=2$). But then we can use
          the cross edges to build an extended deadlock cycle. 
           (Fig.~\ref{fig:deadlock}b).

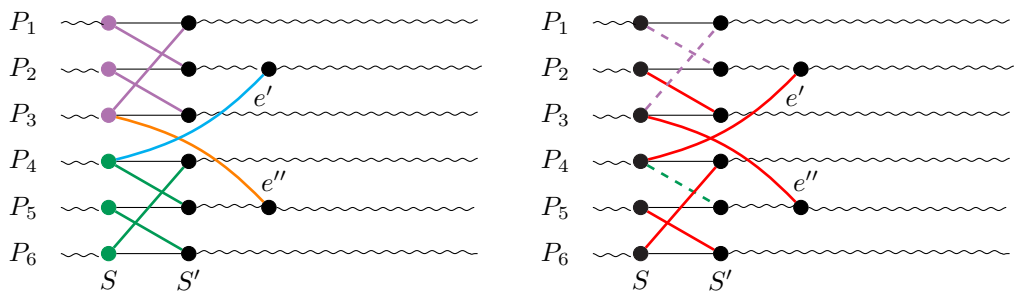
\begin{figure}[t]
			\centering
			\begin{subfigure}{.5\textwidth}
				\centering
				\begin{tikzpicture}[
%  -{Stealth[length = 2.5pt]},
%start chain = going down,
node distance = 10pt,
vertex/.style={minimum width=0.30em, minimum height=0.30em, circle, draw, on 
chain},
svertex/.style={minimum width=0.15em, minimum height=0.15em, circle, draw, 
on chain, inner sep=1pt, text width=3pt, align=center},
fix/.style={inner sep=1pt, text width=6pt, align=center},
graph/.style={minimum width=3em, minimum height=2em, ellipse, draw, on 
chain, decorate, decoration={snake,segment length=1mm,amplitude=0.2mm}},
decoration={
	markings,
	mark=at position 0.5 with {\arrow{>}}},
path/.style=
	{-,
		decorate,
		decoration={snake,amplitude=.2mm,segment length=2mm,}}
],

\node[] (p1) {$P_1$};
\node[below= 2pt of p1] (p2) {$P_2$};
\node[below= 2pt of p2] (p3) {$P_3$};
\node[below= 2pt of p3] (p4) {$P_4$};
\node[below= 2pt of p4] (p5) {$P_5$};
\node[below= 2pt of p5] (p6) {$P_6$};

\foreach \i in {1,...,3}{ 
\node[svertex, fill=black, color=Orchid, right = 20pt of p\i] (p\i1) {};
}
\foreach \i in {4,...,6}{ 
\node[svertex, fill=black, color=ForestGreen, right = 20pt of p\i] (p\i1) {};
}

\foreach \i in {1,...,6}{ 
\node[svertex, fill=black, right = 50pt of p\i] (p\i2) {};
}

\node[svertex, fill=black, right = 80pt of p2] (e1) {};
\node[svertex, fill=black, right = 80pt of p5] (e2) {};

\path[draw, path]
($(p1.east) + (5pt,0)$) -- ($(p11.west) - (1.5pt, 0)$) 
(p12) -- ($(p1) + (170pt, 0)$)
($(p2.east) + (5pt,0)$) -- ($(p21.west) - (0pt, 0)$)
(p22) -- (e1)
(e1) -- ($(p2) + (170pt, 0)$)
($(p3.east) + (5pt,0)$)  -- ($(p31.west) - (0pt, 0)$)
(p32) --  ($(p3) + (171pt, 0)$)
($(p4.east) + (5pt,0)$) -- ($(p41.west) - (2pt, 0)$)
(p42) --  ($(p4) + (170pt, 0)$)
($(p5.east) + (5pt,0)$) -- ($(p51.west) - (0pt, 0)$)
(p52) -- (e2)
(e2) --  ($(p5) + (170pt, 0)$)
($(p6.east) + (5pt,0)$) -- ($(p61.west) - (0pt, 0)$)
(p62) --  ($(p6) + (170pt, 0)$)
;

\path[draw, -] 
(p11) -- (p12)
(p21) -- (p22)
(p31) -- (p32)
(p41) -- (p42)
(p51) -- (p52)
(p61) -- (p62)
;

\path[draw, -, line width = 1pt]
(p11) edge [color=Orchid] (p22)
(p21) edge [color=Orchid] (p32)
(p41) edge [color=ForestGreen] (p52)
(p51) edge [color=ForestGreen] (p62)
(p61) edge [color=ForestGreen] (p42)
(p31) edge [color=Orchid] (p12)
(p31) edge [bend left = 16, color=orange] (e2)
(p41) edge [bend right = 16, color=ProcessBlue] (e1)
;

\node[] at ($(e1) + (-2pt, -10pt)$) (le1) {$e'$};
\node[] at ($(e2) + (+2pt, +10pt)$) (le2) {$e''$};

\node[] at ($(p61) + (0pt, -10pt)$) (ls1) {$S$};
\node[] at ($(p62) + (0pt, -10pt)$) (ls2) {$S'$};

\end{tikzpicture}
				\caption{Two simple deadlock cycles $Z'$ and $Z''$.}
 			\end{subfigure}%
			\begin{subfigure}{.5\textwidth}
				\centering
				\begin{tikzpicture}[
%  -{Stealth[length = 2.5pt]},
%start chain = going down,
node distance = 10pt,
vertex/.style={minimum width=0.30em, minimum height=0.30em, circle, draw, on 
chain},
svertex/.style={minimum width=0.15em, minimum height=0.15em, circle, draw, 
on chain, inner sep=1pt, text width=3pt, align=center},
fix/.style={inner sep=1pt, text width=6pt, align=center},
graph/.style={minimum width=3em, minimum height=2em, ellipse, draw, on 
chain, decorate, decoration={snake,segment length=1mm,amplitude=0.2mm}},
decoration={
	markings,
	mark=at position 0.5 with {\arrow{>}}},
path/.style=
	{-,
		decorate,
		decoration={snake,amplitude=.2mm,segment length=2mm,}}
],

\node[] (p1) {$P_1$};
\node[below= 2pt of p1] (p2) {$P_2$};
\node[below= 2pt of p2] (p3) {$P_3$};
\node[below= 2pt of p3] (p4) {$P_4$};
\node[below= 2pt of p4] (p5) {$P_5$};
\node[below= 2pt of p5] (p6) {$P_6$};

\foreach \i in {1,...,3}{ 
\node[svertex, fill=black, color=Black, right = 20pt of p\i] (p\i1) {};
}
\foreach \i in {4,...,6}{ 
\node[svertex, fill=black, color=Black, right = 20pt of p\i] (p\i1) {};
}

\foreach \i in {1,...,6}{ 
\node[svertex, fill=black, right = 50pt of p\i] (p\i2) {};
}

\node[svertex, fill=black, right = 80pt of p2] (e1) {};
\node[svertex, fill=black, right = 80pt of p5] (e2) {};

\path[draw, path]
($(p1.east) + (5pt,0)$) -- ($(p11.west) - (1.5pt, 0)$) 
(p12) -- ($(p1) + (170pt, 0)$)
($(p2.east) + (5pt,0)$) -- ($(p21.west) - (0pt, 0)$)
(p22) -- (e1)
(e1) -- ($(p2) + (170pt, 0)$)
($(p3.east) + (5pt,0)$)  -- ($(p31.west) - (0pt, 0)$)
(p32) --  ($(p3) + (171pt, 0)$)
($(p4.east) + (5pt,0)$) -- ($(p41.west) - (2pt, 0)$)
(p42) --  ($(p4) + (170pt, 0)$)
($(p5.east) + (5pt,0)$) -- ($(p51.west) - (0pt, 0)$)
(p52) -- (e2)
(e2) --  ($(p5) + (170pt, 0)$)
($(p6.east) + (5pt,0)$) -- ($(p61.west) - (0pt, 0)$)
(p62) --  ($(p6) + (170pt, 0)$)
;

\path[draw, -] 
(p11) -- (p12)
(p21) -- (p22)
(p31) -- (p32)
(p41) -- (p42)
(p51) -- (p52)
(p61) -- (p62)
;

\path[draw, -, line width = 1pt]
(p11) edge [color=Orchid, dashed] (p22)
(p21) edge [color=red] (p32)
(p41) edge [color=ForestGreen, dashed] (p52)
(p51) edge [color=red] (p62)
(p61) edge [color=red] (p42)
(p31) edge [color=Orchid, dashed] (p12)
(p31) edge [bend left = 16, color=red] (e2)
(p41) edge [bend right = 16, color=red] (e1)
;

\node[] at ($(e1) + (-2pt, -10pt)$) (le1) {$e'$};
\node[] at ($(e2) + (+2pt, +10pt)$) (le2) {$e''$};

\node[] at ($(p61) + (0pt, -10pt)$) (ls1) {$S$};
\node[] at ($(p62) + (0pt, -10pt)$) (ls2) {$S'$};

\end{tikzpicture}
				\caption{An extended deadlock cycle.}
			\end{subfigure}%
			\caption{Illustration of the construction of an extended deadlock cycle if our algorithm can not compute $S$ and $S'$ with $S \ne S'$.
			The colored vertices are the vertices in $S^{Z'}$ and $S^{Z''}$ of the deadlock
			cycles $Z'$ and $Z''$  with the same color,
			respectively.}\label{fig:deadlock}
\end{figure}

          For a path $P_i\in \mathcal P$, let us consider the successor $v_i$ on $P_i$ of 
          a vertex $u_i$ in $S^Z$ that is replaced first among all
          vertices in $S^Z$ in the current phase. Clearly, it is replaced by
          $u_i$ in $S'$. Then we have a cross edge $e$ from $u_i$ either to (Case~1) a
          path $P'$ that does not begin with a vertex in $S^Z$ or (Case~2) $e$ ends behind a vertex of
          $S'$. 
          In the Case~1, by the property
          of $Z$, $P'$ can not be part of a deadlock cycle and we are not in
          the special case. 
          In the Case~2, by the property
          of $Z$, $e$ must end on a vertex that belongs to the deadlock
          cycle. Since $e$ ends behind a vertex of
          $S'$, we have an extended deadlock cycle; a contradiction.
          Thus, the vertices being successors of $S^Z$ remain
          in $S'$ and we make progress in the special case.

	{\bfseries Computing the new valid path data scheme.} 
	We know $V^*$ and the new boundary $\B'$.
	For a valid path data scheme it remains to compute
	a path numbering $\A'$, 
	the colors for the boundary vertices and their predecessors and successors
	that we store in an $O(n)$-bit structure $\C'$ using static space allocation.
	We can compute $\A'$ by moving along each path from $s$ to $t$.
	
	It is important that all vertices
	of an $s$-$t$ path are colored with the same color,
	however we want to use Lemma~\ref{lem:recomputeRegion}
	that computes path inside one region.
	Since $s$-$t$ paths are disconnected by boundary vertices
	our approach is to start with an arbitrary coloring 
	of the neighbors $W$ of $s$ part of the paths
	that are the first boundary vertices and from these
	make a parallel run over the paths. 
	Whenever we enter a region we explore it and 
	compute paths within the region.
	Using the paths we propagate the colors of the boundary vertices
	along the paths to the 
	next boundary vertices, their predecessors and
	successors and repeat our algorithm.
	
	Let $D \subseteq V^*$ (initially $D := W \cup \{s\}$) be
	the set of vertices for that we have already decided their color.
	To find a fixed routing within a region 
	we explore the region of every neighbor
	of $W$ in $G[V^* \setminus (\B' \cup D)]$ using a
	BFS that skips over vertices in $\B' \cup D$.
	We collect all visited vertices of one region in a set $U$
	realized by a balanced heap.
	(If $U$ is empty, the region was already explored
	from another neighbor of a vertex of $W$ and we
	continue with the next vertex $v \in W$.)
	In addition we construct two sets $S'$ and $T'$.
	Each visited vertex $u \in U$ is put in $S'$ if
	$\op{prev'}(u) \in \B'$ holds and put in $T'$
	if $\op{next'}(u) \in \B'$ holds (a vertex $u$ can be
	part of $S'$ as well as $T'$ if both holds).
	We apply Lemma~\ref{lem:recomputeRegion}
	with $U$, $S'$ and $T'$ as input and get 
	paths connecting each vertex of $S'$ with another vertex of $T'$.
	We set $D := D \cup U$ to avoid the exploration of explored
	regions again.
	
	Let $F$ be a set of vertices whose successors
	of $W$ are in $U$.
	Intuitively, by having~$F$ we can avoid
	computing the paths of a region again and gain.
	For each vertex $v \in F$ we use $\op{next'}(v)$
	and so move to a vertex $u$ of $S'$---we are now on
	a path computed by Lemma~\ref{lem:recomputeRegion}.
	We store the color $q$ of $v$ for this vertex
	and run along the path until it ends.
	We store the color $q$ for the last vertex~$w$
	of the path and 
	use $\op{next'}(w)$ to move to a boundary vertex $v'$
	that we also color with $q$.
	If the vertex $u' = \op{next'}(v')$ is on a computed path
	we repeat this paragraph with $v := v'$.
	Otherwise, we put $u'$ into a set $W'$ and remove $v$ from $F$.
	If $F$ loses its last vertex we set $W := W'$ and $F = W' = \emptyset$
	and repeat the iteration. Otherwise, we proceed with the
	next vertex $v \in W$.
	
	While moving over the paths we extend $D$ by the visited vertices.
	Moreover we compute the path numbering for $\A'$
	and all colors are stored in $\C'$.
	Recall that in order to use static space allocation we first need to know
	the set of all vertices in advance that we
	afterwards want to color.
	Do get that set we can run the algorithm in two steps.
	In the first we do not store the colors but only compute
	a set of vertices that must be colored.
	Using it as a key set for static space allocation we 
	repeat our algorithm in a second step and store the coloring.
	%
%	We now describe the details of the computation of the paths in $G'$
%	and its construction. 
%	Note that the approach is the same as described 
%	in the proof of Lemma~\ref{lem:recompute}.
%	Let $G'$ be the region in $G[V^* \setminus \B']$ defined in the previous paragraph.
%	We iterate over the vertices $v$ of $G'$ and look for neighbored vertices $u$ of $\B'$.
%	Using the weak path data structure we check if $u$ is the predecessor or successor
%	of~$v$. Let $S'$ and $T'$ be initially empty subsets of $V^*$.
%	If $u$ is the predecessor (successor) of $v$ we put $v$ into a set $S'$ ($T'$).
%	Otherwise ($u$ is neither a predecessor nor a successor of $v$) 
%	we proceed with the next vertex of $G'$.
%	We extend $G'$ by two vertices $s'$ and $t'$ and
%	connect $s'$ and $t'$ to all vertices of $S'$ and $T'$, respectively.
%	To fix the graph representation we
%	sort the vertices and adjacency arrays of the resulting graph.
	%
	The triple $(\A', \B', \C')$ represents the new valid path data structure.
	
	{\bfseries Correctness of the new valid path data scheme.}
	The boundary and the rerouted paths 
	represented by the weak path data structure 
	fixes the boundary vertices and their predecessor and successors.
	Hence our selection of the vertices for $U$, $S'$ and $T'$
	is also fixed for each region. Computing 
	the paths using Lemma~\ref{lem:recomputeRegion}
	fixes the path numbering. Since a (re)computation 
	with the same algorithm produces the same paths
	(being subpaths of our $s$-$t$ paths)
	and since we used the paths
	to propagate the coloring,
	the endpoints of a path and thus,
	all all colored vertices along a path are
	equally colored in $\C'$.
	Therefore our path data scheme is valid.
	Moreover, $S'$ and $T'$ consists of the endpoints
	of subpaths of good paths. Hence, a network-flow 
	algorithm can connect each vertex of $S'$ with another
	vertex of $T'$ and by Lemma~\ref{lem:probMaintain}
	our paths are chordless and deadlock free
	paths, it also holds for
	the paths of them.
	
%	Take $S$ as an array of $|\mathcal P'|$ fields and arbitrary assign
%	the $|\mathcal P'|$ neighbors of $s$ in $G[V^*]$ to $S$.
%	%
%	Run over all paths and store the path numbering of all paths in 
%	an $2n$ array $\A$. 
%	We extend $\B'$ by all large vertices and compute the set $\B''$ by 
%	adding $u$ into $\B''$, 
%	exactly if $\op{next'}(u) \in \B' \lor \op{prev'} \in \B'$.
%	For each member $v \in \B''$ we compute the index $j$ of the
%	predecessor $\op{prev'}(v) = \op{head}(v, j')$ with 
%	and index $j'$ of the successor $\op{next'}(v) = \op{head}(v, j')$
%	and store it as $(j, j')$ in $I$.
%	
%	Knowing the set $\B = \B' \cup \B''$ of vertices for that we have to store color information
%	we run again along the paths and store the colors in $\C$ for all vertices in $\B$.

	{\bfseries Efficiency.}
	We first consider the time to compute the boundary.
	For each new vertex $u$ in $S$, we
	spend $O(\op{deg}(u))$ time to determine once its degree in $G[V^* \setminus D]$
	and to decrement already stored degrees
	of the neighbors of~$u$.
	In each phase we iterate over all elements of $S$ and make $|S \Delta S'| = \Omega(1)$ progress 
	on at least one path,
	which means that we need at most $O(n\ell)$ phases.
	Determining the previous / next vertex~$u$ on a path
	can be done in time $O(\op{deg}(u) + k^3 \log^3 k)$.
	Dealing with simple deadlock cycles requires us
	to compute the $\ell$ initial members of $S'$, which cost us $O(\ell (\op{deg}(u) + k^3 \log^3 k))$ time 
	and roundwise replace members in $S'$ by their predecessors.
	Because $|S'| = O(\ell)$, this can cause $O(\ell)$ rounds.
	We spend $O(\ell^2)$ time to check if $S'$ is a separator.
        Thus,  $O(\ell^3)$ time allows us to compute a separator $S'$
        in the special case, and 
%	and fix $S'$ in $O(\ell (\op{deg}(u) + k^3 \log^3 k))$ time, 
%	which also includes 
        $O(\ell)$ time allows us to add the at most $\ell$ vertices to~$D$
	within each phase.
	
	To construct a valid path data scheme we need to explore all regions once, which
	can be done in 
	time linear to the size of each region plus
	the number of edges that leave the region, i.e.,
	in $O(nk)$ time.
	The construction of $S'$ and $T'$ requires 
	the execution of $\op{prev'}$ and $\op{next'}$
	on each vertex once, which runs in $O(n k^3 \log^3 k)$ time.
	Lemma~\ref{lem:recomputeRegion} uses 
	$O(n' k^2 \log^2 k)$ time for a region of size $n'$.
	Using it on every region once
	can be done in $O(n k^2 \log^2 k)$ time.
	Propagating the color requires us to move along the paths,
	which can be done in in time $O(n k^3 \log^3 k)$ time.
	The total time of the algorithm is $O(n k^3 \log^3 k)$.
	
	Our space bound is mainly determined by $D$, $V^*$, $\A'$, $\B'$ and $\C'$ each of $\Theta(n)$ bits and  
	the application of Lemma~\ref{lem:recompute} 
	($O(k^2 (\log k) \log n)$ bits), which is used
	in the weak path data structure whenever we access a previous or next vertex on a path.
	For $S$, $W$ and $F$ (and their copies) 
	we need $\Theta(\ell \log n)$ bits.
	The construction of the paths inside regions
	uses $O(k^2 (\log k) \log n)$ bits (Lemma~\ref{lem:recomputeRegion}).
	In total, we use $O(n + k^2 (\log k) \log n)$ bits.
	\end{proof}

Given path storage scheme storing $\ell < k$ good $s$-$t$-paths and an $(\ell + 1)$th path $P^*$ we can batchwise compute a path storage scheme storing $(\ell + 1)$ good $s$-$t$ paths by subsequently executing $O(\log k)$ times Lemma~\ref{lem:rerouting} and Lemma~\ref{lem:storageScheme}.

\vspace{3mm}
\begin{corollary}\label{cor:bathwise-inclusion}
Given a path storage scheme storing $\ell < k$ good $s$-$t$-paths $\mathcal P$ 
and an $(\ell + 1)$th dirty path 
with respect to $\mathcal P$
we can compute a path storage scheme storing $(\ell + 1)$ good $s$-$t$ paths
in $O(n k^4 \log^4 k)$ time using $O(n + k^2 (\log k) \log n)$ bits.
%\info{Woher kommt das $\ell$ in der Zeit?}
\end{corollary}

Initially our set of $s$-$t$-paths $\mathcal P$ is empty. % and all vertices of $B$ are zeros.
%The computation of a chordless $s$-$t$ path 
%takes $O(n (k + \log^* n) k^3 \log^3 k)$ time and $O(n + k^2 (\log k) \log n)$ bits (Lemma~\ref{lem:newPath}). 
%The time to construct a path data scheme is included.
Then a repeated execution of Lemma~\ref{lem:newPath} 
and
Corollary~\ref{cor:bathwise-inclusion} allows us to show
our final theorem for computing vertex-disjoint $s$-$t$-paths
in $O(k ((n (k + \log^* n) k^3 \log^3 k)
+ (n k^4 \log^4 k)))
= O(n (k \log k + \log^*n) k^4 \log^3 k)$ time.

\begin{theorem}\label{the:nbitvertexdisjoint}
	Given an $n$-vertex graph $G = (V, E)$ with treewidth $k$ and two 
	vertices $s, t \in V$ there is an 
	algorithm that can compute a maximum amount of $k$ 
	chordless vertex-disjoint paths from $s$ to $t$ in 
	$O(n (k \log k + \log^*n) k^4 \log^3 k)$ time using
	$O(n + k^2 (\log k) \log n)$ bits.
\end{theorem}
	
Knowing a maximum set of $k$ vertex-disjoint paths between two vertices $s$ and $t$, we can
easily construct a vertex separator for $s$ and $t$.

\begin{corollary}\label{cor:speffVertexSeparator}
	Given an $n$-vertex graph $G = (V, E)$ with treewidth $k$, and two 
        vertices $s \in V$ and $t \in V$, $O(n (k \log k + \log^*n) k^4 \log^3 k)$ time
and  $O(n + k^2 (\log k) \log n)$ bits
suffice to construct a bit array $S$ marking all vertices of
	a vertex separator for $s$ and $t$.
\end{corollary}
\begin{proof}
First construct the maximum number~$k$ of possible chordless
pairwise vertex-disjoint paths from~$s$ to~$t$ (Theorem~\ref{the:nbitvertexdisjoint}).
Then try to construct a further $s$-$t$ path with a DFS as describe in the proof of Lemma~\ref{lem:newPath}
and store the set~$Z$ of vertices that are processed by the DFS.
Finally run along the paths from~$s$ to~$t$ and compute the set~$S$ consisting 
of the last vertex on each path that is part of~$Z$.
These vertices are an $s$-$t$-separator.
\end{proof}

Many practical applications that use treewidth algorithms have graphs with 
treewidth $k = O(n^{1/2-\epsilon})$ for an arbitrary $\epsilon>0$, and 
then our space consumption is~$O(n)$ bits.

\section{Sketch of Reed's Algorithm}\label{sec:reedsalg}
In this section we first sketch Reed's algorithm to compute a tree 
decomposition
and then 
the computation of a so-called balanced $X$-separator.
%, i.e., a 
%set of vertices whose removal splits the graph in at least two connected components. 
In the following sections, we modify his
algorithm to make it space efficient.

Reed's algorithm~\cite{Reed92} takes
an undirected $n$-vertex $m$-edge graph $G = (V, E)$ with treewidth $k$ and 
an initially empty vertex set $X$ as input
and outputs a balanced tree decomposition of width $8k + 6$.
To decompose the tree he makes use of separators.
An $X$-separator is a set $S \subset V$ such that $S$ separates $X$ among 
the connected components of $G[V \setminus S]$ and no component 
contains more than $2/3|X|$ vertices of $X$.
A balanced $X$-separator $S$ is an $X$-separator with the additional property 
that no component of $G[V \setminus S]$ contains more than $2/3|V|$ vertices.

The decomposition works as follows.
If  $n \leq 8k + 6$, we return a tree decomposition 
$(T,B)$ consisting of a tree with one node $r$ (the root node) and a mapping 
$B$ with $B(r) = V$.
Otherwise, we search for a so-called balanced $X$-separator $S$ of size 
$2k + 2$ that divides $G$ such that $G[V 
\setminus 
S]$ consists of~$x \ge 2$ connected components $\Gamma = \{G_1, 
\ldots, G_x\}$.
Then, we create a new tree $T$ with a root node $r$, a mapping $B$, and set 
$B(r)$ to $X \cup S$.
For each graph $G_i \in 
\Gamma$ with $1 \leq i \leq x$, we proceed recursively with $G' = G_i[V(G_i) 
\cup 
S]$ and $X' = ((X \cap V(G')) \cup S)$. 
Every recursive call returns a tree decomposition $(T_i, B_i)$
($i=1,\ldots,x$). 
We connect the root of $T_i$ to $r$, 
we then set $B(w) = B_i(w)$ for all nodes $w \in T_i$. 
After processing all elements of $\Gamma$ return the tree decomposition $(T,B)$.

Since a balanced $X$-separator is used, the tree 
 has a depth of 
$O(\log n)$, and thus the recursive algorithm produces
at most $O(\log n)$ stack frames on the call 
stack---each stack frame is associated with a node $w$ of $T$. 
A standard implementation of the algorithm needs a new graph 
structure for each recursive call. In the worst-case, each of these graphs 
contains $2/3$ of the vertices of the previous graph.
Thus, the graphs on the stack frame use $\Theta((n+m) \log n) = \Theta(kn \log 
n)$ bits.
Storing the tree decomposition $(T,B)$ requires $\Theta(kn \log n)$ bits 
 as well. The various other 
structures needed can be realized within the same space bound.
In conclusion, a standard implementation of Reed's algorithm requires 
$\Theta(kn \log 
n)$ 
bits.

The next lemma shows a space-efficient implementation
for finding a balanced $X$-separator.

\begin{lemma}\label{lem:sep}
	Given an $n$-vertex graph $G=(V,E)$ with treewidth $k$ and $X \subseteq V$ of at most $6k + 6$ vertices,
	there is an algorithm for finding a balanced $X$-separator of 
	size $2k + 2$ in~$G$ that, for some constant~$c$, 
	runs in $O(c^k n \log^* n)$ time and 
	uses $O(n + k^2 (\log k) \log n)$ bits.
	%For some constant $\tilde{c}$, 
	%the algorithm searches a set consisting of~$k$ vertex-disjoint
	%paths~$\tilde{c}^k$ times and executes~$O(1)$ extra DFS.
\end{lemma}
\begin{proof}
	We now sketch 
Reed's ideas to compute 
a balanced 
$X$-separator. 
To compute a balanced $X$-separator we compute first an $X$-separator 
$S_1$.
To make it balanced, we compute an additional $R$-separator $S_2$ 
where $R$ is a set of vertices that is in some sense equally distributed in $G$.
Then $S = S_1 \cup S_2$ is a balanced $X$-separator.

		Reed computes an $X$-separator by 
		iterating over 
		all 
		$3^{|X|}$ possibilities 
		to split $X$ into three vertex-disjoints sets $X_1$, $X_2 \subseteq V$ and 
		$X_S$ 
		with $|X_S| \le k$ and $|X_1|, |X_2| \le \max\{k, 2/3|X|\}$.
		For each iteration
		connect two new vertices $s'$ and $t'$ with all vertices of $X_1$ and $X_2$,
		respectively, compute vertex disjoint paths with Corollary~\ref{cor:speffVertexSeparator} 
		to find a separator $S$ and check if $X_S \subseteq S$ holds.

We now shortly describe Reed's computation of the set $R$.
Run a DFS on the graph $G$ and compute in a bottom up process for each 
vertex~$v$ of the 
resulting DFS tree the number of descendants of~$v$. Whenever this number 
exceeds $n / (8k + 6)$,  add~$v$ to the initially empty set 
$R$ and reset the number of descendants of $v$ to zero.
At the end of the DFS, the set $R$ consist of at most $8k + 6$ vertices,
which can be  used to compute  
a balanced $X$-separator as described above.
		Similar to the $X$-separator, we compute an $R$-separator.

%It remains to show how the set $R$ is computed.
We now show how to compute $R$ using $O(n)$ bits.
%in $O(m + n \log^* n)$ time using $O(n)$ bits.
%		
%The idea is to use a balanced parenthesis representation for the DFS-tree
%used during the computation.
%The representation allows us to compute for 
%each node $w$ of the tree the position within an $n$-bit set
%where we can store the number of descendants of $w$ as a self-delimiting 
%number. 
%
%
%For the next lemma, we describe the runtime
%of the computation of $R$ through a number of 
%DFS runs because we later use the lemma on a graph interface that changes
%the runtime of the DFS.\info{Das machen wir nicht!}
%
%		
		To prove the lemma, we use the following observation.
		Whenever the number of descendants for a node $u$ is computed, 
		the numbers of $u$'s children are not required anymore.
		The idea is to use a {\em balanced parentheses} representation, which
		consists 
		of
		an 
		open parenthesis for 
		every node $v$ of a tree, followed by the balanced parentheses representation of the 
		subtrees 
		of every child of $v$, and a matching closed parenthesis. 
		Consequently, if $v$ is a vertex with $x$ 
		descendants 
		having its open parenthesis at 
		position $i$ and its closed parenthesis 
		at position $j$, then the difference 
		between $i$ and $j$ is $2x$. 
		%We can store a self-delimiting number $x$ with
		%$x + 1 < 2x$ bits.
		Note that, taking an array $A$ of $2n$ bits,
		we can store the number of descendants
		of $v$ in $A[i \ldots j]$ as a so-called self-delimiting
		number by writing $x$ as $1^x0$. 
		This means that we overwrite the self-delimiting
		numbers stored fo the descendants of the children of~$v$.
		
		To construct $R$ 
		we run a space-efficient DFS
		twice, first to 
		construct a balanced parentheses representation of the DFS tree, which is
		used to compute the descendants of each vertex in the DFS tree and so 
		choose
		vertices for the set $R$, and a second time to translate the \textsc{id}s of the 
		chosen vertices since 
		the balanced parentheses  representation is an ordinal tree,
		i.e., we lose our original vertex \textsc{id}s and the vertices get a numbering
		in the order the DFS visited the vertices.
		After choosing the vertices that belong to the set $R$
		and marking them in a bit array $R'$, we
		run the DFS again and create a bit array $R^*$
		that marks every vertex $v$ that the DFS visits as the $i$th vertex
		if and only if $i$ is marked in $R'$.

			It remains to show how to compute the bit array $R'$.
			Let $P$ be a bit array of $2n$ bits storing the balanced parentheses  
			representation, and let $A$ be a bit array of $2n$ bits that we use to 
			store 
			the numbers of descendants for some vertices.
			Note that a leaf is identified by an open parenthesis
			followed by an immediately closed parenthesis.
			Moreover, since the balanced representation is computed via a DFS in 
			pre-order, we will visit the vertices by running through~$P$
			in the same order. Note that Munro and Raman~\cite{MunR97} showed a succinct data 
			structure for balanced representation that initializes in~$O(n)$ time 
			and allows to compute the position of a 
			matching parenthesis, i.e., given an index $i$ of an open (closed) 
			parenthesis 
			there is an operation $\op{findclose}(i)$ ($\op{findopen}(i)$) that 
			returns the 
			position~$j$ of the matching closed (open) parenthesis.
			
			The algorithm starts in Case $1$ with $i = 1$ ($i \in \{1, \ldots, 
			2n\}$).
			We associate $0$ with an open parenthesis
			and $1$ with a closed parenthesis.
			\begin{description}
				\item[Case 1] Iterate over $P$ until a leaf is found at position 
				$i$, i.e.,
				find an~$i$ with $P[i] 
				= 0 \land P[i + 1] = 1$.
				Since we found a leaf we write a $1$ as a self-delimiting number in 
				$A[i \ldots i + 1]$.
				Set $i := i + 2$ and check if $P[i] = 1$. If so, move to Case 2, 
				otherwise 
				repeat Case $1$.
				
				\item[Case 2] At position $i$ is a closing parenthesis, i.e., $P[i] 
				= 1$.
				In this cases we reached the end of a subtree with $j = 
				\op{findopen}(i)$ 
				being the position of a corresponding open parenthesis.
				That means we have already computed all numbers for the whole 
				subtree.
%				By the parenthesis representation we know that the number of 
%				children 
%				in 
%				this 
%				subtree is $c = (i - j) / 2$.
				Using an integer variable $x$, sum up all the self-delimiting 
				numbers in 
				$A[j + 1 \ldots i - 1]$.
				Check if the sum $x + 1$ exceeds $\ell$. If it does write $0$ as a 
				self-delimiting number in $A[j \ldots i]$ and set $R'[j] = 1$, otherwise 
				write $x + 1$ in $A[j \ldots i]$.
				Note that we store only one self-delimiting number between an open
				parenthesis and its
				matching closed parenthesis, and this number does not necessary 
				occupy 
				the whole space available. Hence, using $\op{findclose}$ operation 
				we 
				jump to the end of the space that is reserved for a number and 
				start 
				reading the second.
				
				After writing the number we set $i := i + 1$.
				We end the algorithm if $i$ is out of $P$, otherwise we check in 
				which 
				case 
				we fall next and proceed with it.
			\end{description}
			
			{\bfseries Efficiency:}
			To compute the set $R$, we run two space-efficient DFS
			with $O(n)$ bits in $O(m + n \log^* n) = O(nk \log^* n)$ time
			(Lemma~\ref{lem:spaceDFS}).
			The required $X$-separator and $R$-separator are computed
			in $O(n (k \log k + \log^*n) k^4 \log^3 k)$ time (Corollary~\ref{cor:speffVertexSeparator}).
			We so get a balanced $X$-separator in
			$(3^{|X|}((nk \log^* n) + (n (k \log k + \log^*n) k^4 \log^3 k))) = 
			O(c^k n \log^* n)$ time for some constant $c$.
%$O((3^{|X|} + 2^{15,7k})k(m + n \log^* n)) = 
%O((3^{|X|} + 2^{15,7k})(k^2n + kn\log^* n))$

The structures $P$, $A$, and the space-efficient DFS use $O(n)$ bits.
The size of the sets $X_1$, $X_2$, $X_S$ and $R$ are bound by $O(k)$ and so use $O(k \log n)$ bits.
Corollary~\ref{cor:speffVertexSeparator}
uses $O(n + k^2 (\log k) \log n)$~bits, which is the bottleneck of our space bound.
\end{proof}

\section{Iterator for Tree Decompositions using $O(kn)$ bits}\label{sec:stream}

We now introduce our 
iterator by showing a data structure, which 
we 
call 
{\em tree-decomposition iterator}. We think of it as an agent moving through a 
tree decomposition $(T,B)$, one node at a time in a specific order. 
We implement such an agent to traverse $T$  in the order of an 
\mbox{Euler-traversal} and, when visiting some node $w$ in $T$, 
to 
return the tuple $(B(w), d_w)$ with~$d_w$ being the depth of the node~$w$.

The tree-decomposition iterator provides the following operations:
\begin{itemize}
	\item $\op{init}(G, k)$: Initializes the structure for an undirected $n$-vertex graph $G$ with treewidth $k$.
	\item $\op{next}$: Moves the agent to the next node according to
	an Euler-traversal and returns \textsc{true} unless the traversal of $T$ 
	has 
	already 
	finished. In that case, it returns \textsc{false}.
	\item $\op{show}$: Returns the tuple $(B(w), d_w)$ of the node $w$ where 
	the 
	agent is currently positioned. 
\end{itemize} 

We refer to initializing such an iterator and using it to iterate (call 
$\op{show}()$ 
after every call of  
$\op{next}()$) over the entire tree 
decomposition $(T,B)$ of a graph $G$ as  {\em iterating over a tree-decomposition} 
$(T,B)$ of $G$. 
Our goal in this section is to show that we can iterate over the bags of 
a tree 
decomposition by using~$O(kn)$ bits and 
$c^{k}n \log n \log^*n$ time for some constant~$c$.
%
%To save space we are often working with arrays of bits and therefore assume
%that the (vertex) sets used in the following lemma and subsequent proof are
%arrays of bits with a bit at index $i$ set to $1$ exactly if $i$ is contained
%in the set.
Recall that Reed's algorithm computes a tree decomposition using separators.
Assume we are given a separator~$S$ in~$G$.
The separator divides $G$ into several connected components, which we 
have to find for a recursive call.
We refer to a data structure that implements the functionality
of the lemma below as a {\em connected-component finder}.
In the next lemma we use a choice dictionary~\cite{Hag18cd,HagK16,KamS18c}
that manages a subset $U'$ of a universe $U = \{1, \ldots, n\}$
and provides---besides constant-time operations contains, add and remove---a linear-time iteration over $U'$.

\begin{lemma}\label{lem:connectedcomponent}
Given an undirected $n$-vertex $m$-edge  graph $G=(V,E)$ and a vertex set $S
\subseteq V$,  there is an algorithm that iterates over all connected components of
$G[V \setminus S]$. The iterator is realized by the following operations.
\begin{itemize}
	\item $\op{init}(G, S)$: Initializes the iterator.
	\item $\op{next}$: Returns true exactly if there is another connected component left to iterate over.
	Returns false otherwise.
	\item $\op{show}$: Returns the triple $(C, p, n')$ where $C$ is 
	choice dictionary containing all $n'$
	vertices of a connected component, and $p \in C$.
\end{itemize}
The total runtime of all calls of next is $O(n + m)$ unless $\op{next}$
returns false,
and the running time of $\op{init}(G, S)$ as well as of $\op{show}$ 
is $O(1)$. All operations use $O(n)$ bits.
\end{lemma}
\begin{proof}
The iterator is initialized by creating a bit array~$D$
to mark all vertices that are part of connected components 
over which we already have iterated.
To hold the current state of the iterator 
and answer a call of $\op{show}$ we store the triple $(C = \emptyset, p = 0, n' = 0)$
(as defined in the lemma).
Technically, to avoid modifications of the
internal state of the iterator, we maintain a copy of $C$ that we return if $\op{show}$ is called.

If $\op{next}$ is called, we iterate over $V$ starting at vertex $p$
until we either find a vertex $p' \in V \setminus (S \cup D)$ or
reach $|V| + 1$ (out of vertices).
If $p'$ exists, we have found a connected component whose vertices are not part of $D$.
We prepare a possible call of $\op{show}$ by computing new values
for $(C, p, n')$ as follows:
using a space-efficient BFS explore the connected component in $G[V \setminus S]$ starting at~$p'$.
Collect all visited vertices in a choice dictionary $C'$ 
as well as add them to $D$.
Set $C := C', p := p'$ and $n' = |C'|$. Finally output true.
Otherwise,~$p'$ does not exits and we reached $|V| + 1$, set $p' = |V| + 1$ and return false.
For a call of $\op{show}$ return $(C, p, n')$.

A choice dictionary can be initialized in constant time, thus 
$\op{init}$ as well as $\op{show}$ run in constant time.
The operation $\op{show}$ has to scan $V$ for the
vertex $p'$ in $O(n)$ time and runs a BFS in $O(n + m)$ time to explore
the connected component containing~$p'$.
However, the total time of all $\op{next}$ operation is $O(n + m)$
since the operation continues the scan in $V$ from $p$ (the last found $p'$)
and avoids the exploration all vertices in $D$, i.e., all already explored connected components.
The bit array $D$ as well as the choice dictionary $C$ use $O(n)$ bits.

%	As auxiliary structures we use a vertex set~$A$
%	to store all vertices of the 
%	connected components which have been found and a {\em pin}, which is a 
%	pointer to some vertex contained in a 
%not-yet-found connected component of $G$.  The vertex set $A$ together with
%the pin stores the state of the algorithm, i.e., they allow us to collect the
%next connected component and find the size of the next connected component. 
%Initially $A$ contains no vertices, i.e., $A$ is a bit array of size $n$
%with all bits set to $0$.  We initialize the pin to be the first vertex in
%$V$ that is not contained in $S$ or $A$.  To realize this we iterate over
%$V$ until such a vertex is found.  We refer to this as {\em updating the
%pin}.  If the pin can not be updated, there are no more connected components
%to be found, and we set the pin to null.  To find the next connected
%component start a BFS at the pin.  The BFS is not allowed to traverse to a
%vertex contained in $S$ and each time a new vertex $v$ is visited, add $v$ to $A$.
%Once finished, update the pin to a not-yet-found connected 
%component if possible.  To determine the
%size of the next connected component skip the updating process and simply
%count and return the number of vertices visited.  Updating the pin until it
%can not be updated anymore runs in $O(n)$ time.  Running a BFS over all
%vertices runs in $O(n+m)$ total time and uses $O(n)$ bits.  Storing the vertex
%set $A$ uses  $O(n)$ bits and the pin $O(\log n)$ bits.  Thus, we can find
%all connected in $O(n+m)$ time and $O(n)$ bits.
\end{proof}

To implement our tree-decomposition iterator we turn Reed's recursive algorithm~\cite{Reed92} into an iterative version. 
For this we use a stack
structure called {\em record-stack} that manages a set of data structures to 
determine the current state of the algorithm.  Informally, the record-stack
allows us to pause Reed's algorithm at specific time-points and continue
from the last paused point.  With each recursive call of Reed's algorithm we
need the following information: 
\begin{itemize}
	\item an undirected $n_i$-vertex graph $G_i=(V_i,E_i)$ ($i=0,1,2,\ldots$) of treewidth $k$,
	\item a vertex set $X_i$, a separator $S_i$, and
	\item an instance $F_i$ of
the connected-component-finder data structure
 that iterates over the connected components of $G[V_i \setminus S_i]$
and that outputs the vertices of each component 
in a bit array.
\end{itemize}

We call the combination
of these elements a {\em record}.  
Although we use a record-stack
structure, often we think of it to be a combination of
specialized stack structures: a {\em subgraph-stack}, which manages to store
the recursive graphs used as a parameter for the call of Reed's algorithm, a
stack for iterating over the connected components of $G[V \setminus S]$,
called {\em component-finder stack}, a stack containing the separators as
bit arrays, called {\em$\mathcal{S}$-stack}, a stack containing the vertex
sets $X$ as bit arrays, called $\mathcal{X}$-stack. %
The bit arrays $S_i$, $X_i$ and $F_i$ contain information referring to $G_i$
and are thus of size $O(n_i)$.
On top of $S_i$ and $X_i$ we create rank-select data
structures.
%We require these structures to calculate the bag associated with the
%current record in $O(k)$ time.
Pushing a record
$r_{\ell+1}=(G_{\ell+1},S_{\ell+1},X_{\ell+1},F_{\ell+1})$ to the record-stack
is equivalent to pushing each element in $r_{\ell+1}$ to the corresponding
stack (and analogous for popping).

\begin{lemma}\label{cor:recordstack}
	When a record-stack $R$ is initialized for an undirected $n$-vertex graph 
	$G$ with treewidth $k$ such that each subgraph $G_i$  of $G_0=G$ on the 
	subgraph-stack of $R$ contains 2/3 of the vertices of $G_{i-1}$ for $0 < i 
	< \ell$ and $\ell=O(\log n)$, then the record-stack occupies 
	$O(n + k \log^2 n)$ bits plus $O(kn)$ bits for the subgraph stack.
\end{lemma}
\begin{proof}
Let $m$ be the number of edges in $G$.
We know that the size of the subgraph-stack structure is $O(n+m)$ bits when the
number of vertices 
of the subgraphs shrink with every push by a factor $0 < c < 1$. Since 
each subgraph of $G_0$ has also a treewidth $k$, the number of edges of each 
subgraph is bound by $k$ times the number of vertices. Thus, the subgraph stack 
uses $O(n+m)=O(kn)$ bits.

The size of the bit arrays $X_i$, $S_i$ (including the respective rank-select 
structures) and the component-finder $F_i$ is 
$O(n_i)$ for $0 \leq i \leq \ell$. This means the total size of the 
stacks 
containing these elements is $O(n)$ bits since they shrink in the same way as 
the 
vertex sets of the subgraphs. Storing the bag that is currently being output 
uses $O(k \log n)$ bits.
Thus, the size of the record-stack
without the subgraph stack is $O(n + k \log^2 n)$ bits.
\end{proof}

We call a tree decomposition $(T, B)$ {\em balanced} 
if $T$ has logarithmic height, 
and {\em binary} if $T$ is binary. Using our 
space-efficient separator computation for 
finding a balanced $X$-separator we are now able to show the following theorem.

\begin{theorem}\label{thm:spaceefficientiterator1}
	Given an undirected $n$-vertex
	graph $G$ with treewidth $k$,
	there exists an iterator that outputs a balanced and binary
tree de\-com\-po\-si\-tion $(T,B)$ of width $8k+6$ in Euler-traversal order 
using $O(kn)$ bits and
$c^{k}n
\log n \log^*n$ time for some constant $c$.
\end{theorem}

\begin{proof}
We implement our tree-decomposition iterator by showing $\op{init}$, $\op{next}$ and $\op{show}$.
We initialize the iterator for a graph $G$ with $n>8k+6$
vertices, by initializing a flag $f=0$, which indicates that the agent
is not yet finished, and initialize a record-stack.  
The record stack is
initialized by first initializing its subgraph stack with a reference to $G$
as the first graph $G_0$.  Next, we push the empty vertex set $X_0$ on the
$\mathcal{X}$-stack in form of an initial-zero bit array $X_0$ of length
$n$.  Now, using the techniques described in Lemma~\ref{lem:sep}, we find a
balanced $X_0$-separator $S_0$ of $G_0$ and push it on the
$\mathcal{S}$-stack. Then we create a new connected-component-finder
instance $F_0$ (Lemma~\ref{lem:connectedcomponent}) and push $F_0$ on the
component-finder stack.  
 
%For this we also store a boolean value $f=0$ to indicate that
%the agent has not yet finished its traversal of $T$.
	
	We now view our implementation of $\op{next}()$, which has the task to 
	calculate the next bag on the fly. 
Since the tree $T$ of our tree decomposition does not exist as a real 
structure, 
we only {\em virtually} move the agent to the next node by advancing the state 
of Reed's algorithm.
If $f=1$, we return false (the agent can 
	not be moved) and do not change the state of the record-stack.
Otherwise, we first virtually move the agent and then return true.
	
{\bfseries Move to parent:} 
	If $n_\ell \leq 8k+6$ (we leave a leaf), we pop the record stack.
	
	Otherwise, if the connected-component-finder instance $F_\ell$
	has iterated over all connected components and the 
	record-stack contains more than one record, we pop it.
	If afterwards the record-stack contains only one 
	record, we set $f=1$ (the agent moved to the root from the rightmost child; the traversal is finished).

{\bfseries Move to next child:}
	If $F_{\ell}$ has not iterated over all connected components (the agent is moving to a
previously untraversed node), we use $F_{\ell}$ to get the next
connected-component~$C$ in~$G_\ell[V_\ell \setminus S_\ell]$, we push
the vertex-induced subgraph $G[C \cup S_\ell]$ on the subgraph stack as
$G_{\ell+1}=(V_{\ell+1}, E_{\ell+1})$ and proceed with one of the next two cases.

\begin{itemize}
	\item If $n_{\ell + 1} \leq 8k+6$, we are calculating the bag of a leaf of
        $T$ by setting $B(w)=V_{\ell+1}$.
        We do this by pushing a bit array with all bits set to
        $1$ on the $\mathcal{S}$-stack and $\mathcal{X}$-stack and an empty
        component-finder on the component-finder stack.
	\item If $n_{\ell + 1} > 8k+6$, the agent is moving to a new internal node
	whose bag we are calculating as 
	follows: we push a new bit array $X_{\ell+1}=(X_\ell \cap V_{\ell+1}) \cup 
	S_\ell$ on the $\mathcal{X}$-stack. We then find the balanced 
	$X_{\ell+1}$-separator $S_{\ell+1}$ of $G_{\ell+1}$ and push it on the 
	$\mathcal{S}$-stack. Then, create a new connected component-finder 
	$F_{\ell+1}$ for $G_{\ell+1}$ and $S_{\ell+1}$ and push it on the 
	component-finder stack and return true.
\end{itemize}
        
        Anytime we pop or push a new record, we call the \op{toptune} function 
        of the subgraph stack to speed up the graph-access operations.
    This finishes our computation of $\op{next}()$.

	To implement $\op{show}()$, we return the tuple $(B(w), 
	d_w)$ with $B(w)$ being the current bag, and $d_w$ being the number of 
	records of the record-stack.
	The current bag is defined as $S \cup X$.
	Thus, we iterate over all elements of $S \cup X$ via 
	their rank-select data structures. 
	Note that, since the subgraph 
	$G_{\ell}$ on top of the record stack is toptuned, 
	we can return the bag as vertices of $G_0$ or $G_{\ell}$ in $O(k)$ 
	time.

	{\bfseries Efficiency:}
	The iterator uses a record-stack structure, which occupies $O(kn)$ bits.
	Since the running time of Lemma~\ref{lem:sep} is $O(c^k n \log^*n)$,
	for some constant~$c$, and the input graphs are split
	in each recursion level into
	vertex disjoint subgraphs, %that share only $O(k)$ vertices,
	the running time in each recursion level is $O(c^k n \log^*n)$.
	Summed over all $O(\log n)$ recursion levels we get 
	a running time of $O(c^k n \log n \log^*n)$.
	
%	For finding separators, we use the balanced $X$-separator search of 
%	Lemma~\ref{lem:sep} 
%        with its
%	internal vertex-disjoint path-construction algorithm. The DFS has a runtime
%	of $O(kn'\log^*n')$ and uses $O(n')$ bits on an $n'$-vertex graph $G$ 
%	with treewidth~$k$ (Lemma~\ref{th:dfs}).  
%	Since our recursion depth is $O(\log n)$, by Lemma~\ref{lem:sep} we get a total runtime over all separator searches of $c^kn\log n\log^*n$
%	for some constant~$c > 0$,
%	which makes the
%	overhead for toptune 
%	calls of the subgraph stack negligible.
%	All other operations, such as finding all connected 
%	components, have runtime $O(n\log n)$.
%	We thus get a runtime of $c^{k}n\log n \log^*n$
%	for some constant~$c > 0$.

	{\bf Make $T$ binary.}
	The balanced $X$-separator $S$ partitions $V\setminus S$ into any number of 
	vertex 
	disjoint sets between $2$ and $n$ such that no set contains more than $2/3$ 
	of the vertices of $V$ (and $X$). The idea is to combine these vertex sets 
	into 
	exactly two sets such that neither contains more than $2/3|V|$ vertices. 
	For this we change our usage of the connected component finder slightly.
%	Previously, before retrieving the next connected component with some
%	 connected component finder $F_\ell$ (initialized for the graph $G_\ell$ on 
%	 top of the subgraph stack) we have allocated a new bit array $C$ of size 
%	 $n_\ell$,
%	 with all bits set to $0$, and then collected the next connected component 
%	 in $C$.
	 After we initialize $F_\ell$ we also initialize two bit arrays 
	 $C_1$ and $C_2$ of size $n_\ell$ each with all bits set to $0$. 
	We also store the number of bits set to $1$ for each of the bit arrays as 
	$s_1$ and $s_2$, i.e., the number of vertices contained in them (initially 
	$0$).
	 We now want to collect the vertices of all connected components of $G_\ell[V_\ell 
	 \setminus S_\ell]$ in $C_1$ and $C_2$. 
	While there are still connected components to be returned by $F_\ell$, this 
	is done by obtaining the size of the next connected component via $F_\ell$ 
	as $s$.
	 If $s_1 + s \leq 2/3|V_\ell|$, we collect the next connected 
	 component in $C_1$ and increment $s_1=s_1+s$. 
	Otherwise, we do the same but for $C_2$ and $s_2$. Doing this until all 
	connected components are found results in $C_1$ and $C_2$ to contain all 
	connected components of $G_\ell[V_\ell \setminus S_\ell]$. 
	For $(C_1,C_2)$ we implement a function that  returns $C_1$ if it was not yet 
	returned, or $C_2$ if it was not yet returned, or null otherwise.
	 We store $(C_1,C_2)$ with the respective functions on the connected 
	 component finder stack (instead of $F_\ell$). 
	Any time we do this during our iterator, 
        the graph $G_\ell$ is 
	toptuned, resulting in constant time graph access operations. 
	The previous runtime and space bounds still hold.
\end{proof}

	Often it is needed to access the subgraph $G[B(w)]$ induced by a {\em bag} $B(w)$ of a tree decomposition $(T,B)$ for further computations. We call such a subgraph {\em bag-induced}. For this we show the following:
	
	\begin{lemma}\label{lem:baginducedgraph}
		Given an undirected $n$-vertex
		graph $G$ with treewidth $k$ and an iterator $\mathcal{A: G \rightarrow (T,B)}$ that iterates
		over a balanced
		tree de\-com\-po\-si\-tion of width $O(k)$ we can additionally
		output
		the bag-induced subgraphs using $O(k^2 \log n)$ bits additional space and $O(kn\log n)$ additional time.
	\end{lemma}
	\begin{proof}
		To obtain the edges we use a bit matrix $M_{\ell}$ of size $O(k^2)$ 
		whose bit at index $[v][u]$ is set to $1$ exactly if there exists an
		edge~$\{v,u\}$ in~$G[B(w)_\ell]$.
		To quickly
		find the edges in~$G[B(w)_\ell]$ we use $M_\ell$ together with
		rank-select data structures on $S_\ell$
		and $X_{\ell}$ that allow us to map the vertices
		in $B(w)_\ell = S_\ell \cup X_\ell$, names as vertices of $G_\ell$, to $\{1, \ldots, k\}$.
		%of $G_\ell$ to $\{1 \ldots k\}$.
		We create $M_{\ell}$ anytime a new record $r_{\ell}$
		is pushed on the record stack, and anytime $r_{\ell}$ is popped, we throw away 
		$M_{\ell}$.
		For reasons of performance relevant in Section~\ref{sec:improvementsinspace},
		we want to avoid accessing
		vertices multiple times in a graph.
		Hence, when we push a bit matrix $M_{\ell}$, we first use $M_{\ell-1}$ to
		initialize edges that were contained in $B(w)_{\ell-1}$ and are still contained in 
		$B(w)_{\ell}$. 
		To obtain the edges of a vertex $v$ in $B(w)_\ell$,
		we can iterate over all the edges of~$v$ in $G_\ell$
		and check if the opposite
		endpoint is in $B(w)_\ell$. 
		However, by using definition {\textbf (TD2)} of a 
		tree decomposition, it suffice to iterate only over the edges of such a vertex 
		once, i.e., the first time they are contained in a bag
		to create $M_{\ell}$.
		
		{\bfseries Efficiency:}
		We can see that we store at most $O(\log n)$ matrices this way
		since the record stack contains $O(\log n)$ records (i.e., the height of the 
		treedecomposition).
		Storing all bit matrices uses $O(k^2 \log n)$ bits and initializing all bit matrices takes $O(kn\log n)$ time, including initializing and storing the rank-select structures if not already present.
	\end{proof}

We conclude the section with a remark on the output scheme of 
our iterator. 
The specific order of an Euler-traversal encompasses 
many other orders of tree traversal such as pre-order, in-order or post-order. 
To achieve these orders we simply filter the output of 
our iterator, i.e., skip some output values.

\section{Modifying the Record Stack to Work with $O(n + k^2 \log n)$
Bits}\label{subsect:recStack}

The space requirements of the record stack used by the  
iterator shown in Theorem~\ref{thm:spaceefficientiterator1}
 is $O(kn)$ bits.
Our goal in Sect.~\ref{sec:improvementsinspace} is to reduce it to $O(n)$ bits.
The bottleneck here is the record stack.
Assume that, for $\ell = O(\log n)$, a graph $G_\ell = (V_\ell, 
E_\ell)$ with 
$n_\ell$ vertices is on top of the 
record stack.
When considering the record $r_\ell$ on top
of the record stack we see that
most structures use $O(n_\ell)$ bits: a separator $S_\ell$, a vertex set
$X_\ell$ and a connected component finder $F_\ell$. The only structure that 
uses more space is the subgraph 
stack, which
uses $O(kn_\ell)$ bits.  This is due to the storage of the edge set
$E_\ell$ using $O(kn_\ell)$ bits.  The strategy we want to pursue is to
store only the vertices of the subgraphs (but not the edges) such that the space requirement of the
subgraph stack is $O(n)$ bits.  We call such a subgraph stack a {\em minimal
subgraph stack}. 
In the following we always assume that the
number of subgraphs on the minimal subgraph stack is
 $O(\log n)$ and that the subgraphs shrink by a constant factor.
This is in particular the case for the subgraphs generated by Reed's
algorithm.

In the following we make a distinction between {\em
complete} and {\em incomplete} vertices. Only complete vertices 
 have all their original edges, i.e., they have the
same degree in the original graph as they do in the subgraph. 
Note that the number of 
incomplete vertices in
each subgraph is $O(k)$, which follows directly from the separator size.  To
clarify, a vertex in the subgraph $G_{\ell}$ on top of the subgraph stack is
incomplete exactly if it is contained in a separator of the parent graph
$G_{\ell - 1}$.
%$G_{i}$ with $0 \leq i < \ell$ and if it is still contained in $G_\ell$.

\begin{lemma}\label{lem:adjacencyentyiteration}
Assume that we are given an undirected $n$-vertex graph $G=(V,E)$ with
treewidth $k$ and a toptuned minimal subgraph stack ($G_0=G,\ldots,G_\ell)$
with $\ell \in O(\log n)$.
Assume further that each graph $G_i$ ($1 \le i \le \ell$) has $n_i$ vertices,
$m_i$ edges and contains $O(k)$ incomplete vertices.
%We can iterate over the arcs of all vertices in $G_\ell$ in $O(km_\ell)$ time,
%which is a factor of $O(k)$ slower as iterating over the edges in a standard 
%graph.
The modified subgraph stack can be realized with $O(n+k^2 \log n)$ bits and allows us to 
push
an
$n_{\ell+1}$-vertex graph $G_{\ell+1}$ on top of a minimal subgraph stack
in $O(k^2n_{\ell} \log^*\ell)$ time. 
The resulting graph interface allows us 
to access the adjacency array
of the complete vertices  in constant time whereas an iteration
over the adjacency list of an incomplete vertex runs in $O(m_\ell)$ time.
\end{lemma}
\begin{proof}
	Recall that the subgraph stack considers every edge as a pair of directed arcs.  Let
	$\phi$ be the vertex translation between $G_i$ and  $G_0$.  Each
	complete vertex of $G_i=(V_i, E_i)$ has the same degree in both~$G_i$ and~$G_0$.
	Thus, to iterate over all arcs of a complete vertex $v \in
	V_i$, iterate over every arc $(\phi(v), u)$ of $\phi(v)$ and return the
	arc $(v, \phi^{-1}(u))$. 
%	Since $G_\ell$
%	contains $n_\ell-O(k)$ complete vertices and $O(km_\ell)$ arcs with
%	complete vertices at both endpoints, the iteration over the complete
%	vertices can be done in $O(km_\ell)$ time.
	For a complete vertex~$v$ we can use the adjacency array of~$v$.
	To iterate over the arcs
	of an incomplete vertex $v$ (via support of adjacency lists of~$v$),
	we differ two cases: (1) the arcs to a
	complete vertex and (2) the arcs to another incomplete vertex. 
	To iterate over all arcs of (1) we iterate over all complete vertices
	$u$ of $G_i$ and check in $G_0$ if $\phi(u)$ has an edge to
	$\phi(v)$.  If it does, $\{v, u\}$ is an edge of $v$. 
%	Since there are
%	at most $O(k)$ incomplete vertices, and iterating over all complete
%	vertices runs in $O(km_\ell)$ time,
	Thus, the iteration over all arcs incident to an incomplete vertex of (1)
	runs in $O(m_i)$ time.
	
	For the arcs according to (2), we use matrices $M_i$ storing the edges
	in between incomplete vertices. To build the matrices we proceed as follows.
	Whenever a new graph $G_{\ell+1}$ is pushed on the subgraph
	stack, we create a bit matrix $M_{\ell+1}$ of size $k^2$ and a
	rank-select data structure $I_{\ell+1}$ of size $n_{\ell+1}$ with $I_{\ell+1}[v]=1$
	exactly if $v$ is incomplete.  $M_{\ell+1}$ is used to store the
	information if $G_{\ell+1}$ contains an edge 
	$\{u', v'\}$
	between any two incomplete vertices $u'$ and
	$v'$ of $G_{\ell+1}$, which is the case exactly if
	$M_{\ell+1}[I_{\ell+1}.\op{rank}(u')][[I_{\ell+1}.\op{rank}(v')]]=1$ and
	$M_{\ell+1}[I_{\ell+1}.\op{rank}(v')][[I_{\ell+1}.\op{rank}(u')]]=1$.
%	\info{Red comment! Johannes: Habe hier die $\ell+1$ angepasst und $u'$ und $v'$ eingeführt}
	%
	First,
	we initialize $M_{\ell+1}$ to contain only $0$ for all bits.  Then
	we use $M_{\ell}$ to find edges between incomplete vertices of
	$G_{\ell}$ and set the respective bits in $M_{\ell+1}$ to $1$ if
	those incomplete vertices are still contained in $G_{\ell+1}$ (if
	$\ell=0$, we set all bits to $0$). Afterwards we are able to find
	edges between incomplete vertices that are already incomplete in the
	previous graph.
	 We need to update $M_{\ell+1}$ to contain the
	information of the edges between the vertices that are complete in
	$G_{\ell}$, but are not complete in $G_{\ell+1}$.  Since they are
	complete in $G_{\ell}$, we can simply iterate over all complete
	edges $e$ of $G_{\ell}$ in $O(km_\ell)$ time and check if
	both endpoints of $e$ are incomplete in $G_{\ell+1}$ via $I_{\ell+1}$.  If so, we set
	the respective bits in %$M$
	$M_{\ell + 1}$
	to $1$.
	
	{\bfseries Efficiency:}
	Queries on $M_i$ ($1 \le i \le \ell$) allow us to 
	iterate over all arcs of (2) of an incomplete vertex in~$O(k)$ time.
	This results in a
	combined runtime of $O(m_\ell + k) = O(m_\ell)$
	by ignoring zero-degree vertices in $M_\ell$.
 	Storing all bit matrices $M_i$ uses $O(k^2 \ell) = O(k^2 \log n)$ bits
 	and the space used by the rank-select structures is negligible.
 	%
	%We can now support a graph interface with $O(k)$ restricted vertices.
	%The adjacency arrays of the non-restricted vertices can be accessed in 
	%constant time whereas the iteration over the adjacency list of all
	%restricted vertices runs in $O(k^2 n)$ time. 
	The adjacency lists
	are realized by storing a pointer for each vertex. 
%	For constant time evaluation
% 	of the degree of the restricted vertices it is needed
% 	to iterate over all restricted
%	vertices once (in $O(km_\ell)$ time) to find the
%	respective degree and store it.
	This uses negligible
	additional $\Theta(k\log n)$ bits for implementing the interface.
	Our modified subgraph stack uses $O(n + k^2 \log n)$ bits.
\end{proof}

The last lemma allows us to store all recursive instances of Reed's algorithm 
with $O(n + k^2 \log n)$ bits. We use the result in the next section to show our
first $O(n)$-bit iterator to output a tree decomposition on graphs of small treewidth.

\section{Iterator for Tree Decomposition using $O(n)$ Bits for $k = O(n^{1/2-\epsilon})$
}\label{sec:improvementsinspace}

%To further reduce the space of the iterator introduced in Section~\ref{sec:stream} to $O(n)$ bits for 
%sufficiently small $k$, we make use of special properties
%of subgraphs created during Reed's algorithm. 
By combining the $O(kn)$-bit iterator of Theorem~\ref{thm:spaceefficientiterator1} with the modified record stack of 
Lemma~\ref{lem:adjacencyentyiteration} we can further reduce the space to $O(n)$ bits for a sufficiently
small $k$. 
Recall that the only structure using more than $O(n)$ bits
in the proof of Theorem~\ref{thm:spaceefficientiterator1} was the subgraph stack whose 
space is stated in Lemma~\ref{cor:recordstack}.
%, which we mitigate with 
%Lemma~\ref{lem:adjacencyentyiteration}.
This allows us to show the following theorem.
 
\begin{theorem}\label{the:streaminonlogk}
	Given an undirected $n$-vertex
	graph $G$ with treewidth $k$,
	there exists an iterator that outputs a balanced and binary
tree de\-com\-po\-si\-tion $(T,B)$ of width $8k+6$ in Euler-traversal order 
using $O(n + k^2 \log^2 n)$ bits and 
$c^{k} n \log n \log^*n$ time for some constant~$c$.
For $k=O(n^{1/2-\epsilon})$
	with an arbitrary $\epsilon>0$,
our space consumption is $O(n)$ bits.
\end{theorem}
\begin{proof}
Recall that the tree decomposition iterator of Theorem~\ref{thm:spaceefficientiterator1}
uses the
algorithm of Theorem~\ref{the:nbitvertexdisjoint} to find 
$k$ vertex-disjoint paths for the construction of the separators. 
%Recall that we call vertices of
%a graph $G_i$ stored in a minimal subgraph 
%stack (Lemma~\ref{lem:adjacencyentyiteration})
%incomplete if we can 
%not access its edges in constant time.

With the construction of the $k$ vertex-disjoint paths
we have to compute and access path storage schemes.
To avoid querying the neighborhood of an incomplete vertex~$v$
several times by querying the region
with~$v$ several times, 
%
%
%By storing the stack predecessor and stack successor for the $O(k)$ slow
%vertices, we can avoid querying edges of slow vertices several times so that
%the DFS runs in $O(n(\log^* n +k))$ time and $O(n+k\log n)$ bits as stated in
%Lemma~\ref{lem:spaceDFS}---the access
%restricted vertices are exactly the slow vertices.
%In other words, we {\em process} every slow vertex once, i.e.,
%we iterate once over the edges of every slow vertex.
%
we add the~$O(k)$ incomplete vertices to the
boundary vertices when we build a path data structure.  
This neither increases our asymptotical time
nor our total space bounds.
Moreover, by adding the $O(k)$ incomplete vertices~$v$ 
to the boundary vertices
and storing $\op{prev}(v)$/$\op{next}(v)$, 
we avoid running a BFS on incomplete vertices
and searching for the predecessor and successor of~$v$.
If we afterwards construct a path,
 we iterate over the adjacency lists of incomplete vertices only once---even if regions are constructed several times. 

The construction of a vertex-disjoint path involves $O(\log k)$ reroutings. 
Recall that each rerouting
extends a clean area by traversing over a path $P^*$ (two DFS) and then searches
``backwards'' within the area with two further DFS.
Altogether, a rerouting
can be done with $O(1)$ DFS runs and $O(\log k)$ DFS runs over all reroutings
for one path $P^*$.
%
%To sum up,
%the iteration over the edges of the~$O(k)$ 
%incomplete vertices is part of 
%a graph exploration, for which the algorithm
%of Theorem~\ref{thm:spaceefficientiterator1}
%pays $O(k m_i + n_i \log^* n_i)$.
%
%By taking the incomplete vertices into the boundary
%the running time of the $\op{prev}$/$\op{next}$ operation~(Lemma~\ref{lem:recompute}) changes.
%For a boundary vertex $v$ the operation iterates over $v$'s neighborhood
%which takes $O(m_i)$ time for an incomplete vertex 
%and $O(\op{deg}(v))$ time for a complete vertex.
%For vertices outside of the boundary the running time remains the same, 
%because for them we do not iterate over the neighborhood of boundary vertices
%when exploring the region of~$v$.
%Recall that $\op{color}$ is a constant time operation for boundary vertices.

%When constructing a path we iterate over the neighbors of each vertex once.
%Moreover, when rerouting the paths we avoid revisiting already 
%visited vertices by putting them into a clean area.
%We pay for the construction of a path or the rerouting
%$\Omega(m_i)$ in the runtime. Thus, the iteration of the neighborhood
%of $O(k)$ incomplete vertices can increase the runtime only 
%by a factor~$k$.
%Thus we can safely add another factor $k$ to our total running
%time to pay for the total running time increased by $\op{prev}$ and $\op{next}$.
%To sum up we run $O(k \log k)$ DFS to construct $O(k)$
%paths with the runtime of $O((k \log k)(km_i + n_i \log^* n_i))$
%
To construct a balanced $X$-separator a set $R$ is constructed by
another DFS run. However, this iterates over the neighbors
of each vertex only once.
To sum up a balanced $X$-separator can be constructed
by $O(\tilde{c}^k)$ DFS runs for some constant~$\tilde{c}$.
%Taking the extra factors of the DFS
%into account we can construct $k$ paths, as stated in 
%Theorem~\ref{the:nbitvertexdisjoint}, in 
%$O(k n(k + \log^* n) k^4 \log^3 k)
%= O(n (k \log k + \log^*n) k^5 \log^3 k)
%= O(\mathrm{poly}(k) n \log^*n)$ time.

Given an $n_i$-vertex, $m_i$-edge graph $G_i$ 
stored in a minimal subgraph stack, 
we can run the space-efficient DFS
of Lemma~\ref{lem:spaceDFS} in
$O(k m_i + n_i \log^* n_i) = O(k^2 n_i \log^* n_i)$ time
and $O(n + k \log n)$ bits, which is a factor of~$O(k)$
slower than the DFS of Theorem~\ref{th:dfs}.
Thus, a balanced $X$-separator can be constructed in
$O(\bar{c}^k n_i \log^* n_i)$ time for some constant~$\bar{c}$,
which is the same time as stated in 
%Compared to the space-efficient DFS of Thereom~\ref{th:dfs}
%the running time has an extra factor of $O(k)$.
the proof of Theorem~\ref{thm:spaceefficientiterator1}
(the constant $\bar{c}$ here is larger).
Concerning the space consumption
note that we have a record-stack of size $O(n + k \log^2 n)$ bits
by Lemma~\ref{cor:recordstack} by using our minimal subgraph stack
of $O(n + k^2 \log n)$ bits. The $O(k)$ extra values for
$\op{prev}$ and $\op{next}$ are negligible.
Finally note that we can search a separator 
with $O(n + k^2 (\log k) \log n)$ bits by 
Corollary~\ref{cor:speffVertexSeparator}.

%
%
%
%As described in the proof of
%Lemma~\ref{lem:newPath}, each path is 
%constructed in such a way that colored vertices of large degree
%do not query the storage scheme for their neighbors on the paths. The same is true for the 
%construction of the $u$-$w$ path (Fig.~\ref{fig:monotone}c) in the proof of Lemma~\ref{lem:rerouting}.
%Therefore, the term ${\mathrm{deg}}(v)$ in the runtime of
%Lemma~\ref{lem:recompute} %{lem:storageScheme}
%can be ignored.
%
%Since 
%the DFS has to access the graph through the graph interface of the minimal subgraph stack (Lemma~\ref{lem:adjacencyentyiteration})
%and our
%storage scheme (Lemma~\ref{lem:recompute} %{lem:storageScheme}
%), we get an additional factor of $O(k)$ and $O(k^3 
%\log^3 k)$,
%respectively, in the runtime. Therefore, we can construct the $\ell$-disjoint paths in 
%$O( \ell k (k^3 \log^3 k) n(k + \log^*n) )$ time.
%%
%Note that it suffices to compute the set $R$ on the graph induced by the
%complete vertices and the runtime for the computation of $R$ is 
%$O(kn)$.
%
%For constants $c$ and $d$, this results in a
%runtime of $d^k  (\mathrm{poly}(k) n\log^*n) = c^kn\log^*n$ when
%searching for a balanced $X$-separator for $G$ (Lemma~\ref{lem:sep}).
%We know that the runtime of calculating the entire balanced tree 
%decomposition for $G$ is
%based on instances consisting of $O(n)$ total vertices in each recursion level
%and that the recursion depth is $O(\log n)$.
%Therefore, our total runtime is $c^{k} n \log n \log^*n$.
\end{proof}

We next combine the theorem above with a recent tree-decomposition 
algorithm by Bodlaender et al.~\cite{BodDDFLP16} as follows.
The algorithm by Bodlaender et al.\ finds a tree decomposition for
a given $n$-vertex graph $G$ of treewidth $k$ in $b^k n$
time for some constant~$b$~\cite{BodDDFLP16}. 
The resulting tree 
decomposition has a width of
$5k+4$. The general strategy pursued by them is to first
compute a tree decomposition of large width
and then use dynamic programming on that tree decomposition to obtain the 
final tree decomposition
of width $5k+4$. For an overview of the construction, we refer to
\cite[p.~3]{BodDDFLP16}.
The final tree decomposition is balanced due to the fact that
its construction uses balanced $X$-separators at every second level%
, 
alternating between an
$8/9$-balanced and an unbalanced 
$X$-separator~\cite[p.~26]{BodDDFLP16}.
Further details of the construction of different kinds of the final tree 
decomposition 
can be found in~\cite[p.~20, and p.~39]{BodDDFLP16}.
Since
its runtime is bounded by $b^k n$,
it can write at most~$b^kn$ words and thus has a space requirement of 
at most~$b^kn\log n$ bits.

Our following idea is to use a hybrid approach to improve the runtime %
of our iterator.
We %
first run our iterator 
(Theorem~\ref{the:streaminonlogk}). Once the 
height of the record-stack of our tree-decomposition iterator is equal to 
$z = b^k \log \log n$, the call of $\op{next}()$ uses an unbalanced 
$X$-separator $S^*$.
This ensures that the size of the bag is at most $4k+2$ instead of $8k+6$.
(We later add all vertices in the bag to all following bags.)  Note that
using a single unbalanced $X$-separator $S^*$ on all root-to-leaf paths
of our computed tree decomposition
increases the height of the tree decomposition only by one.
A following
call of $\op{next}()$ toptunes the graph $G_\ell$ and then uses Bodlaender et al.'s  
linear-time tree-decomposition algorithm~\cite{BodDDFLP16} to calculate a
tree decomposition $(T',B')$ of an $n_\ell$ vertex subgraph $G_\ell$, which 
we then turn 
by folklore techniques 
into a binary
tree decomposition $(T'',B'')$ by neither increasing the asymptotic size nor the width of the 
tree 
decomposition. In detail, this is done by
repeatedly replacing all nodes $w$ with more than two children by a node $w_0$ 
with two children
$w_1$ and $w_2$, with $B(w_0)=B(w_1)=B(w_2)=B(w)$, followed by adding the
original children of $w$ to $w_1$ and $w_2$, alternating between them both.
To ensure property $(TD2)$ of a tree 
decomposition, we add the vertices in $S^*$ to all bags of
$(T'',B'')$. 
We so get a tree decomposition of the width 
$(5k+4)+(4k+2)=9k+6$.

 Since $G_\ell$ contains $n_\ell = O(n/2^{z}) = O(n/(b^k \log n))$ 
 vertices, the 
 space usage of the linear-time
tree-decomposition algorithm is $b^k n_\ell \log n=O(n)$ bits.  The runtime of
the algorithm is bounded by~$b^k n_\ell$.  
Once we obtain $(T',B')$, we also need to transform each bag $b'$ of 
$B'$
since $B'$ contains mappings in relation to $G'$, but we want them to
contain mappings in relation to $G$.  This can be done in negligible time
since $G'$ was toptuned before.  We then initialize a tree-decomposition
iterator $I$ for $(T',B')$ as described in the beginning of
Subsection~\ref{sec:stream}.  Now, as long as $I'$ has not finished its
traversal of $(T',B')$, a call to $\op{next}$ on $I$ is equal to a call to
$\op{next}$ on $I'$.  Similarly, a call to $\op{show}$ on $I$ now returns
the tuple $(B'(w), d_w)$ with $d_w$ being the depth of $w$ in $T'$ plus the
size of the record stack of $I$.  Once iterator $I'$ is finished, we throw
away $(T',B')$.  Then, the operations of $\op{next}$ and $\op{show}$ work
normally on $I$ until the size of the record stack again is $O(b^k \log \log
n)$ or until the iteration is finished.
Since we use our iterator
only to recursion depth~$z$, our 
algorithm runs in $a^kn(\log^*n)z$ time for some constant~$a$.
The total runtime is 
$a^kn(\log^*n)(b^k\log\log n) + b^kn = c^k n \log \log n\log^*n$ for 
some constant $c$.

\begin{corollary}\label{cor:spaceefficientiterator4}
	There is an iterator to output
a balanced binary tree-decomposition $(T,B)$ of width
$9k+6$ for an  $n$-vertex $G=(V,E)$ with treewidth $k$ in
Euler-traversal order in 
$c^k n\log \log n\log^*n$ time 
for some constant $c$
using $O(n + k^2 \log^2 n)$ bits.
For $k = O(n^{1/2-\epsilon})$ and an arbitrary $\epsilon>0$, the 
space consumption is $O(n)$ bits.
\end{corollary}

If we try to run our iterator 
from the last corrolary on a graph that 
has a treewidth greater than~$k$,
then either the computation of a vertex separator or the
computation of Bodleander et al.'s algorithm for finding a tree decomposition~\cite{BodDDFLP16}
fails. In
both cases, our iterator 
stops and we can output that the treewidth 
of
$G$ is larger than $k$.%

%     \blue{Using the iterator from the last corollary together with the dynamic programming strategy outlined in \cite{10.1007/978-3-319-03898-8_5} we are able to show our final theorem. 
%     To stay within our space bound of $O(n)$ bits we only store a small amount of the tables used for dynamic programming and recompute them once they are needed again. In the proof we introduce a problem dependent constant a such that the dynamic programming tables have size $O(a^{k + 1} n)$ which affects the number of tables we can store concurrently. Roughly speaking, our approach applies to all MSO-problems whose table sizes are single exponential in $k$.
%     
%     \begin{theorem}\label{th:mdsspace}
%     Let $G$ be an n-vertex graph with treewidth $k \le c'$ log n for some constant $0 < c' < 1$. 
%     Using $O(n)$ bits and $c^k n \log n \log^* n$ time for some constant $c$ we can solve the following problems: \textsc{Vertex Cover}, 
%     \textsc{Independent Set}, 
%     \textsc{Dominating Set}, 
%     \textsc{MaxCut} and 
%     $q$-\textsc{Coloring}.	
%     \end{theorem}}\info{Only in ISAAC, i moved this theorem to the end of 8 where we actually prove this. this can be removed}
%     
%     \newcommand{\mdsSpace}{%
%     }

\section{Applications}\label{sect:appl}
From~\cite[Theorem~7.9]{CygFKLMPPS15} we know that there is an algorithm
	that solves all problems
mentioned in Theorem~\ref{th:mdsspace} on an
$n$-vertex graph with
treewidth~$k$ in $c^kn$ time for some constant $c$ when a tree decomposition
with approximation-ratio $O(1)$ is given. The general strategy used for solving
these problems is almost identical. First, traverse the tree decomposition
bottom-up 
and compute a table for 
each node $w$. 
The table stores the size of all best possible solutions in the graph induced by all
bags belonging to nodes below $w$ under certain conditions for the
vertices in bag $B(w)$. 
E.g., for \textsc{Vertex
	Cover} the table contains $2^{k+1}$ solutions ($v \in B(w)$ does belong or does not belong
to the solution) and for \textsc{Dominating Set} it contains
$3^{k+1}$ solutions (one additionally differs, if a vertex is already dominated
or not). We only consider problems whose table has 
at most $c^k$ solutions for some constant~$c$. For each possible solution, the table stores the size 
of the solution and thus uses
$O(c^k \log n )$ bits. After the bottom-up traversal, the 
 minimal/maximal solution size in the table at the root is the solution for the minimization/maximization
problem, respectively.
An optimal solution set can 
be obtained in a top-down traversal by using the tables.

It is clear that, for large $k$,  we can not  
store all tables when
trying to use $O(n)$ bits.
Our strategy is to 
store the tables only for the nodes on 
a single root-leaf path of the
tree decomposition and for nodes with a depth less than some threshold value. 
The other tables are recomputed during the construction of the solution set.
For a balanced tree decomposition this results in $O(c^k
\log^2n)$ bits storing $O(\log n)$ tables for vertices 
on a root-leaf path each of size $O(c^k \log n)$ bits. 
Using this strategy we have all information to use the standard bottom-up
traversal to compute the
size of the solution for the given problem for $G$. To obtain an optimal solution set we need 
a balanced and binary tree decomposition with a constant-factor
approximation of the treewidth.
		
As an example we first give an algorithm for vertex cover that can be easy generalized
to other problems. We conclude this section
by giving a list of problems that can be solved with the same asymptotic time 
and space bound.

A vertex cover of a graph~$G = (V, E)$ is a set of vertices $C \subseteq V$ 
such that, for each edge $\{u, v\} \in E$, $u \in C \vee v \in C$ holds.
For graphs 
with a small treewidth, one can find a minimum vertex cover by first computing 
a tree decomposition of the graph and then, using dynamic programming, 
calculate a minimal vertex cover. We start to sketch the standard approach.

Let $(T,B)$ be a tree decomposition of width $O(k)$ 
of an undirected graph $G$ with treewidth~$k$. Now, iterate over $T$ in 
Euler-traversal order and, if a node $w$ is visited for the first time, 
calculate and store in a table $\mathcal{T}_w$ all possible solutions of the 
vertex cover problem for~
$G[B(w)]$. Also store the value of each solution, which is equal to the number 
of vertices used for the cover. If the solution is not valid, store $\infty$ 
instead.

When visiting a node $w$ and  $\mathcal{T}_w$ already exists, we update 
$\mathcal{T}_{w}$ by using $\mathcal{T}_{w'}$ with $w'$ being the node visited 
during the Euler-traversal right before $w$ ($w'$ is a child of $w$). The 
update process is done by comparing each solution $s$ in $\mathcal{T}_w$ with 
each overlapping solution in $\mathcal{T}_{w'}$. A solution $s' \in 
\mathcal{T}_{w'}$ is chosen if it has the smallest value among overlapping 
solution. The value of $s'$ is added to the value of $s$, and the two solutions 
are linked with a pointer structure. Two solutions $s$ and $s'$ are overlapping 
exactly if, for each $v \in B(w) \cap B(w')$, $(v \in s \wedge v \in s') \vee 
(v \notin s \wedge v \notin s')$ is true. Once the Euler-traversal is finished, the 
table $\mathcal{T}_{r}$, with $r$ being the root of $T$, contains the size of the 
minimum vertex cover $C$ of $G$ as the smallest value of all solutions. This is 
the first step of the algorithm.

The second step is obtaining $C$, which is done by traversing top-down through 
all tables with the help of the pointer structures, starting at the solution 
with the smallest value in $\mathcal{T}_r$, and adding the vertices used by the 
solutions to the initially empty set $C$ if they are not yet contained in $C$. 
%Identifying the tables with nodes and pointers between the
%tables with edges,
%we obtain a tree (with multiple edges).
For simplicity, we first only focus on \textsc{Vertex Cover},
but use a problem specific 
constant $\lambda$
 for the size 
of the tables and runtime of solving subproblems.
For \textsc{Vertex Cover} $\lambda = 2$. 
This allows us an easy generalization
step to other problems subsequently.
In the following 
the default base of $\log$ is $2$.

\begin{lemma}\label{cor:vertexcover1}
	Given an $n$-vertex graph $G$ with treewidth $k \leq \log_\lambda n - 2 \log_\lambda \log n$ 
	we 
	can calculate the size of the optimal \textsc{Vertex Cover} $C$ of $G$ in 
	$c^k n \log 
	\log n \log^*n$ time using $O(n)$ bits for some constant $c$.
\end{lemma}
\begin{proof}
For an $n$-vertex graph $G$ with treewidth $k$ and a given tree decomposition 
$(T,B)$ of width $k'=O(k)$ the runtime of the algorithm is $O(2^{k'}n)$. A table 
$\mathcal{T}_w$ constructed for a bag $B(w)$ consists of a bit array of size 
$O(k')$ for each of the $2^{k'+1}$ possible solutions, and their respective values 
and pointer structures. This uses $O(2^{k'+1} (k + \log n))= O(\lambda^k\log n)$
bits per table 
for some problem specific constant~$a$,
which is two for \textsc{Vertex Cover}.
 Thus, storing the tables for the 
entire tree decomposition uses $O(\lambda^{k}n \log n )$ bits. Our goal is to 
obtain the optimal vertex cover using only $O(n)$ bits for both the tree 
decomposition $(T,B)$ and the storage of the tables. For obtaining only the 
size 
of $C$, i.e., the first step of the algorithm, we only need to store the tables 
for 
the current root-node path of the tree decomposition iterator. The reason is
that 
once a table has been used to update its parent table it is only needed for 
later 
obtaining the final cover via the pointer structures. We can iterate
over a balanced 
binary tree decomposition of width $O(k)$ in $\tilde{c}^k n \log \log n 
\log^*n$ time
using $O(n)$ bits (Theorem~\ref{cor:spaceefficientiterator4}) for some constant $\tilde{c}$. 
To obtain the bag-induced subgraphs we use Lemma~\ref{lem:baginducedgraph}. 
We have to 
store $O(\log n)$ tables, which results in $O(\lambda^k \log^2 n)$ bits used, which 
for 
$k \leq \log_\lambda n - 2 \log_\lambda \log n$ equals $O(n)$ bits ($\lambda^k \log^2 n 
\leq n \Rightarrow 
\lambda^k \leq n \log^{-2} n \Rightarrow k \leq \log_\lambda n - 2 \log_\lambda \log n$). 
Initializing and 
updating all tables can be done in $O(\lambda^k n)$ time. 
\end{proof}

To obtain the final vertex cover we need access to all tables and bags they
have been initially created for. 
%Since we want to only use $O(n)$ bits, we are not able to store all of them. 
We now use the 
previous lemma with modifications.
Our idea is to fix some $\ell \in \Nat$ 
and to use partial tree decompositions of depths~$\ell$
rooted at every $\ell$th node of a root-leaf path.
%In addition to 
%storing the tables of the current root-node path we store all tables 
%$\mathcal{T}_w$ with $w$ having a depth $d_w < \ell$ ($\ell$ is specified 
%later), which we call the {\em upper tree}, and each subtree with a root at 
%depths $\ell$, a {\em lower tree}.
%
For this, let us define $\op{ptd}_{G,\ell}(w)$ to be the 
partial tree decomposition~$(T', B)$
where~$T'$ is a subtree of~$T$ with root~$w$ and depths~$\ell$.
%We refer to this as the operation $\op{vc}(G,\ell)$, 
%which calculates and stores all tables of the upper tree of $G$. 
%This allows us to show the following lemma.

\begin{lemma}\label{cor:vertexcover2}
Given an $n$-vertex graph $G$ with treewidth $k \leq \log_\lambda n - 3 \log_\lambda 
	\log n=c'\log n$ for some constant $0 < c' <1$
we can calculate the optimal vertex cover $C$ of $G$ 
	in $c^k n \log n \log^* n$ 
time using $O(n)$ bits for some constant $c$.
\end{lemma}
\begin{proof}
By iterating over $(T,B)$ we first compute the tables
and pointer structures for $(T', B) = \op{ptd}_{G,\ell}(r)$ 
where $r$ is the root
of the tree decomposition of~$G$ and where $\ell = \log \log n$.
Thus, the partial tree $T'$ consists of $O(\log n)$ nodes, each with a table of $O(\lambda^k \log n)$ bits. (Recall that for \textsc{Vertex Cover} $\lambda = 2$.)

We then start to follow the pointer structures starting from~$r$.
%Anytime we output a new 
%bag $B(w)$ we move to $\mathcal{T}_{w}$ via the pointer structure 
%and process the vertices of the best solution in that table, together with the graph $G[B(w)]$. 
When we arrive at a node $w$ having a table where the 
pointer structure is 
invalid (because the next table does not exist) we build
$\op{ptd}_{G,\ell}(w)$.
Afterwards, we can continue 
to follow the pointer structures since the next tables now exist. 
We repeat to build the partial tree decompositions 
anytime we try to follow a pointer that is invalid until we arrive at a 
leaf (at which point we backtrack).
When we have built and processed all tables of a 
subtree with root~$w$, we can throw away all tables 
of $\op{ptd}_{G,\ell}(w)$.
%Since those tables have been processed by that point.
In other words, we have to store tables for only $O(\log n / \ell)$ partial
trees. 
Note that each partial tree uses $O(\log n)$ tables
of $O(\lambda^k \log^2 n)$ bits in total.
Summed over all partial trees on a root-leaf path,
we need to store $O(\lambda^k \log^3 n)$ bits. 
For $k \leq \log_\lambda n - 3 \log_\lambda \log n$ this uses 
$O(n)$ bits ($\lambda^k 
\frac{\log^3 n}{\log \log n}  \leq n \Rightarrow \lambda^k \leq n \frac{\log \log n}{\log^3 n} \Rightarrow k \leq \log_\lambda n - 3 \log_\lambda \log n + \log_\lambda \log \log n$).
It remains to show the impact on the runtime. 
Anytime we want to obtain the tables 
of $\op{ptd}_{G, \ell}(w)$, we need to compute the tables
of the subtree rooted by $w$. 
%of all its partial trees, recursively. 
Thus partial subtrees $\op{ptd}_{G, \ell}(w)$ with deepest nodes~$w$ 
%tables of the subtrees at the bottom 
need to be calculated $O(\log n / \log \log n)$-times, the 
partial subtrees above 
$(O(\log n / \log \log n) - 1)$-times and so forth. This can be thought of as 
iterating over
the tree decomposition $(T,B)$ of $G$ for $O(\log n / \log \log n)$ times.
In other words, we run the algorithm of Theorem~\ref{cor:vertexcover1}, 
$O(\log n / \log \log n)$ times.
\end{proof}

We finally present our last theorem.

     \begin{theorem}\label{th:mdsspace}
     	Let $G$ be an n-vertex graph with treewidth $k \leq c' \log n$ for 
     	some 
     	constant $0 < c' <1$. 
     	Using $O(n)$ bits and $c^k n \log n \log^*n$ time 
     	for some constant $c$ we can solve the following problems:
     	%	\begin{itemize}
     	%\item 
     	\textsc{Vertex Cover}, %and 
     	\textsc{Independent Set}, %in time \red{$2^kk^{O(1)}n$},
     	%\item 
     	\textsc{Dominating Set}, % in time \red{$4^kk^{O(1)}n$},
     	%\item
     	\textsc{MaxCut} %in time \red{$2^kk^{O(1)}n$},
     	%\item 
     	and $q$-\textsc{Coloring}. % in time \red{$q^kk^{O(1)}n$}
     	%	\end{itemize}
     	%	using \red{$O(k^{O(1)}n)$} bits of working memory.
     \end{theorem}

\begin{proof}
As in the proof of Lemma~\ref{cor:vertexcover1} $k' = O(k)$ is the width of the tree decomposition.
To change the previous two proofs to the different problems mentioned in
Theorem~\ref{th:mdsspace} the only change is in the
computation of the tables and the size of the tables.
Thus, we have to take another 
value for the problem specific constant~$\lambda$
of the Lemma~\ref{cor:vertexcover1} and Lemma~\ref{cor:vertexcover2}. 
For \textsc{Vertex Cover},
\textsc{Independent Set} and \textsc{Max Cut} the tables contain $2^{k'+1}$
possible solutions and thus $\lambda=2$.  For \textsc{Dominating Set} they contain
$3^{k'+1}$ and $q$-\textsc{Coloring} they contain $q^{k'+1}$ possible
solutions, so $\lambda=3$ and $\lambda=q$, respectively.
For further details
see~\cite{10.1007/978-3-319-03898-8_5}.
\end{proof}

\phantomsection
\addcontentsline{toc}{chapter}{Bibliography}
\bibliography{main}

\end{document}